\documentclass[10pt]{article}
\usepackage[margin=1in]{geometry}

\usepackage{minibox}
\usepackage{tikz}
\usetikzlibrary{shapes.misc,shapes.geometric,arrows,positioning,calc,backgrounds}

\usepackage{amsthm}
\usepackage{amsmath}
\usepackage{amsfonts,amssymb}
\usepackage{graphicx}
\usepackage{longtable}
\usepackage{multirow} 
\usepackage{algorithmicx}
\usepackage{algorithm}

\usepackage{thmtools}
\usepackage{thm-restate}

\usepackage{comment}
\usepackage{epstopdf} 
\usepackage{bbold} 

\usepackage{chngcntr}
\usepackage{apptools}

\usepackage{tabularx}
\newcolumntype{D}{>{\centering\arraybackslash}X}
\usepackage{hyperref}
\hypersetup{
    bookmarksnumbered=true, 
    unicode=false, 
    pdfstartview={FitH}, 
    pdftitle={}, 
    pdfauthor={}, 
    pdfsubject={}, 
    pdfcreator={}, 
    pdfproducer={}, 
    pdfkeywords={}, 
    pdfnewwindow=true, 
    colorlinks=true, 
    linkcolor=blue, 
    citecolor=blue, 
    filecolor=blue, 
    urlcolor=blue 
}

\newcommand{\ket}[1]{|#1\rangle}
\newcommand{\bra}[1]{\langle#1|}

\newtheorem{lemma}{Lemma}
\newtheorem{corollary}{Corollary}
\newtheorem{problem}{Problem}
\newtheorem{claim}{Claim}

\declaretheorem[name=Theorem]{theorem}
\declaretheorem[name=Definition]{definition}

\begin{document}

\title{Hamiltonian Simulation by Uniform Spectral Amplification}

\author{
\normalsize Guang Hao Low\thanks{Department of Physics, Massachusetts Institute of Technology  \texttt{\{glow@mit.edu\}}},\quad 
\normalsize Isaac L. Chuang\thanks{Department of Electrical Engineering and Computer Science, Department of Physics, Research Laboratory of Electronics, Massachusetts Institute of Technology  \texttt{\{ichuang@mit.edu\}}}
\date{\today}
}
\maketitle

\begin{abstract}
The exponential speedups promised by Hamiltonian simulation on a quantum computer depends crucially on structure in both the Hamiltonian $\hat{H}$, and the quantum circuit $\hat{U}$ that encodes its description. In the quest to better approximate time-evolution $e^{-i\hat{H}t}$ with error $\epsilon$, we motivate a systematic approach to understanding and exploiting structure, in a setting where Hamiltonians are encoded as measurement operators of unitary circuits $\hat{U}$ for generalized measurement. This allows us to define a \emph{uniform spectral amplification} problem on this framework for expanding the spectrum of encoded Hamiltonian with exponentially small distortion. We present general solutions to uniform spectral amplification in a hierarchy where factoring $\hat{U}$ into $n=1,2,3$ unitary oracles represents increasing structural knowledge of the encoding. Combined with structural knowledge of the Hamiltonian, specializing these results allow us simulate time-evolution by $d$-sparse Hamiltonians using $\mathcal{O}\left(t(d \|\hat H\|_{\text{max}}\|\hat H\|_{1})^{1/2}\log{(t\|\hat{H}\|/\epsilon)}\right)$ queries, where $\|\hat H\|\le \|\hat H\|_1\le d\|\hat H\|_{\text{max}}$. Up to logarithmic factors, this is a polynomial improvement upon prior art using $\mathcal{O}\left(td\|\hat H\|_{\text{max}}+\frac{\log{(1/\epsilon)}}{\log\log{(1/\epsilon)}}\right)$ or $\mathcal{O}(t^{3/2}(d \|\hat H\|_{\text{max}}\|\hat H\|_{1}\|\hat H\|/\epsilon)^{1/2})$ queries.  In the process, we also prove a matching lower bound of $\Omega(t(d\|\hat H\|_{\text{max}}\|\hat H\|_{1})^{1/2})$ queries, present a distortion-free generalization of spectral gap amplification, and an amplitude amplification algorithm that performs multiplication on unknown state amplitudes.
\end{abstract} 

\tableofcontents
\section{Introduction}
\label{Sec:Introduction}
Quantum algorithms for matrix operations on quantum computers are one of its most exciting applications. In the best cases, they promise exponential speedups over classical approaches for problems such as matrix inversion~\cite{Harrow2009} and Hamiltonian simulation, which is matrix exponentiation. Intuitively, any arbitrary unitary matrix applied to an $q$-qubit quantum state is `exponentially fast' due to a state space of dimension $n=2^q$. However, if these matrix elements are presented as a classical list of $\mathcal{O}(n^2)$ numbers, simply encoding the data into a quantum circuit already takes exponential time. Thus the extent of this speedup is sensitive to both the properties of the Hamiltonian and the input model defining how that information is made accessible to a quantum computer. 

Broad classes of Hamiltonians $\hat{H}$, structured so as to enable this exponential speedup, are well-known. The most-studied examples include local Hamiltonians~\cite{Lloyd1996universal} built from a sum of terms each acting on a constant number of qubits, and its generalization as $d$-sparse matrices~\cite{Aharonov2003Adiabatic} with at most $d$ non-zero entries in every row, whose values and positions must all be efficiently computable. 
More recent innovations consider matrices that are a linear combination of unitaries~\cite{Childs2012,Berry2015Truncated,Novo2016improved} or density matrices~\cite{LloydMohseniRebentrost2014,Kimmel2017}. Though different classes define different input models, that is unitary quantum oracles that encode $\hat{H}$, it is still helpful to quantify the cost of various quantum matrix algorithms through the query complexity, which in turn depends on various structural descriptors of $\hat{H}$, such as, but not limited to, its spectral norm $\|\hat{H}\|$, induced $1$-norm $\|\hat{H}\|_{1}$, max-norm $\|\hat{H}\|_{\text{max}}$, rank, or sparsity. 

A challenging open problem is how knowledge of any structure may be maximally exploited to accelerate quantum algorithms. As the time-evolution operator $e^{-i\hat{H}t}$ underlies numerous such quantum algorithms, one common benchmark is the Hamiltonian simulation problem of converting this description of $\hat{H}$ into a quantum circuit that approximates $e^{-i\hat{H}t}$ for time $t$ with some error $\epsilon$. To illustrate, we recently provided an algorithm with optimal query complexity $\mathcal{O}\big(td\|\hat{H}\|_{\text{max}}+\frac{\log{(1/\epsilon)}}{\log\log{(1/\epsilon)}}\big)$~\cite{Low2016HamSim} in all parameters  for sparse matrices~\cite{Childs2010,Berry2014,Berry2015Hamiltonian}, based on \emph{quantum signal processing} techniques~\cite{Low2016methodology}. Though this settles the worst-case situation where only $d$ and the max-norm $ \|\hat{H}\|_{\text{max}}$ are known in advance, there exist algorithms that exploit additional knowledge of the spectral norm $\|\hat{H}\|$ and induced-one norm $\|\hat{H}\|_{1}$ to achieve simulation with $\mathcal{O}(t^{3/2}(d\|\hat{H}\|_{\text{max}}\|\hat{H}\|_{1}\|\hat{H}\|\frac{1}{\epsilon})^{1/2})$~\cite{Berry2012} queries. Though this square-root scaling in sparsity alone is optimal, it is currently unknown whether the significant penalty is paid in the time and error scaling is unavoidable. Motivated by the inequalities $\|\hat{H}\|_{\text{max}}\le \|\hat{H}\|\le \|\hat{H}\|_{1}\le d \|\hat{H}\|_{\text{max}}$~\cite{Childs2010Limitation}, one could hope for a best-case algorithm in Claim~\ref{Claim:Sparse_Ham_Sim} that interpolates between these possibilities.
\begin{claim}[Sparse Hamiltonian simulation]
\label{Claim:Sparse_Ham_Sim}
Given the standard quantum oracles that return values of $d$-sparse matrix elements of the Hamiltonian $\hat{H}$, there exists a quantum circuit that approximates time-evolution $e^{-i\hat{H}t}$ with error $\epsilon$ using $Q=\mathcal{O}\big(t(d\|\hat{H}\|_{\text{max}}\|\hat{H}\|_1)^{1/2}+\frac{\log{(1/\epsilon)}}{\log\log{(1/\epsilon)}}\big)$ queries and $\mathcal{O}(Q\log{(n)})$ single and two-qubit quantum gates.
\end{claim}

The challenge is exacerbated by how unitary time-evolution, though a natural consequence of Schr{\"o}dinger's equation in \emph{continuous}-time, is not natural to the gate model of \emph{discrete}-time quantum computation. In some cases, such as quantum matrix inversion~\cite{Kothari2014efficient}, algorithms that are more efficient as well as considerably simpler in both execution and concept can be obtained by creatively bypassing Hamiltonian simulation as an intermediate step. The need to disentangle the problem of exploiting structure from that of finding best simulation algorithms is highlighted by celebrated Hamiltonian simulation techniques such Lie-Product formulas~\cite{Lloyd1996universal}, quantum walks~\cite{Childs2010}, and truncated-Taylor series~\cite{Berry2015Truncated}, each radically different and specialized to some class of structured matrices.

A unifying approach to exploiting the structure of Hamiltonians, independent of any specific quantum algorithm, is hinted at by recent results on Hamiltonian simulation by \emph{qubitization}~\cite{Low2016hamiltonian}. There, we focus on a \emph{standard-form} encoding of matrices (Def.~\ref{Def:Standard_Form}), which, in addition to generalizing a number of prior input models, also appears more natural. On measurement outcome $\ket{0}_a$ with best-case success probability $(\|\hat{H}\|/\alpha)^2\le 1$, a Hermitian measurement operator $\hat{H}/\alpha$ is applied on the system -- thus the standard-form is no more or less than the fundamental steps of generalized measurement~\cite{Nielsen2004}.
Treating this quantum circuit as a unitary oracle, this amounts possessing no structural information whatsoever about $\hat{H}$. In this situation, we provided an optimal simulation algorithm (Thm.~\ref{Thm:Ham_Sim_Qubitization}), notably with only $\mathcal{O}(1)$ ancilla overhead.
\begin{restatable}[Standard-form matrix encoding]{definition}{StandardForm}
\label{Def:Standard_Form}
A matrix $\hat{H} \in \mathbb {C}^{n\times n}$ acting on the system register $s$ is encoded in standard-form-$(\hat{H},\alpha,\hat{U},d)$ with normalization $\alpha \ge \|\hat{H}\|$ by the computational basis state $\ket{0}_a \in \mathbb {C}^d$ on the ancilla register $a$ and signal unitary $\hat{U} \in \mathbb {C}^{d n \times dn}$ if $(\bra{0}_a\otimes \hat{I}_s)\hat{U}(\ket{0}_a\otimes \hat{I}_s)=\hat{H}/\alpha$.\footnote{The unitary $\hat{G}$ defined in~\cite{Low2016hamiltonian} such that $((\bra{0}\hat{G}^\dag)\otimes \hat{I})\hat{U}((\hat{G}\ket{0})\otimes \hat{I})=\hat{H}/\alpha$, which encodes $\hat{H}$ with normalization $\alpha$, may be absorbed into a redefinition of $\hat{U}$. Moreover, for any $\beta > 0$, this is identical to encoding $\hat{H}\beta$ with normalization $\alpha\beta$.
} 
If $\hat{H}$ is also Hermitian, this is called a Herimitian standard-form encoding.
\end{restatable}
\begin{theorem}[Hamiltonian simulation by qubitization Thm.~1~\cite{Low2016hamiltonian}]
\label{Thm:Ham_Sim_Qubitization}
Given Hermitian standard-form-$(\hat{H},\alpha,\hat{U},d)$, there exists a standard-form-$(\hat{X},1,\hat{V},4d)$ such that $\|\hat{X}-e^{-i\hat{H}t}\|\le\epsilon$, where $\hat{V}$ requires $Q=\mathcal{O}\big(t\alpha +\frac{\log{(1/\epsilon)}}{\log\log{(1/\epsilon)}}\big)$ queries to controlled-$\hat{U}$ and $\mathcal{O}(Q\log{(d)})$  primitive gates\footnote{As error $\epsilon$ occurs only in logarithms, it may refer to the trace distance, failure probability, or any other polynomially-related distance without affecting the complexity scaling.}.
\end{theorem}

This motivates the standard-form encoding as the appropriate endpoint when structural information about $\hat{H}$ is provided, though it does not exclude the possibility of superior simulation algorithms not based on the standard-form. As Thm.~\ref{Thm:Ham_Sim_Qubitization} is the optimal simulation algorithm, any exploitation of structure should manifest in minimizing the normalization $\alpha$ of a Hamiltonian encoded in Def.~\ref{Def:Standard_Form}. In order to avoid  accumulating polynomial factors of errors, this must only be with an exponentially small distortion to its spectrum. Moreover, the cost of the procedure should allow for a favorable trade-off in the query complexity of Hamiltonian simulation. Thus manipulation of the standard-form and any additional structural information to this end is what we call the \emph{uniform spectral amplification} problem.
\begin{problem}[Uniform spectral amplification]
\label{Problem:Uniform_Spectral_Amplification}
Given Hermitian standard-form-$(\hat{H},\alpha,\hat{U},d)$, and an upper bound $\Lambda \in[ \|\hat{H}\|,\alpha]$ on the spectral norm, exploit any additional information about $\hat{H}$ or the signal unitary $\hat{U}$ to construct a $Q$-query quantum circuit that encodes $\hat{H}_{\text{amp}}$ in standard-form with normalization $\Lambda$, such that $\|\hat{H}_{\text{amp}}-\hat{H}\|\le \epsilon$, and $Q=o(\alpha/\Lambda)\cdot \mathcal{O}(\text{polylog}(1/\epsilon))$.
\end{problem}
Uniform spectral amplification is non-trivial as it precludes a number of standard techniques. First, amplitude amplification is precluded as the success probability must be boosted for \emph{all} input states to the system. Second, oblivious amplitude amplification~\cite{Berry2014,Berry2015Truncated} is also precluded as $\hat{H}$ is not in general unitary, or even close to unitary. Third, spectral gap amplification~\cite{Somma2013SpectralGap} is precluded as it distorts the spectrum. As such, solving this problem would be of broad interest beyond Hamiltonian simulation. For instance, spectral gap amplification is fundamental to adiabatic state preparation and understanding properties of condensed matter system. Moreover, the prevalence of generalized measurements means that this could also be applicable to quantum observable estimation in metrology and repeat-until-success gate synthesis~\cite{Paetznick2014}.

Some forms of spectral gap amplification have an underlying structure that resembles the amplitude amplification algorithm for quantum state preparation. This suggests that at least one possible solution to uniform spectral amplification could be obtained by solving a related non-trivial \emph{amplitude multiplication} problem, and vice-versa.
\begin{problem}[Amplitude multiplication]
\label{Problem:Uniform_Amplitude_Amplification}
Given a quantum state preparation oracle $\hat{G}\ket{0}_a\ket{0}_b=\lambda \ket{t}_a\ket{0}_b+\sqrt{1-\lambda^2} \ket{t^\perp}_{ab}$, and an upper bound $\Gamma \in [\lambda ,1]$ on the target state overlap, construct a $Q$-query quantum circuit $\hat{V}$ that prepares $\hat{V}\ket{0}_a\ket{0}_b=\lambda_{\text{amp}} \ket{t}_a\ket{0}_b+\cdots \ket{t^\perp}_{ab}$ such that
$|\lambda_{amp}-\lambda/\Gamma|\le \epsilon$, and $Q=\mathcal{O}(\Gamma^{-1}\log{(1/\epsilon)})$.
\end{problem}
Amplitude multiplication is particularly interesting as amplitude amplification and its many other variations~\cite{Yoder2014} amplify target states with the same optimal scaling $\mathcal{O}(\Lambda^{-1})$, but with a highly non-linear dependence on the initial overlap. In contrast, Problem~\ref{Problem:Uniform_Amplitude_Amplification} performs arithmetic multiplication on the amplitudes with exponentially small error, notably \emph{independent of, and without any prior knowledge of their values}.

\subsection{Our Results}
We present quantum algorithms for Hamiltonian simulation based on the general principle of finding solutions to the uniform spectral amplification Problem~\ref{Problem:Uniform_Spectral_Amplification}, which may be  broadly categorized as follows. In `uniform spectral amplification by quantum signal processing', we make no assumptions on the form of the signal unitary in the standard-form encoding of $\hat{H}$, and thus treat as a single unitary oracle. In `uniform spectral amplification by amplitude multiplication', we assume that signal unitary has the structure of factoring into two or three unitary oracles, and by solving amplitude multiplication in Problem~\ref{Problem:Uniform_Amplitude_Amplification}, also approach the sparse simulation results of Claim.~\ref{Claim:Sparse_Ham_Sim}. We then provide a unifying perspective in `universality of the standard-form' which further motivates the standard-form encoding of Hamiltonians as a fundamental ingredient in quantum computation. In greater detail, these results are as follows.

\subsubsection{Uniform Spectral Amplification by Quantum Signal Processing}
\label{Sec:Branch_QSP}
If we make no assumptions on the form of the signal unitary $\hat{U}$ that realizes the standard-form encoding, we treat $\hat{U}$ as a black-box oracle, which we call the standard-form oracle. In this situation, the first result is uniform spectral amplification in Thm.~\ref{Cor:Operator_Amplification} that reduces the normalization $\alpha$ of encoded Hamiltonians to $\mathcal{O}(\Lambda)$ using $\mathcal{O}(\alpha\Lambda^{-1}\log(1/\epsilon))$ queries. This produces a quadratic improvement in success probability when the standard-form is applied to perform quantum measurement, but serves no advantage to Hamiltonian simulation.
\begin{theorem}[Uniform spectral amplification by spectral multiplication]
\label{Cor:Operator_Amplification}
Given Hermitian standard-form-$(\hat{H},\alpha,\hat{U},d)$, let $\Lambda\in[\|\hat{H}\|,\alpha]$. Then for any $\epsilon\le \mathcal{O}(\Lambda/\alpha)$, there exists a standard-form-$(\hat{H}_{\text{amp}},2\Lambda,\hat{V},4d)$ such that $\frac{1}{2\Lambda}\|\hat{H}_{\text{amp}}-\hat{H}\| \le \epsilon$, and $\hat{V}$ requires $\mathcal{O}(\alpha\Lambda^{-1}\log{(1/\epsilon)})$ queries to controlled-$\hat{U}$.
\end{theorem}
The second result is uniform spectral amplification of only the low-energy subspace in Thm.~\ref{Thm:Ham_Encoding_Uniform_Amplification}, of $\hat{H}$ with eigenvalues $\in[-\alpha,-\alpha(1-\Delta)]$, which is of interest to quantum chemistry and adiabatic computation. There, the effective normalization is reduced to $\mathcal{O}(1)$ using $\mathcal{O}(\Delta^{-1/2}\log^{3/2}{(\frac{1}{\Delta\epsilon})})$ queries. This is a generalization of spectral gap amplification~\cite{Somma2013SpectralGap} with the distinction of preserving the relative energy spacing of all relevant states, and of applying to any Hamiltonian encoded in standard-form. When applied to Hamiltonian simulation, an acceleration to $\mathcal{O}\big(t\alpha\sqrt{ \Delta}\log^{3/2}{(t\alpha/\epsilon)}\big)$ queries is obtained in Cor.~\ref{Cor:Ham_Sim_Spectral_Amplification}.
\begin{theorem}[Uniform spectral amplification of low-energy subspaces]
\label{Thm:Ham_Encoding_Uniform_Amplification}
Given Hermitian standard-form-$(\hat{H},\alpha,\hat{U},d)$ with eigenstates $\hat{H}/\alpha\ket{\lambda}=\lambda\ket{\lambda}$, let $\Delta \in(0,1)$ be a positive constant, and $\hat{\Pi}=\sum_{\lambda \in[-1,-1+\Delta]}\ket{\lambda}\bra{\lambda}$ be a projector onto the low-energy subspace of $\hat{H}$. Then there exists a standard-form-$(\hat{H}_{\text{amp}},\Delta\alpha,\hat{V},4d)$ such that $\|\hat{\Pi}(\frac{\hat{H}_{\text{amp}}}{\Delta\alpha}-\frac{\hat{H}+\alpha\hat{I}(1-\Delta)}{\Delta\alpha})\hat{\Pi}\|\le \epsilon$, and $\hat{V}$ requires $\mathcal{O}(\Delta^{-1/2}\log^{3/2}{(\frac{1}{\Delta\epsilon})})$ queries to controlled-$\hat{U}$.
\end{theorem}

These results stem primarily from constructing polynomials with desirable properties, which we implement using the technique of Thm.~\ref{Thm:QSP_B}. This flexible variant of quantum signal processing is subject to fewer constraints than in prior art. Moreover, the advantage of quantum signal processing over the related technique of linear-combination-of-unitaries~\cite{Berry2015Hamiltonian} is its avoidance of Hamiltonian simulation as an intermediate step. This reduces overhead in space, query complexity, and error, and leads to an extremely simple algorithm that directly implements polynomial functions of $\hat{H}$ without any approximation.  
\begin{theorem}[Flexible quantum signal processing]
\label{Thm:QSP_B}
Given Hermitian standard-form-$(\hat{H},1,\hat{U},d)$, let $B$ be any function that satisfies the all the following conditions:
\\
(1) ${B}(x)=\sum^N_{j=0}b_j x^j$ is a real parity-$(N \mod 2)$ polynomial of degree at most $N$;
\\
(2) $B(0)=0$;
\\
(3) $\forall x\in[-1,1]$, $B^2(x)\le 1$.
\\
Then there exists a Hermitian standard-form-$(B[\hat{H}],1,\hat{V},4d)$, where $B[\hat{H}]=\sum^N_{j=0}b_j \hat{H}^j$, and $\hat{V}$ requires $\mathcal{O}(N)$ queries to controlled-$\hat{U}$ and $\mathcal{O}(N\log(d))$ primitive quantum gates pre-computed in classical $\mathcal{O}(\text{poly}(N))$ time.
\end{theorem}
\subsubsection{Uniform Spectral Amplification by Amplitude Multiplication}
\label{Sec:Branch_AM}
Alternatively, we here assume that the signal unitary $\hat{U}$ that realizes the standard-form encoding factors into two or three unitary quantum oracles $\hat{U}_\text{row}$, $\hat{U}_\text{col}$, and $\hat{U}_{\text{mix}}$, which we also call standard-form oracles. When the signal unitary factors into two components $\hat{U}=\hat{U}_\text{row}^\dag\hat{U}_\text{col}$, this constrains the representation of matrices in the standard-form to have matrix elements of $\hat{H}$ that are exactly the overlap of appropriately defined quantum states, and generalizes the sparse matrix model first introduced by Childs~\cite{Childs2010} for quantum walks. When  the signal unitary factors into three $\hat{U}=\hat{U}_\text{row}^\dag\hat{U}_{\text{mix}}\hat{U}_\text{col}$ components, amplitude amplification can be applied to obtain non-trivial Hamiltonians.

Note that amplitude amplification had been previously considered in the context of sparse Hamiltonian simulation~\cite{Berry2012}. However, its non-linearity introduced a polynomial dependence on error, which compounded into a polynomial overhead in scaling with respect to time and error. In constrast, our solution to the amplitude multiplication problem Problem~\ref{Problem:Uniform_Amplitude_Amplification} achieves uniform spectral amplification by multiplying all state overlaps by the same constant factor.
Specializing the general result Lem.~\ref{Thm:Ham_Encoding_Uniform_Amplification_State_Overlaps} to the case of sparse Hamiltonians, which are described by standard black-box quantum oracles (Def.~\ref{Def:Sparse_Oracle}) to its non-zero matrix elements and positions, furnishes a simulation algorithm matching the complexity of Claim.~\ref{Claim:Sparse_Ham_Sim}, up to logarithmic factors. Modulo these logarithmic factors, this an improvement over prior art, with either best-case square-root improvement in sparsity~\cite{Low2016HamSim}, or a polynomial improvement in time and exponential improvement in precision~\cite{Berry2012}
\begin{definition}[Sparse matrix oracles~\cite{Berry2012}]
\label{Def:Sparse_Oracle}
Sparse matrices with at most $d$ non-zero elements in every row are specified by two oracles.
The oracle $\hat{O}_{H}\ket{j}\ket{k}\ket{z}=\ket{j}\ket{k}\ket{z\oplus \hat{H}_{jk}}$ queried by $j\in[n]$ row and $k\in[n]$ column indices returns the value $\hat{H}_{jk}=\bra{j}\hat{H}\ket{k}$, with maximum absolute value $\|\hat{H}\|_{\text{max}}=\max_{jk}{|\hat{H}_{jk}|}$. The oracle $\hat{O}_{F}\ket{j}\ket{l}=\ket{j}\ket{f(j,l)}$ queried by $j\in[n]$ row and $l\in[d]$ column indices computes in-place the column index $f(j,l)$ of the $l^{\text{th}}$ non-zero entry of the $j^{\text{th}}$ row.
\end{definition}
\begin{theorem}[Sparse Hamiltonian simulation by amplified state overlap]
\label{Cor:Ham_Sim_Sparse_Amplified}
Given the $d$-sparse matrix oracles in Def.~\ref{Def:Sparse_Oracle} for the Hamiltonian $\hat{H}$, let $\|\hat{H}\|_{\text{max}}=\max_{jk}|\hat{H}_{jk}|$ be the max-norm,  $\|\hat{H}\|_1=\max_{j}\sum_{k}|\hat{H}_{jk}|$ be the induced $1$-norm  , and $\|\hat{H}\|$ be spectral norm. Then $\forall t\ge 0,\; \epsilon >0$, the operator
$e^{-i\hat{H}t}$ can be approximated with error $\epsilon$ using $\mathcal{O}\left(t(d\|\hat{H}\|_{\text{max}}\|\hat{H}\|_1)^{1/2}\log{(\frac{t\|\hat{H}\|}{\epsilon})}\left(1+\frac{1}{t\|\hat{H}\|_1}\frac{\log{(1/\epsilon)}}{\log\log{(1/\epsilon)}}\right)\right)$ queries.
\end{theorem}
Observe that in the asymptotic limit of large $\|\hat{H}\|_1 t \gg \log{(1/\epsilon)}$, the query complexity simplifies to $\mathcal{O}\Big(t(d\|\hat{H}\|_{\text{max}}\|\hat{H}\|_1)^{1/2}\log{(\frac{t\|\hat{H}\|}{\epsilon})}\Big)$.
The algorithm of Thm.~\ref{Cor:Ham_Sim_Sparse_Amplified} is particularly flexible. If none of the above norms are known, they may be replaced by any upper bound, such as determined by the inequalities $\|\hat{H}\|_{\text{max}}\le \|\hat{H}\|\le \|\hat{H}\|_{1}\le d \|\hat{H}\|_{\text{max}}$~\cite{Childs2010Limitation}. Even in the worst case, the results are similar to previous optimal simulation algorithms.
 Moreover, the scaling in these parameters is optimal  as we prove matching lower bound Thm.~\ref{Thm:Lower_Bound} by finding a Hamiltonian that solves $\text{PARITY}\circ\text{OR}$.
\begin{theorem}
\label{Thm:Lower_Bound}
For any $d\ge 1$, $s\ge 1$, and $t>0$, there exists a Hamiltonian $\hat{H}$ with sparsity $\Theta(d)$, $\|\hat{H}\|_{\text{max}}=\Theta(1)$, and $\|\hat{H}\|_1 = \Theta(s)$, such that approximating time evolution $e^{-i\hat{H}t}$ with constant  error requires $\Omega(t\sqrt{d s})$ queries.
\end{theorem}

Some of these results stem from constructing polynomials with desirable properties, which we implement using the technique of Thm.~\ref{Thm:Controlled_Generalized_Amplitude_Amplification}. The existence of a weaker version of this amplitude amplification algorithm was suggested in our prior work~\cite{Low2016methodology}. Here, we present that and go further. This variant of amplitude amplification allows one to amplify target state overlaps with almost arbitrary polynomials functions.
\begin{theorem}[Flexible amplitude amplification]
\label{Thm:Controlled_Generalized_Amplitude_Amplification}
Given a state preparation unitary $\hat{G}$ acting on the computational basis states $\ket{0}_a\in \mathbb{C}^d$, $\ket{0}_b\in \mathbb{C}^2$ such that $\hat{G}\ket{0}_a\ket{0}_b=\lambda\ket{t}_a\ket{0}_b+\sqrt{1-\lambda^2}\ket{t^\perp}_{ab}$, where $\ket{t^\perp}_{ab}$ has no support on $\ket{0}_b$, let $D$ be any function that satisfies all the following conditions:
\\
(1) $D$ is an odd real polynomial in $\lambda$ of degree at most $2N+1$;
\\
(2) $\forall \lambda\in[-1,1]$, ${D}^2(\lambda)\le 1$. 
\\
Then there exists a quantum circuit $\hat{W}_{\vec\phi}$ such that $\bra{t}_a\bra{0}_b\bra{0}_c\hat{W}_{\vec\phi}\ket{0}_a\ket{0}_b\ket{0}_c=D(\lambda)$, using  $N+1$ queries to $\hat{G}$, $N$ queries to $\hat{G}^\dag$, $\mathcal{O}(N\log{(d)})$ primitive quantum gates pre-computed from $D$ in classical $\mathcal{O}(\text{poly}(N))$ time, and an additional qubit ancilla $c$, such that 
\end{theorem}
Amplitude multiplication in Thm.~\ref{Thm:Linear_Amplitude_Amplification} is then a special case that solves Problem~\ref{Problem:Uniform_Amplitude_Amplification} up to a factor of $\frac{1}{2}$ in the range of the input and output amplitudes.
\begin{theorem}[Amplitude multiplication algorithm]
\label{Thm:Linear_Amplitude_Amplification}
$\forall\;\lambda \in [-1/2,1/2]$, $\Gamma \in (|\lambda|, 1/2]$, $\epsilon \le \mathcal{O}(\Gamma)$, let $\hat{G}$ be a state preparation unitary acting on the computational basis states $\ket{0}_a\in \mathbb{C}^d$, $\ket{0}_b\in \mathbb{C}^2$ such that $\hat{G}\ket{0}_a\ket{0}_b=\lambda\ket{t}_a\ket{0}_b+\sqrt{1-\lambda^2}\ket{t^\perp}_{ab}$, where $\ket{t^\perp}_{ab}$ has no support on $\ket{0}_b$. Then there exists a quantum circuit $\hat{G}'$ such that $\left|\bra{t}_a\bra{0}_b\bra{0}_c\hat{G}'\ket{0}_a\ket{0}_b\ket{0}_c- \frac{\lambda}{2\Gamma}\right|\le \frac{|\lambda|}{2\Gamma}\epsilon$, using $Q=\mathcal{O}(\Gamma^{-1}\log{(1/\epsilon)})$ queries to $\hat{G},\hat{G}^\dag$, $\mathcal{O}(Q\log{(d)})$ primitive quantum gates, and an additional ancilla qubit $c$.
\end{theorem}

\subsubsection{Universality of the Standard-Form}
\label{Sec:Branch_UNI}

Uniform spectral amplification is motivated by the idea that structure in the signal unitary and its encoded Hamiltonian can be fully exploited by focusing only on the manipulating the standard-form, independent of any later application such as Hamiltonian simulation. This is supported by the simulation algorithm Thm.~\ref{Thm:Ham_Sim_Qubitization} which is optimal with respect to all parameters when the standard-form is provided as a black-box oracle. This perspective would be further justified if one could rule out, to a reasonable extent, the existence of superior simulation algorithms not based on the standard-form.

We show certain universality of the standard-form by proving an equivalence between quantum circuits for simulation and those for quantum measurement, up to a logarithmic overhead in time and a constant overhead in space. Where Thm.~\ref{Thm:Ham_Sim_Qubitization} transforms a measurement of $\hat{H}$ to time-evolution by $e^{-i \hat{H}t}$, we prove the converse in Thm.~\ref{Thm:Standard_Form_From_Ham_Sim} which transforms time-evolution $e^{-i \hat{H}t}$ back into measurement $\hat{H}$. In particular, this is with an exponential improvement in precision over standard techniques based on quantum phase estimation. Thus any non-standard-form simulation algorithm for $e^{-i\hat{H}t}$ that exploits structure can be always mapped in this manner onto the standard-form with a small overhead. 
 \begin{theorem}[Standard-form encoding by Hamiltonian simulation]
\label{Thm:Standard_Form_From_Ham_Sim}
Given oracle access to the controlled time-evolution $e^{-i\hat{H}}$ such that $\|\hat{H}\|\le 1/2$, there exists a standard-form-$(\hat{H}_{\text{lin}},1,\hat{U},4)$ such that $\|\hat{H}_{\text{lin}}-\hat{H}\| \le \epsilon$, where $\hat{U}$ requires $Q=\mathcal{O}\left(\log{(1/\epsilon)}\right)$ queries and $\mathcal{O}(Q)$ primitive quantum gates.
\end{theorem}
This is proven through the flexible quantum signal processing Thm.~\ref{Thm:QSP_B} using a particular choice of polynomial. It is important to note however the caveat that our equivalence limits $\|\hat{H}t\| = \mathcal{O}(1)$, and also fails when time-evolution can be approximated with $o(t)$ queries. Fortunately, the latter scenario can be disregarded with limited loss as `no-fast-forwarding' theorems~\cite{Childs2010Limitation} prove the necessity of $\Omega(\|\hat{H}\|t)$ queries for generic computational problems and physical systems.

One useful application of this reverse direction is an alternate technique Cor.~\ref{Cor:HamExponentials} for simulating time evolution by a sum of $d$ Hermitian components $\sum_{d=1}\hat{H}_j$, given their controlled-exponentials $e^{-i\hat{H}_jt_j}$. This approach is considerably simpler than that of compressed fractional queries~\cite{Berry2014}, and essentially works by using Thm.~\ref{Thm:Standard_Form_From_Ham_Sim} to map each $e^{-i\hat{H}_jt_j}$, where $\|\hat{H}_jt_j\|=\mathcal{O}(1)$ to a standard-form encoding of $\hat{H}_jt_j$.
\begin{corollary}[Hamiltonian simulation with exponentials]
\label{Cor:HamExponentials}
Given standard-form-$(\sum^d_{j=1}\alpha_je^{-i\hat{H}_j},\alpha,\hat{G}^\dag_a\hat{U}\hat{G}_a,d)$, where $\hat{G}$ that prepares $\ket{G}_a=\sum^d_{j=1}\sqrt{\alpha_j/\alpha}\ket{j}_a$ with $\alpha_j\ge 0$, normalization $\alpha=\sum^d_{j=1}\alpha_j$ and signal oracle
$\hat{U}=\sum_{j=i}^d\ket{j}\bra{j}_a\otimes e^{-i \hat{H}_j}$, with $\|\hat{H}_j\|\le 1$, there exists a standard-form-$(\hat{X},1,\hat{V},4d)$ such that $\|\hat{X}-e^{-i\hat{H}t}\|\le\epsilon$, where $\hat{V}$ requires $\mathcal{O}\left(\alpha t \log{(\alpha t/\epsilon)}+\frac{\log{(1/\epsilon)}\log{(\alpha t/\epsilon)}}{\log\log{(\alpha t/\epsilon)}}\right)$ controlled-queries, and $\mathcal{O}(Q\log{(d)})$ primitive quantum gates.
\end{corollary}

\subsection{Organization}
The dependencies of our results are summarized in Figure~\ref{Fig:Dependencies}.
\begin{itemize}
\item [Part I] is where where we achieve uniform spectral amplification by quantum signal processing. We describe in Sec.~\ref{Sec:Standard-form_QSP} the technique of quantum signal processing in prior art and prove the more useful variant Thm.~\ref{Thm:QSP_B}. This applied in Sec.~\ref{Sec:Uniform_Hamiltonian_Amplification}, where we treat the signal unitary as a single unitary oracle, and prove the solutions Thm.~\ref{Cor:Operator_Amplification} and Thm.~\ref{Thm:Ham_Encoding_Uniform_Amplification} to the uniform spectral amplification problem.
\item [Part II] is where we achieve uniform spectral amplification by amplitude multiplication. We prove in Sec.~\ref{Sec:AA_by_QSP} a generalization of amplitude amplification in Thm.~\ref{Thm:Controlled_Generalized_Amplitude_Amplification}, which is applied to obtain the amplitude multiplication algorithm of Thm.~\ref{Thm:Linear_Amplitude_Amplification}. Subsequently in Sec.~\ref{Sec:Ham_Sim_Overlaps}, we consider signal unitaries that factors into two or three unitary oracles. This motivates a general model of Hamiltonians encoded by state overlaps, where uniform spectral amplification in Lem.~\ref{Thm:Ham_Encoding_Uniform_Amplification_State_Overlaps} is enabled by amplitude multiplication. Applying these results to the special case of sparse matrices leads to the simulation algorithm Thm.~\ref{Cor:Ham_Sim_Sparse_Amplified}, which matches the lower bound Thm.~\ref{Thm:Lower_Bound}.
\item [Part III] in Sec.~\ref{Sec:Equivalence_Sim_Mea} is where we offer a unifying perspective of simulation algorithms and prove a certain universality of the standard-form. This is through the equivalence between quantum circuits for simulation and those for measurement described by Thm.~\ref{Thm:Standard_Form_From_Ham_Sim}, and leads to the simulation algorithm Cor.~\ref{Cor:HamExponentials}. 
\end{itemize}
We conclude in Sec.~\ref{Sec:Amp_concluson}.

\begin{figure}
\centering
\label{Fig:Dependencies}
\begin{tikzpicture}
[->,>=stealth',shorten >=1pt,auto, thick,
main node/.style={circle,draw}, node distance = 0.3cm and 0.4cm,
block/.style   ={rectangle, draw, text width=10em, text centered, rounded corners, minimum height=2.3em, fill=white, align=center, font={\fontsize{8pt}{9}\selectfont}},
emptyblock/.style   ={block, draw = none , font={\fontsize{8pt}{9}\selectfont}},
]
\node[main node,  block] (1) {Composite quantum gates~\cite{Low2016methodology}};
\node[main node, block,  draw=none] (d1) [above = of 1] {};
  
\node[main node,  block] (QSP) [above = of d1] {Quantum signal processing~\cite{Low2016HamSim,Low2016hamiltonian}};
\node[main node, block,  draw=none] (d2) [above = of QSP] {}; 
 
\node[main node,  block] (Qubit) [above = of d2] {Qubitization~\cite{Low2016hamiltonian}};
\node[main node,  block] (Standard-form) [above = of Qubit] {Standard-form\\(Sec.~\ref{Sec:Introduction}, Def.~\ref{Def:Standard_Form})}; 

\node[emptyblock] (heading-background) [ above = of Standard-form] {Background};
\node[emptyblock] (heading-techniques) [ right = of heading-background] {General techniques};
\node[emptyblock] (heading-USA) [ right = of heading-techniques] {Uniform spectral \\ amplification results};
\node[emptyblock] (heading-HS) [ right = of heading-USA] {Hamiltonian \\ simulation results};

\draw[thick,-] let \p1=(heading-background),\p2=(heading-techniques) in 
    ({(\x1+\x2)/2},\y1+0.3cm) -- ({(\x1+\x2)/2},\y1-0.3cm);
\draw[thick,-] let \p1=(heading-USA),\p2=(heading-techniques) in 
    ({(\x1+\x2)/2},\y1+0.3cm) -- ({(\x1+\x2)/2},\y1-0.3cm);
\draw[thick,-] let \p1=(heading-USA),\p2=(heading-HS) in 
    ({(\x1+\x2)/2},\y1+0.3cm) -- ({(\x1+\x2)/2},\y1-0.3cm);

\node[main node,  block] (HS-Qubit) [right=  of QSP] {Hamiltonian simulation by qubitization~\cite{Low2016hamiltonian}\\(Sec.~\ref{Sec:Introduction}, Thm.~\ref{Thm:Ham_Sim_Qubitization})};

\node[main node,  block] (F-AA) [right=  of 1] {Flexible amplitude amplification \\ (Sec.~\ref{Sec:Flexible_AA}, Thm.~\ref{Thm:Controlled_Generalized_Amplitude_Amplification})};
 
\node[main node,  block] (F-QSP) [right = of Qubit] {Flexible quantum signal processing \\ (Sec.~\ref{Sec:QSP_Constraint_Free}, Thm.~\ref{Thm:QSP_B})}; 
 
\node[main node,  block] (USA-Spec-Low-E) [right = of F-QSP] {Spectral amplification of low-energy subspaces \\(Sec.~\ref{Sec:Uniform_Hamiltonian_Amplification}, Thm~\ref{Thm:Ham_Encoding_Uniform_Amplification})}; 

\node[main node,  block] (USA-Spec-Mult) [above right = of F-QSP] {Spectral multiplication \\ (Sec.~\ref{Sec:Uniform_Hamiltonian_Amplification}, Thm.~\ref{Cor:Operator_Amplification})}; 

\node[main node,  block] (HS-Low-E) [right = of USA-Spec-Low-E] {Simulation of low-energy  subspaces \\ (Sec.~\ref{Sec:Uniform_Hamiltonian_Amplification}, Cor.~\ref{Cor:Ham_Sim_Spectral_Amplification})}; 
 
\node[main node,  block, draw=none, fill=none, minimum height=3.05em] (Middle-HamSim-Low-E) at ($(USA-Spec-Low-E)!1.05!(HS-Low-E)$) {}; 

\node[main node,  block] (SF-Sim) [below right = of F-QSP] {Standard-form given exponentials \\ (Sec.~\ref{Sec:Equivalence_Sim_Mea}, Thm.~\ref{Thm:Standard_Form_From_Ham_Sim})}; 

\node[main node,  block] (HS-Exp) [right = of SF-Sim] {Simulation of exponentials \\(Sec.~\ref{Sec:Equivalence_Sim_Mea}, Cor.~\ref{Cor:HamExponentials})}; 
 
\node[main node,  block] (Amp-Mult) [below = of F-AA] {Amplitude multiplication \\(Sec.~\ref{Sec:Amp_Mult}, Thm.~\ref{Thm:Linear_Amplitude_Amplification})};  

  \node[main node,  block] (USA-Overlap-Mult) [above right = of Amp-Mult] {Spectral amplification by multipled  overlaps \\(Sec.~\ref{SubSec:Amplified_Overlap}, Lem.~\ref{Thm:Ham_Encoding_Uniform_Amplification_State_Overlaps})};  
\node[main node,  block, draw=none, fill=none] (Middle-Amp-Mult) at ($(F-AA)!0.5!(USA-Overlap-Mult)$) {};

  \node[main node,  block] (HamSim-Overlap-Mult) [above right = of USA-Overlap-Mult] {Simulation  by multipled overlaps \\(Sec.~\ref{SubSec:Amplified_Overlap}, Lem.~\ref{Thm:HamSim_Amplified_Overlaps})};  
 
\node[main node,  block] (HamSim-Sparse-Mult) [below = of HamSim-Overlap-Mult] 
 {Sparse simulation by multipled overlaps \\
 (Sec.~\ref{SubSec:Reduction_Sparse_Matrices}, Thm.~\ref{Cor:Ham_Sim_Sparse_Amplified})
 }; 
 
 \node[main node,  block] (HamSim-Lower-Bound) [below = of HamSim-Sparse-Mult] { Simulation lower bound (Sec.~\ref{Sec:Lower_Bound}, Thm.~\ref{Thm:Lower_Bound})};
 
 \node[main node, block,  draw=none, fill=none] (d3) [right = of HS-Qubit] {};
 \node[main node, block,  draw=none, fill=none] (d4) [right = of d3] {}; 
 
\node[draw=none] () [above = of F-QSP] {Part I};
\node[draw=none] () at ($(F-QSP)!0.5!(HS-Qubit)$) {Part III};
\node[draw=none] () [below = of HS-Qubit] {Part II};
\begin{pgfonlayer}{background}
	\filldraw[fill=blue!40!white, draw=black, rounded corners, opacity=0.5] let \p1=(F-QSP),\p2=(USA-Spec-Mult), \p3=(HS-Low-E) in ({\x1-6em},\y1-1.9em) rectangle (\x3+6em,\y2+1.5em);
	\filldraw[fill=green!40!white, draw=black, rounded corners, opacity=0.5] let \p1=(F-QSP),\p2=(SF-Sim), \p3=(HS-Low-E) in ({\x1-6em},\y2-1.9em) rectangle (\x3+6em,\y2+1.9em);
	\filldraw[fill=yellow!40!white, draw=black, rounded corners, opacity=0.5] let \p1=(Amp-Mult),\p2=(HamSim-Overlap-Mult) in ({\x1-6em},\y1-1.5em) rectangle (\x2+6em,\y2+1.9em);

\end{pgfonlayer}  


  \path[every node/.style={font=\sffamily\small}]
  (Standard-form) edge node [right] {} (Qubit)
  (QSP) edge node [right] {} (HS-Qubit)
  (Qubit) edge node [right] {} (QSP)
  (QSP) edge node [right] {} (F-QSP)
  (1) edge node [right] {} (QSP)
  (1) edge node [right] {} (F-AA)
 (SF-Sim) edge node [right] {}  (HS-Exp)
  (F-QSP) edge node [right] {} (SF-Sim)
 (F-QSP) edge node [right] {} (USA-Spec-Low-E)
 (F-QSP) edge node [right] {} (USA-Spec-Mult)
 (USA-Spec-Low-E) edge node [right] {} (HS-Low-E)
(F-AA) edge node [right] {} (Amp-Mult)   
(Amp-Mult) edge node [right] {} (USA-Overlap-Mult)
  (USA-Overlap-Mult) edge node [right] {} (HamSim-Overlap-Mult)
   (HamSim-Overlap-Mult) edge node [right] {} (HamSim-Sparse-Mult)
  ;
  \begin{pgfonlayer}{background}
\draw[->,shorten >=1pt,>=stealth',semithick] (HS-Qubit)-|(HamSim-Overlap-Mult);
\draw[->,shorten >=1pt,>=stealth',semithick] (HS-Qubit)-|(HS-Exp);
 \draw[->,shorten >=1pt,>=stealth',semithick] (HS-Qubit)-|(Middle-HamSim-Low-E);
  \end{pgfonlayer}
\end{tikzpicture}
\caption{Dependencies of new results.}
\end{figure}

\section{Quantum Signal Processing Techniques}
\label{Sec:Standard-form_QSP}
Quantum signal processing is a very new technique~\cite{Low2016HamSim,Low2016hamiltonian}, based on optimal quantum control~\cite{Low2016methodology} and qubitization~\cite{Low2016hamiltonian}, for implementing polynomial functions of the Hamiltonian $\hat{H}$ given its standard-form encoding. This is performed with optimal query complexity, $\mathcal{O}(1)$ ancilla overhead, and without approximation. We outline in Sec.~\ref{Sec:Standard-form_QSP_Prior_art} the basic version Lem.~\ref{Thm:QSP_AB} that was introduced in~\cite{Low2016hamiltonian}, which imposes certain unintuitive constraints on valid polynomials. Subsequently, we prove in Sec.~\ref{Sec:QSP_Constraint_Free} its generalization Thm.~\ref{Thm:QSP_B} that drops these constraints, and is applied frequently to obtain our other results. 

\subsection{Quantum Signal Processing in Prior Art}
\label{Sec:Standard-form_QSP_Prior_art}
Given any Hermitian matrix encoded in the standard-form-$(\hat{H},\alpha,\hat{U},d)$ of Def.~\ref{Def:Standard_Form}, let Hermitian $\hat{H}\in \mathbb{C}^{n\times n}:\mathcal{H}_s\rightarrow \mathcal{H}_s$ act on the system Hilbert space $\mathcal{H}_s$ of dimension $n$. Then the signal unitary $\hat{U}\in\mathbb{C}^{nd\times nd}:\mathcal{H}_s\otimes \mathcal{H}_a\rightarrow \mathcal{H}_s\otimes \mathcal{H}_a$ acts jointly on the system $\mathcal{H}_s$ and  dimension $d$ ancilla $\mathcal{H}_a$ register. Using the computational basis state $\ket{0}_a$, $(\bra{0}_a\otimes \hat{I}_s)\hat{U}(\ket{0}_a\otimes \hat{I}_s)=\hat{H}/\alpha$ with normalization $\alpha \ge \|\hat{H}\|$. Note that in~\cite{Low2016hamiltonian}, a different measurement basis $\ket{G}_a=\hat{G}\ket{0}_a\in \mathcal{H}_a$ is used to encode $(\bra{G}_a\otimes \hat{I}_s)\hat{U}(\ket{G}_a\otimes \hat{I}_s)=\hat{H}/\alpha$ as some structured Hamiltonians are more naturally represented that way. Assuming oracle access to the state preparation unitary $\hat{G}$, this is entirely equivalent as we may always absorb $\hat{G}$ into a redefinition $\hat{G}^\dag\hat{U}\hat{G}$ of the signal unitary. In Sec.~\ref{Sec:Standard-form_QSP} and Sec.~\ref{Sec:Equivalence_Sim_Mea} only, we find it useful to have $\hat{G}$ explicit, and also absorb the normalization into a rescaled Hamiltonian $\hat{H}'=\hat{H}/\alpha$ with eigenstate $\hat{H}'\ket{\lambda}=\lambda\ket{\lambda}$ and spectral norm $\|\hat{H}'\|\le 1$. 

Quantum signal processing~\cite{Low2016HamSim} characterizes the query complexity of implementing large classes of functions $f[\hat{H}']\doteq\sum_\lambda f(\lambda)\ket{\lambda}\bra{\lambda}$. Using $\mathcal{O}(N)$ standard-form queries, $\mathcal{O}(N\log{(d)})$ primitive quantum gates, and at most $1$ additional ancilla qubit $b$, one can construct a useful quantum circuit $\hat{W}_{\vec{\phi}}$, the \emph{composite qubiterate} depicted in Fig.~\ref{Fig:Circuit_Qubitization_QSP}, that is parameterized by $
\vec{\phi}\in \mathbb{R}^N$ and an ancilla state $\ket{0}_{ab}$. 
The gate cost of reflections about the $2d$ dimensional state $\ket{0}_a\ket{0}_b$ depends on the $2^{\mathcal{O}(\log_2 d)}$-controlled Toffoli gate. A $\mathcal{O}(\log(d))$ primitive gate decomposition is provided in~\cite{He2017decompositions} using any one other uninitialized ancilla qubit, which we may take from register $s$. For each eigenstate $\ket{\lambda}_s$, $\hat{V}_{\vec{\phi}}$ has the following properties:
\begin{align}
\label{Eq:QSP_Baseline}
\hat{W}_{\vec{\phi}}\ket{0}_{ab}\ket{\lambda}_s&=e^{-i \hat{\sigma}_{\phi_{N}}\theta_\lambda}e^{-i \hat{\sigma}_{\phi_{N-1}}\theta_\lambda}\cdots e^{-i \hat{\sigma}_{\phi_{1}}\theta_\lambda}\ket{0}_{ab}\ket{\lambda}_s, 
\\\nonumber
&= \left(\mathcal A(\theta_\lambda)\hat{I}_\lambda+i\mathcal B(\theta_\lambda))\hat\sigma_{z,\lambda} + i\mathcal C(\theta_\lambda)\hat\sigma_{x,\lambda}+i\mathcal D(\theta_\lambda)\hat\sigma_{y,\lambda}\right)\ket{0}_{ab}\ket{\lambda}_s
\\\nonumber
&= (\mathcal A(\theta_\lambda)+i\mathcal B(\theta_\lambda))\ket{0}_{ab}\ket{\lambda}_s + (i\mathcal C(\theta_\lambda)-\mathcal D(\theta_\lambda))\ket{0\lambda^\perp}_{abs}, 
\quad (\bra{0}_{ab}\bra{\lambda}_s)\ket{0\lambda^\perp}_{abs}=0,
\end{align}
where $\theta_\lambda = \cos^{-1}{(\lambda)}$. The Pauli matrices $\hat{I}_\lambda,\hat{\sigma}_{x,\lambda},\hat{\sigma}_{y,\lambda},\hat{\sigma}_{z,\lambda}$ act on the two-dimensional subspace $\mathcal{H}_\lambda=\text{span}\{\ket{0}_{ab}\ket{\lambda}_s,\ket{0\lambda^\perp}_{abs}\}$ with bases defined through $\hat{\sigma}_{\lambda,z} \ket{0}_{ab}\ket{\lambda}_s=\ket{0}_{ab}\ket{\lambda}_s$, $\hat{\sigma}_{z,\lambda} \ket{0\lambda^\perp}_{abs}=-\ket{0\lambda^\perp}_{abs}$. The only property of the states $\ket{0\lambda^\perp}_{abs}$ that concerns us is they are mutually orthogonal, and also orthogonal to all states $\ket{0}_{ab}\ket{\lambda}_s$. Note that the functions $(\mathcal A,\mathcal B,\mathcal C,\mathcal D)$ of an angle are implicitly parameterized $\vec\phi$. We find it useful to define the functions $(A,B,C,D)$ of $\lambda$ related by a variable substitution e.g. $\mathcal A(\theta_\lambda)=A(\cos{(\theta_\lambda)})$. These functions are not independent as unitarity at the very least requires $\mathcal A^2+\mathcal B^2+\mathcal C^2+\mathcal D^2=1$. By identifying $\hat{W}_{\vec{\phi}}$ as the signal unitary and $\ket{0}_{ab}$ as the measurement basis, $(\bra{0}_{ab}\otimes\hat{I}_s)\hat{W}_{\vec{\phi}}(\ket{0}_{ab}\otimes\hat{I}_s)=A[\hat{H}']+iB[\hat{H}']$ itself encodes the matrix $A[\hat{H}']+iB[\hat{H}']$ in  standard-form-$(A[\hat{H}']+iB[\hat{H}'],1,\hat{W}_{\vec{\phi}},2d)$.
\begin{figure}[t]
\centering
\includegraphics{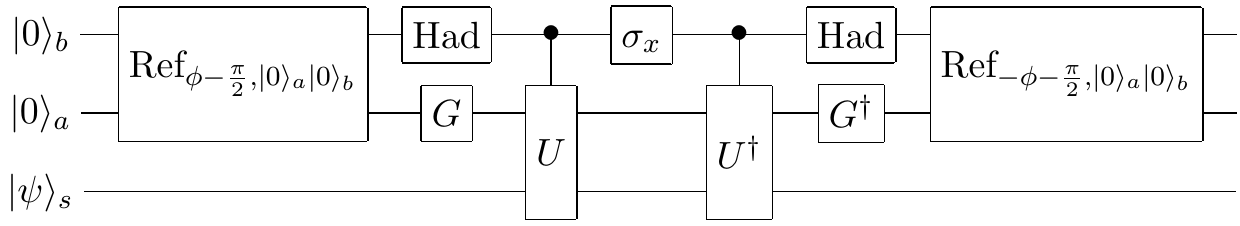}\\
\vspace{0.25cm}
\includegraphics{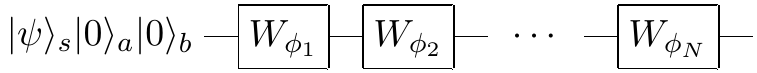}
\caption{
\label{Fig:Circuit_Qubitization_QSP}(top) Circuit diagram for the phased qubiterate $\hat{W}_\phi$ constructed by the qubitization of a standard-form encoding-$(\hat{H},1,(\hat{G}^\dag\otimes \hat{I}_s)\hat{U}(\hat{G}\otimes\hat{I}_s),d)$ from~\cite{Low2016hamiltonian}. The phased qubiterate $\hat{W}_\phi$ encodes $\hat{H}$ in standard-form-$(\hat{H},1,\hat{W}_\phi,2d)$. Note that $\widehat{\text{Had}}$ is the Hadamard gate, $\hat{\sigma}_x$ is a single-qubit NOT gate, and we define the reflection $\widehat{\text{Ref}}_{\alpha,\ket{0}_a\ket{0}_b}=\hat{I}_{ab}-(1-e^{-i\alpha})\ket{0}\bra{0}_a\otimes\ket{0}\bra{0}_b$. The gate complexity of $\hat{W}_\phi$ is $\mathcal{O}(\log{(d)})$. (bottom) Circuit diagram for the composite qubiterate $\hat{W}_{\vec{\phi}}$ that encodes a standard-form-$(A[\hat{H}]+iB[\hat{H}],1,\hat{W}_{\vec{\phi}},2d)$. The query complexity of $\hat{W}_{\vec{\phi}}$ is $N$ to $\hat{G}$, controlled-$\hat{U}$, and their inverses. Its gate complexity is $\mathcal{O}(N\log{(d)})$.
}
\end{figure}

We previously studied~\cite{Low2016methodology} sequences of single-qubit rotations isomorphic to those in Eq.~\ref{Eq:QSP_Baseline}:
\begin{align}
\label{Eq:QSP_Single_Qubit}
e^{-i \hat{\sigma}_{\phi_{N}}\theta}e^{-i \hat{\sigma}_{\phi_{N-1}}\theta}\cdots e^{-i \hat{\sigma}_{\phi_{1}}\theta}
=
\mathcal{A}(\theta)\hat{I}+i\mathcal{B}(\theta)\hat\sigma_{z} + i\mathcal{C}(\theta)\hat\sigma_{x}+i\mathcal{D}(\theta)\hat\sigma_{y},
\end{align}
fully characterized the functions $(\mathcal{A},\mathcal{B},\mathcal{C},\mathcal{D})$ implementable by any choice of $\vec{\phi}$, and also provided an efficient classical algorithm to invert any valid partial specification of $(\mathcal{A},\mathcal{B},\mathcal{C},\mathcal{D})$ to obtain its implementation $\vec{\phi}$. For instance, we have the following result regarding Eq.~\ref{Eq:QSP_Single_Qubit}
\begin{lemma}[Achievable $(\mathcal{A},\mathcal{B})$ -- Thm.~2.3 of~\cite{Low2016methodology}]
\label{Lem:AchievableAB}
For any integer $N>0$, a choice of functions $\mathcal{A},\mathcal{B}$ in Eq.~\ref{Eq:QSP_Single_Qubit} is achievable by some $\vec\phi\in\mathbb{R}^{N}$ if and only if all the following are true:\\
(1) $\mathcal{A}(\theta)= A(x),\mathcal{B}(\theta)= B(x)$, where $A,B$ are real parity-$(N\mod{2})$ polynomials in $x=\cos{(\theta)}$ of degree at most $N$;
\\
(2) $A(1)=1$;
\\
(3) $\forall x\in[-1,1]$, $A^2(x)+B^2(x)\le 1$;
\\
(4) $\forall x\ge 1$, $A^2(x)+B^2(x)\ge 1$;
\\
(5) $\forall L\;\text{even}, x\ge 0$, $A^2(ix)+B^2(ix)\ge 1$.
\\
Moreover, $\vec\phi\in\mathbb{R}^{N}$ can be computed in classical $\mathcal{O}(\text{poly}(N))$ time.
\end{lemma}
This automatically implies the following quantum signal processing result regarding Eq.~\ref{Eq:QSP_Baseline}.
\begin{lemma}[Quantum signal processing; adapted from~\cite{Low2016hamiltonian}]
\label{Thm:QSP_AB}
Given Hermitian standard-form-$(\hat{H},1,\hat{U},d)$, let any $A,B$ be degree $N$ polynomials that satisfy the conditions of Lem.~\ref{Lem:AchievableAB}.
Then there exists a standard-form-$(A[\hat{H}]+iB[\hat{H}],1,\hat{W}_{\vec\phi},2d)$, where $\hat{W}_{\vec\phi}$ requires $\mathcal{O}(N)$ queries to controlled-$\hat{U}$ and $\mathcal{O}(N\log(d))$ primitive quantum gates.
\end{lemma}
The many other partial specifications of $(\mathcal A,\mathcal B,\mathcal C,\mathcal D)$ described in~\cite{Low2016methodology} imply analogous constructions. Relevant to us are characterizations of achievable $(\mathcal B),(\mathcal C,\mathcal D),(\mathcal D)$ stated in Lems.\ref{Thm:AchievableB},\ref{Thm:AchievableCD},\ref{Lem:AchievableD}, respectively. These powerful tools reduce the problem of designing quantum circuits for arbitrary target functions $f[\hat{H}']$ to finding good polynomial approximations to $f(x)$ over the interval $x\in[-1,1]$, of which the optimal Hamiltonian simulation result $f[\hat{H}']=e^{-i\hat{H}'t}$ in Thm.~\ref{Thm:Ham_Sim_Qubitization} is an example. In the following, we focus the query complexity as any ancilla overhead will always be $\mathcal{O}(1)\le 3$ qubits, and the additional number of primitive gates required will typically be only a multiplicative factor $\mathcal{O}(\log{(d)})$ of the query complexity.

\subsection{Flexible Quantum Signal Processing}
\label{Sec:QSP_Constraint_Free}
Lem.~\ref{Thm:QSP_AB} would be more useful if we could drop the unintuitive constraints (4,5) that impose restriction on what the target functions must be \emph{outside} the domain of interest. In Thm.~\ref{Thm:QSP_B}, we present a generalization that computes functions with only one component $B[\hat{H}'] = (\bra{0}_{abc}\otimes\hat{I}_s)\hat{V}_{\vec{\phi}}(\ket{0}_{abc}\otimes\hat{I}_s)$ without those constraints, using an additional single-qubit ancilla register $c$. Note that this does not follow immediately from the discussion of Sec.~\ref{Sec:Standard-form_QSP_Prior_art} as the constraint $A(1)=1$ means there will always be some $A$ component, even if the characterizations of other partial specifications of $(A,B,C,D)$ are used. The trick is to exploit the structure of single-qubit rotations Eq.~\ref{Eq:QSP_Single_Qubit} to stage a perfect cancellation of the $A[\hat{H}']$ term by taking a linear combination of two standard-form encodings for $(\bra{0}_{ab}\otimes\hat{I}_s)\hat{V}_{\pm\vec{\phi}}(\ket{0}_{ab}\otimes\hat{I}_s)=A[\hat{H}']\pm iB[\hat{H}']$. 
\begin{proof}[Proof of Thm.~\ref{Thm:QSP_B}]
Consider the composite qubiterate in Eq.~\ref{Eq:QSP_Baseline} controlled by a single-qubit ancilla $c$. Let
\begin{align}
\hat{V}'_{\vec\phi}=-i\ket{1}\bra{0}_c\otimes\hat{W}_{\vec\phi}+i\ket{0}\bra{1}_c\otimes\hat{W}_{-\vec\phi}=(\hat{\sigma}_{y}\otimes\hat{I}_{abs})(\ket{0}\bra{0}_c\otimes\hat{W}_{\vec\phi}+\ket{1}\bra{1}_c\otimes\hat{W}_{-\vec\phi}).
\end{align}
Note that details in the construction of $\hat{W}_{\vec{\phi}}$ actually allow for the implementation of $\hat{V}'_{\vec\phi}$ with the same query complexity, as seen in Figure~\ref{Fig:Circuit_Qubitization_Flexible_QSP}.
By applying the similarity transformation $\hat{\sigma}_xe^{-i\hat{\sigma}_\phi\theta}\hat{\sigma}_x=e^{-i\hat{\sigma}_{-\phi}\theta}$, $\hat{\sigma}_x\hat{\sigma}_z\hat{\sigma}_x=-\hat{\sigma}_z$, and $\hat{\sigma}_x\hat{\sigma}_y\hat{\sigma}_x=-\hat{\sigma}_y$,
\begin{align}
\hat{W}_{-\vec{\phi}}\ket{0}_{ab}\ket{\lambda}_s&=e^{-i \hat{\sigma}_{-\phi_{N}}\theta_\lambda}e^{-i \hat{\sigma}_{-\phi_{N-1}}\theta_\lambda}\cdots e^{-i \hat{\sigma}_{-\phi_{1}}\theta_\lambda}\ket{0}_{ab}\ket{\lambda}_s, 
\\\nonumber
&= \left(A(\lambda)\hat{I}_\lambda-iB(\lambda)\hat\sigma_{z,\lambda} + iC(\lambda)\hat\sigma_{x,\lambda}-iD(\lambda)\hat\sigma_{y,\lambda}\right)\ket{0}_{ab}\ket{\lambda}_s.
\end{align}
Thus using the ancilla state $\ket{+}_c\ket{0}_{ab}$, where $\ket{\pm}=\frac{1}{\sqrt{2}}(\ket{0}+\ket{1})$, as the input to $\hat{V}'_{\vec\phi}$ results in:
\begin{align}
\label{Eq:Controlled_Generalized_Reflections}
\hat{V}'_{\vec\phi}\ket{+}_c\ket{0}_{ab}\ket{\lambda}_s&=
\left(-i A(\lambda)\ket{-}_c+B(\lambda)\ket{+}_c\right)\ket{0}_{ab}\ket{\lambda}_s+\left(C(\lambda)\ket{-}_c+D(\lambda)\ket{+}_c\right)\ket{0\lambda^\perp}_{abs}.
\end{align}
Thus $(\bra{+}_c\bra{0}_{ab}\otimes\hat{I}_s)\hat{V}'_{\vec{\phi}}(\ket{+}_c\ket{0}_{ab}\otimes\hat{I}_s)=B[\hat{H}']$ encodes $B[\hat{H}']$ in standard-form. Note that this is independent of all the other functions $A,C,D$ which are in general non-zero. Thus we may apply Lem.~\ref{Thm:AchievableB} on achievable $(B)$ even those all other components are in general non-zero. Finally, let $\hat{V}_{\vec\phi}=(\widehat{\text{Had}}\otimes\hat{I}_{abs})\hat{V}'_{\vec{\phi}}(\widehat{\text{Had}}\otimes\hat{I}_{abs})$.
\end{proof}
\begin{figure}[t]
\centering
\includegraphics{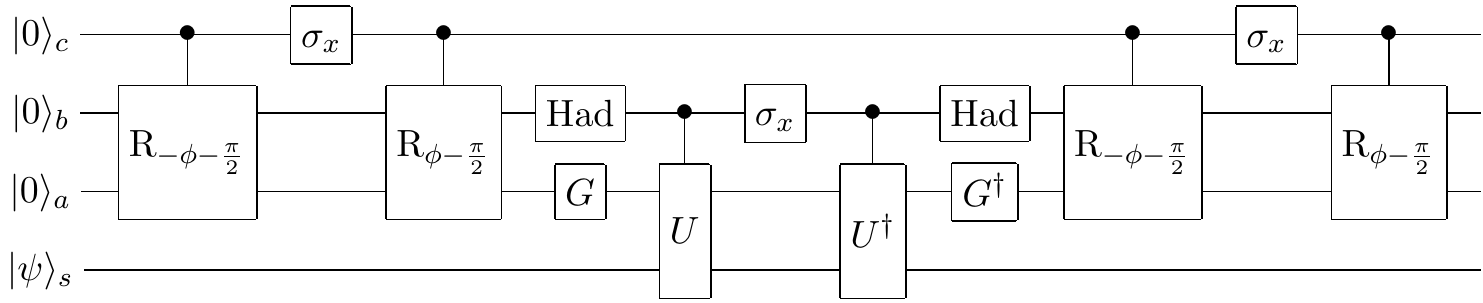}\\
\vspace{0.25cm}
\includegraphics{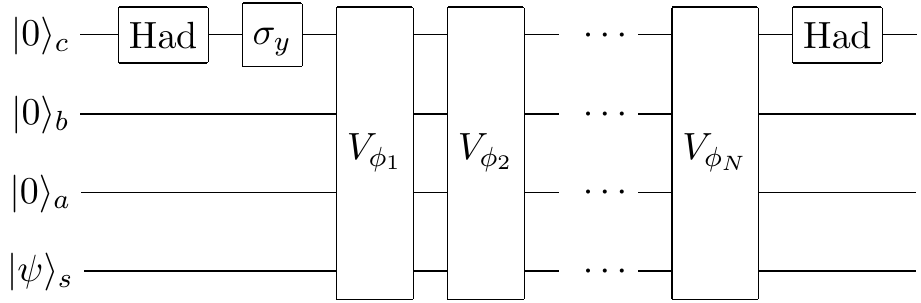}
\caption{
\label{Fig:Circuit_Qubitization_Flexible_QSP}(top) Circuit diagram for the flexible qubiterate $\hat{V}'_\phi=\ket{0}\bra{0}_c\otimes\hat{W}_{\phi}+\ket{1}\bra{1}_c\otimes\hat{W}_{-\phi}$, where $\hat{\text{R}}=\hat{I}_{ab}-(1-e^{-i\phi})\ket{0}\bra{0}_a\otimes\ket{0}\bra{0}_b$. (bottom) Circuit diagram for the flexible composite qubiterate $\hat{V}_{\vec{\phi}}$ used to encode a standard-form-$(B[\hat{H}],1,\hat{V}_{\vec\phi},4d)$. The query complexity of $\hat{V}_{\vec{\phi}}$ is $N$ to $\hat{G}$, controlled-$\hat{U}$, and their inverses. Its gate complexity is $\mathcal{O}(N\log{(d)})$.
}

\end{figure}
\begin{lemma}[Achievable $(\mathcal{B})$ -- Thm.~3.2 of~\cite{Low2016methodology}]
\label{Thm:AchievableB}
For any integer $N>0$, a choice of function $\mathcal{B}$ in Eq.~\ref{Eq:QSP_Single_Qubit} is achievable by some $\vec\phi\in\mathbb{R}^{N}$ if and only if all the following are true:\\
(1) $\mathcal B(\theta)= {B}(x)$, where ${B}$ is a real parity-$(N\mod{2})$ polynomial in $x=\cos{(\theta)}$ of degree at most $N$;
\\
(2) $B(0)=0$;
\\
(3) $\forall x\in[-1,1]$, $B^2(x)\le 1$.
\\
Moreover, $\vec\phi\in\mathbb{R}^{N}$ can be computed in classical $\mathcal{O}(\text{poly}(N))$ time.
\end{lemma}
With Thm.~\ref{Thm:QSP_B}, we are assured that any degree $N$ bounded matrix polynomial that goes to zero at the origin can be implemented exactly on a quantum computer using $\mathcal{O}(N)$ queries, $\mathcal{O}(N)$ additional primitive quantum gates, and $\mathcal{O}(1)$ additional ancilla qubits. 


\section{Uniform Spectral Amplification by Quantum Signal Processing}
\label{Sec:Uniform_Hamiltonian_Amplification}
When provided with no information on any structure in the standard-form encoding $(\bra{0}_a\otimes \hat{I}_s)\hat{U}(\ket{0}_a\otimes \hat{I}_s)=\hat{H}/\alpha$ of the Hermitian matrix $\hat{H}$, all we have is access to the signal oracle $\hat{U}$. Thus our only option is to apply quantum signal processing and study the polynomial functions $f[\cdot]$ of $\hat{H}/\alpha$ that achieve uniform spectral amplification.
In this setting, Thm.~\ref{Cor:Operator_Amplification} performs uniform spectral amplification, though the trade-off between its implementation cost and the achieved reduction of $\alpha$ provides no advantage to Hamiltonian simulation. However, a speedup is possible through Thm.~\ref{Thm:Ham_Encoding_Uniform_Amplification} when interested only in the lower energy subspace of $\hat{H}$.

As the normalization $\alpha$ is always greater or equal than $\|\hat{H}\|$, any input state $\ket{\psi}$ on the system has support only on eigenstates $\hat{H}/\alpha\ket{\lambda}=\lambda\ket{\lambda}$ with eigenvalues $|\lambda| \le \|\hat{H}\|/\alpha\le 1$. Given an upper bound $\Lambda \in [ \|\hat{H}\|,\alpha]$ on the spectral norm, this means that in any polynomial function $p(x)$ that we construct, only its restriction to the domain $x\in[-\Lambda/\alpha,\Lambda/\alpha]$ is of interest, so long as $|p(x)|$ remains bounded by $1$ over $x\in[-1,1]$. Thus one approach to minimizing the normalization is to use quantum signal processing to encode a polynomial with the property $p[\hat{H}/\alpha]\approx \frac{\hat{H}}{\Lambda}$ in standard-form. Thus, we should find a polynomial that approximates a truncated linear function, such as
\begin{align}
\label{Eq:Linear_target_function}
f_{\text{lin},\Gamma}(x)=
\begin{cases}
 \frac{x}{2\Gamma}, & |x| \in [0, \Gamma], \\
\in [-1,1], & |x| \in (\Gamma,1].
\end{cases}
\end{align}
In Thm.~\ref{Thm.Polynomial_LAA} of Appendix.~\ref{Sec:Polynomials_Amplitude_Multiplication}, we approximate $f_{\text{lin},\Gamma}(x)$ with a polynomial with the following properties:
$\forall\;\Gamma \in [0,1/2]$ and $\epsilon \le\mathcal{O}(\Gamma)$, the odd polynomial $p_{\text{lin},\Gamma,n}$ of degree $n=\mathcal{O}(\Gamma^{-1}\log{(1/\epsilon)})$ satisfies
\begin{align}
\forall\; {x\in[- \Gamma,\Gamma]},\; \left|p_{\text{lin},\Gamma,n}(x)- \frac{x}{2\Gamma}\right|\le \frac{\epsilon|x|}{2\Gamma} \quad\text{and}\quad \max_{x\in [-1,1]} |p_{\text{lin},\Gamma,n}(x)|\le 1.
\end{align}
This polynomial satisfies the conditions of flexible quantum signal processing in Thm.~\ref{Thm:QSP_B}, and provides us with the solution Thm.~\ref{Cor:Operator_Amplification} to uniform spectral amplification.
\begin{proof}[Proof of Thm.~\ref{Cor:Operator_Amplification}]
Given Hermitian standard-form-$(\hat{H},\alpha,\hat{U},d)$ and an upper bound $\Lambda\in[\|\hat{H}\|,\alpha]$, Define $\Gamma = \Lambda/\alpha\le 1$. Using Thm.~\ref{Thm:QSP_B} with the polynomial $p_{\text{lin},\Gamma,n}$, encode $p_{\text{lin},\Gamma,n}[\hat{H}/\alpha]\approx \frac{\hat{H}}{2\Gamma\alpha}=\frac{\hat{H}}{2\Lambda}$ in Hermitian standard-form-$(p_{\text{lin},\Gamma,n}[\hat{H}/\alpha],1,\hat{V},4d)$. This requires $\mathcal{O}(n)$ queries, and is identical to the Hermitian standard-form-$(2\Lambda p_{\text{lin},\Gamma,n}[\hat{H}/\alpha],2\Lambda ,\hat{V},4d)$. Define $\hat{H}_{\text{amp}}=2\Lambda p_{\text{lin},\Gamma,n}[\hat{H}/\alpha]$. Then the error of approximation $\left\|\frac{\hat{H}_{\text{amp}}}{2\Lambda}-\frac{\hat{H}}{2\Lambda}\right\|\le 
\max_{x\in[-\Lambda,\Lambda]} \left|p_{\text{lin},\Gamma,n}\left(\frac{x}{\alpha}\right)-\frac{x}{2\Lambda}\right|\le 
\max_{x\in[-\Gamma,\Gamma]} \left|p_{\text{lin},\Gamma,n}(x)-\frac{x}{\Gamma}\right|\le \frac{\epsilon_1}{2}$. Finally, note that $p_{\text{lin},\Gamma,n}$ requires $\epsilon_1\le\mathcal{O}(\Gamma)$, and has degree scaling like $n=\mathcal{O}(\Gamma^{-1}\log{(1/\epsilon_1)})$, so let us define $\epsilon = \frac{\epsilon_1}{2}$. 
\end{proof}

Unfortunately, this provides absolutely no advantage to Hamiltonian simulation as the decrease in normalization by factor $\alpha/\Lambda$ is exactly balanced by an increase in query complexity by factor $\alpha/\Lambda$. Nevertheless,  Thm.~\ref{Cor:Operator_Amplification} may be of use to applications involving measurement such as quantum metrology and repeat-until-success circuits, as the success probability $\|\frac{\hat{H}}{\Lambda}\|^2$ is improved by a quadratic factor $(\alpha/\Lambda)^2$. This is analogous to oblivious amplitude amplification which only applies to matrices that are approximately unitary~\cite{Berry2014}.

One workable possibility is highlighted by the deep connection between quantum signal processing and the properties of polynomials. Thm.~\ref{Cor:Operator_Amplification} uses a degree $\mathcal{O}(\Lambda^{-1})$ polynomial with maximum gradient $\mathcal{O}(\Lambda^{-1})$. Yet a famous inequality by Markov indicates a best-case quadratic advantage in the gradient $p'$ of any degree $n$ polynomial $\max_{x\in[-1,1]}|p'(x)|\le n^2\max_{x\in[-1,1]}|p(x)|$. Thus we have not fully exhausted the capabilities of polynomials. As this inequality becomes an equality for Chebyshev polynomials of the first kind $T_L(x)=\cos{(L \cos^{-1}{(x)})}$ at $x=\pm1$, this suggests that a speedup is possible if we are only concerned with time evolution on eigenstates with eigenvalues $|\lambda| \in [1-\Delta,1]$ where $\Delta\ll 1$. With this assumption, we may prove Thm.~\ref{Thm:Ham_Encoding_Uniform_Amplification}.
\begin{proof}[Proof of Thm.~\ref{Thm:Ham_Encoding_Uniform_Amplification}]
Consider the truncated linear function
\begin{align}
\label{Eq:Linear_target_function}
f_{\text{gap},\Delta}(x)=
\begin{cases}
\frac{x+1-\Delta}{\Delta}, & x \in [-1, -1+\Delta], \\
\in[-1,1], & \text{otherwise}.
\end{cases}
\end{align}
As $\hat{\Pi}(f_{\text{gap},\Delta}[\frac{\hat{H}}{\alpha}]-\frac{\hat{H}+\alpha\hat{I}(1-\Delta)}{\Delta\alpha})\hat{\Pi}=0$, the theorem is proven by finding degree $n$ odd polynomial $p_{\text{gap},\Delta,n}(x)$ that uniformly approximates $f_{\text{gap},\Delta}(x)$ with error $\max_{x\in[-1,-1+\Delta]}|p_{\text{gap},\Delta,n}(x)-f_{\text{gap},\Delta}(x)|\le\epsilon$ and also satisfies all the conditions of quantum signal processing Thm.~\ref{Thm:QSP_B}. We provide such a polynomial of degree $\mathcal{O}(\Delta^{-1/2}\log^{3/2}{(\frac{1}{\Delta\epsilon})})$ in Lem.~\ref{Lem.Polynomial_gapped_linear} of Appendix.\ref{Sec:Polynomials_Low_energy}. And so we define $\frac{\hat{H}_{\text{amp}}}{\Delta\alpha} = p_{\text{gap},\Delta,n}[\frac{\hat{H}}{\alpha}]$, which approximates the desired amplified Hamiltonian with error $\|\hat{\Pi}(\frac{\hat{H}_{\text{amp}}}{\Delta\alpha}-\frac{\hat{H}+\alpha\hat{I}(1-\Delta)}{\Delta\alpha})\hat{\Pi}\|
\le \max_{x\in[-1,-1+\Delta]}|p_{\text{gap},\Delta,n}(x)-\frac{ x+(1-\Delta)}{\Delta}|\le \epsilon$.
\end{proof}
As energy gaps in an interval of width $\Delta$ are stretched by factor $\Delta^{-1}$ using only $\mathcal{O}(\Delta^{1/2})$ queries, a quadratic advantage in normalization is achieved. This is essentially spectral gap amplification~\cite{Somma2013SpectralGap} with two important distinctions: first, it applies to any Hamiltonian through the standard-form, though as highlighted in~\cite{Somma2013SpectralGap}, only those encoded with $\alpha=\|\hat{H}\|$, such as frustration-free Hamiltonians, can fully exploit the effect. Second, it amplifies the spectral gap of all eigenvalues uniformly, rather than non-uniformly. By combining with Thm.~\ref{Thm:Ham_Sim_Qubitization}, one obtains a Hamiltonian simulation algorithm for low-energy subspaces, relevant to quantum chemistry and adiabatic computation. 
\begin{corollary}[Hamiltonian simulation of low-energy subspaces]
\label{Cor:Ham_Sim_Spectral_Amplification}
Given Hermitian standard-form-$(\hat{H},\alpha,\hat{U},d)$ with eigenstates $\hat{H}/\alpha\ket{\lambda}=\lambda\ket{\lambda}$, let $\Delta \in(0,1)$ be a positive constant, and $\hat{\Pi}=\sum_{\lambda \in[-1,-1+\Delta]}\ket{\lambda}\bra{\lambda}$ be a projector onto the low-energy subspace of $\hat{H}$. Then time-evolution $e^{-i\hat{H}t}$ on eigenstates with eigenvalues $\lambda \in [-1,-1+\Delta]$ can be approximated with error $\epsilon$ using $\mathcal{O}(t\alpha\sqrt{\Delta}\log^{3/2}{(\frac{t\alpha}{\epsilon})}+\Delta^{-1/2}\log^{5/2}{(\frac{t\alpha}{\epsilon})})$ queries to controlled-$\hat{U}$.
\end{corollary}
\begin{proof}
This follows from multiplying the query complexities of Thm.~\ref{Thm:Ham_Sim_Qubitization} with Thm.~\ref{Thm:Ham_Encoding_Uniform_Amplification}, similar to the proof of Cor.~\ref{Cor:HamExponentials}, to obtain a cost of $\mathcal{O}\left(t\alpha\Delta+\frac{\log{(1/\epsilon_1)}}{\log\log{(1/\epsilon_1)}}\right)\mathcal{O}(\Delta^{-1/2}\log^{3/2}{(\frac{1}{\Delta\epsilon_2})})$ queries for approximating $e^{-i\hat{H}t}$ with error $\epsilon_1 + t\alpha\Delta \epsilon_2$. Thus we choose $\epsilon_1=\epsilon/2$ and $t\alpha\Delta \epsilon_2=\epsilon/2$.
\end{proof}
It is worth mentioning that Thm.~\ref{Thm:Ham_Encoding_Uniform_Amplification} also performs uniform spectral amplification on \emph{high} energy states. This follows from the polynomial $p_{\text{gap},\Delta,n}(x)$ being odd. Thus its ability to stretch eigenvalues $\lambda\in[-1,-1+\Delta]$ applies to those $\lambda\in[1-\Delta,1]$ as well.

\section{Amplitude Amplification Techniques}
\label{Sec:AA_by_QSP}
Amplitude amplification is a staple quantum subroutine for state preparation that used in many quantum algorithms. The basic version, is based on reflections, in described in Sec.~\ref{Sec:Amplitude_Amplification}. The most common generalization of amplitude amplification replaces the reflection with partial reflections. This allows for constructing more interesting variations in the final state amplitude as a function of the initial state amplitudes, though a systematic approach to designing these variations is not known to prior art. We show in Sec.~\ref{Sec:AA_partial_ref} that these functions are polynomials subject to certain constraints and solve the design problem through Lem.~\ref{Thm:Generalized_Amplitude_Amplification}. We then generalize this in Sec.~\ref{Sec:Flexible_AA} to obtain the flexible amplitude amplification Thm.~\ref{Thm:Controlled_Generalized_Amplitude_Amplification} that relaxes some constraints on these polynomials. In in Sec.~\ref{Sec:Amp_Mult}, an application of flexible amplitude amplification with a particular choice of polynomials yields the amplitude multiplication Thm.~\ref{Thm:Linear_Amplitude_Amplification}.
\subsection{Amplitude Amplification}
\label{Sec:Amplitude_Amplification}
Amplitude amplification is a quantum algorithm for state preparation.
Suppose the state creation operator $\hat{G}$ prepares the start state $\ket{s}=\hat{G}\ket{0}\in\mathbb C^d$ from the computational basis. The start state has overlap $\sin{(\theta)}=\langle t\ket{s}$ with the target state $\ket{t}$ is thus
\begin{align}
\ket{s}=\sin{(\theta)}\ket{t}+\cos{(\theta)}\ket{t^\perp}, \quad \langle t\ket{t^{\perp}}=0,
\end{align}
and the goal is to prepare the state $\ket{t}$. 

The standard solution to this problem boosts the amplitude $\sin{(\theta)}$ of $\ket{t}$ to $\mathcal{O}(1)$. This requires access to two oracles that perform reflections about $\ket{s},\ket{t}$ respectively:
\begin{align}
\widehat{\text{Ref}}_{\ket{s}}=\hat{I}-2\ket{s}\bra{s}=\hat{G}(\hat{I}-2\ket{0}\bra{0})\hat{G}^\dag=\hat{G}\widehat{\text{Ref}}_{\ket{0}}\hat{G}^\dag,\quad \widehat{\text{Ref}}_{\ket{t}}=\hat{I}-2\ket{t}\bra{t}.
\end{align}
As $(\hat{I}-2\ket{0}\bra{0})$ is a conditional phase gate, it may be implemented with $\mathcal{O}(\log(d))$ primitive quantum gates. The cost of implementing reflections about an arbitrary target state $\widehat{\text{Ref}}_{\ket{t}}$ it not always as straightforward. However, this cost is typically built into a definition of $\hat{G}$ that marks the target state with a single flag qubit subscripted by $b$. In other words. $\hat{G}\ket{0}_a\ket{0}_b=\sin{(\theta)}\ket{t}_a\ket{0}_b+\cos{(\theta)}\ket{t^\perp}_{ab}$. By defining the new target state as $\ket{t}_a\ket{0}_b$, a reflection about $\ket{t}_a\ket{0}_b$ may be constructed with a single $\hat{I}_a\otimes\hat{\sigma}_z$  gate.

The product $\widehat{\text{Ref}}_{\ket{s}}\widehat{\text{Ref}}_{\ket{t}}$, with query cost $2$, is known as the Grover iterate, and it easily shown that 
\begin{align}
\ket{s}=
\left(
\begin{matrix}
\cos{(\theta)} \\ 
\sin{(\theta)}
\end{matrix}
\right),
\quad
\widehat{\text{Ref}}_{\ket{s}}\widehat{\text{Ref}}_{\ket{t}}
=\left(
\begin{matrix}
\cos{(2\theta)} & -\sin{(2\theta)} \\ 
\sin{(2\theta)} & \cos{(2\theta)}
\end{matrix}
\right),
\end{align}
in the $\{\ket{t}^\perp,\ket{t}\}$ basis. Thus we obtain the well-known result 
\begin{align}
\label{Eq:RegularAmplitudeAmplification}
(\widehat{\text{Ref}}_{\ket{s}}\widehat{\text{Ref}}_{\ket{t}})^N \ket{s} 
=
\sin{\left((2N+1)\theta\right)}\ket{t}+\cos{\left((2N+1)\theta\right)}\ket{t^\perp}.
\end{align}
By choosing $N = \lceil \frac{\pi}{4\theta}-\frac{1}{2}\rceil =\mathcal{O}(1/\theta)$ repetitions, $\bra{t}(\widehat{\text{Ref}}_{\ket{s}}\widehat{\text{Ref}}_{\ket{t}})^N \hat{G}\ket{0}=\mathcal{O}(1)$ as desired with $Q=2N+1=\mathcal{O}(1/\theta)$ queries.

\subsection{Amplitude Amplification by Partial Reflections}
\label{Sec:AA_partial_ref}
The more general phase matching technique~\cite{Long1999PhaseMatching} applies partial reflections parameterized by phases $\alpha,\beta$:
\begin{align}
\label{Eq:Generalized_Reflections}
\widehat{\text{Ref}}_{\alpha,\ket{s}}=\hat{I}-(1-e^{-i \alpha})\ket{s}\bra{s},
\quad 
\widehat{\text{Ref}}_{\beta,\ket{t}}=\hat{I}-(1-e^{-i \beta})\ket{t}\bra{t},
\end{align}
and the generalized Grover iterate is then $\widehat{\text{Ref}}_{\alpha,\ket{s}}\widehat{\text{Ref}}_{\beta,\ket{t}}$ which has query cost $2$. An $N=2n+1$ query sequence of these iterates produces the state
\begin{align}
\label{Eq:Generalized_Sequence}
\prod^{n}_{k=1}\widehat{\text{Ref}}_{\alpha_k,\ket{s}}\widehat{\text{Ref}}_{\beta_k,\ket{t}} \ket{s}
=
(i\mathcal C(\theta)+\mathcal D(\theta))\ket{t}+(\mathcal A(\theta)-i\mathcal B(\theta))\ket{t^\perp},
\end{align}
where $\widehat{\text{Ref}}_{\alpha_1,\ket{s}}\widehat{\text{Ref}}_{\beta_1,\ket{t}}$ acts first on the input, and $\mathcal A,\mathcal B,\mathcal C,\mathcal D$ are real functions parameterized by $\vec{\alpha},\vec{\beta}$. Unfortunately, the dependence of $\vec{\alpha},\vec{\beta}$ on any arbitrary choice of $\mathcal A,\mathcal B,\mathcal C,\mathcal D$ appears quite mysterious. Only in very few cases can the $\mathcal A,\mathcal B,\mathcal C,\mathcal D$ can be specified for arbitrary $N$ and then inverted to obtain a consistent set of $\vec{\alpha},\vec{\beta}$ in closed-form~\cite{Yoder2014}. For instance, standard amplitude amplification corresponds to $\alpha_k=\beta_k=\pi$. 

We resolve this mystery by proving the following result
\begin{lemma}[Amplitude amplification with partial reflections]
\label{Thm:Generalized_Amplitude_Amplification}
Given a state preparation unitary $\hat{G}$ acting on the computational basis states $\ket{0}_a\in \mathbb{C}^d$, $\ket{0}_b\in \mathbb{C}^2$ such that $\hat{G}\ket{0}_a\ket{0}_b=\lambda\ket{t}_a\ket{0}_b+\sqrt{1-\lambda^2}\ket{t^\perp}_{ab}$, where $\ket{t^\perp}_{ab}$ has no support on $\ket{0}_b$, let $C,D$ be any two functions that satisfies all the following conditions:
\\
(1) $C,D$, where are odd real polynomials in $\lambda$ of degree at most $2N+1$;
\\
(2) $\forall \lambda\in[-1,1]$, $ {C}^2(\lambda)+ {D}^2(\lambda)\le 1$;
\\
(3) $\forall \lambda\ge 1$, $ {C}^2(\lambda)+ {D}^2(\lambda)\ge 1$,
\\
Then there exists a quantum circuit $\hat{V}_{\vec\phi}$ such that
$\bra{t}_a\bra{0}_b\hat{V}_{\vec\phi}\ket{0}_a\ket{0}_b=i C(\lambda)+D(\lambda)$, using $N+1$ queries to $\hat{G}$, $N$ queries to $\hat{G}^\dag$, and $\mathcal{O}(n\log{(d)})$ primitive quantum gates pre-computed from $C,D$ in classical $\mathcal{O}(\text{poly}(N))$ time.
\end{lemma}

This result is quite remarkable as the constraints are lax and allow for many interesting functions. For instance, choosing $ {C}(y) = \pm T_{2N+1}(y) = \sin{((2N+1)\theta)}$ to be Chebyshev polynomials of the first kind and ${D}(y)=0$, recovers the baseline amplitude amplification algorithm. 

The application of Lem.~\ref{Thm:Generalized_Amplitude_Amplification} requires finding a good polynomial approximation, say ${D}$ to the target function. However, it is not always clear how constraint (3) on properties of the polynomial outside the interval of interest may always be satisfied. We rectify this in Thm.~\ref{Thm:Controlled_Generalized_Amplitude_Amplification} by adding an additional ancilla qubit to stage a cancellation of the $C$ term, similar to the proof of Thm.~\ref{Thm:QSP_B}.
Subject only to parity and being bounded, we can implement without approximation any arbitrary polynomial of degree exactly equal to the number of queries to the state preparation operator $\hat{G}$. This enables us to compute any real function with a query complexity exactly that of the its best polynomial approximations thus allowing us to transfer powerful results from approximation theory~\cite{Meinardus1967} to quantum computation.
\begin{proof}[Proof of Lem.~\ref{Thm:Generalized_Amplitude_Amplification}.]
Our starting point is $Q=2N+1$ query sequence of Eq.~\ref{Eq:Generalized_Sequence}. We define $\hat{G}$ to mark the target state with an ancilla flag qubit $b$ e.g. $\ket{t}\rightarrow \ket{t}_a\ket{0}_b$, $\ket{t^\perp}\rightarrow \ket{t^\perp}_{ab}$, where $\ket{t^\perp}_{ab}$ has no support on $\ket{0}_b$ This allows us to perform partial reflections about $\ket{t}$ using single-qubit phase gates. Let us re-express the generalized reflection in Eq.~\ref{Eq:Generalized_Reflections} as:
\begin{align}
\widehat{\text{Ref}}_{\alpha,\ket{s}}&=\hat{I}_{ab}-(1-e^{-i \alpha})\hat{G}\ket{0}\bra{0}_{ab}\hat{G}^{\dag}
=
\hat{G}\left(\hat{I}_{ab}-(1-e^{-i \alpha})\ket{0}\bra{0}_{ab}\right)\hat{G}^{\dag}
=
\hat{G}\widehat{\text{Ref}}_{\alpha,\ket{0}}\hat{G}^{\dag}.
\end{align}
If $\ket{0}_{ab}$ is of dimension $2d$, $\widehat{\text{Ref}}_{\alpha,\ket{0}}$ is a conditional phase gate and may be implemented with $\mathcal{O}(\log(d))$ primitive gates.
As $\text{span}\{\ket{t}_a\ket{0}_b,\ket{t^\perp}_{ab}\}$ is an invariant subspace of $\widehat{\text{Ref}}_{\alpha,\ket{s}}\widehat{\text{Ref}}_{\beta,\ket{t}}$, we may represent it equivalently with Pauli matrices $\hat{\sigma}_{x,y,z}$ through the replacements
\begin{align}
\hat{G}&\rightarrow e^{-i \hat{\sigma}_y\theta}=\left(
\begin{matrix}
\cos{(\theta)} & -\sin{(\theta)} \\ 
\sin{(\theta)} & \cos{(\theta)}
\end{matrix}
\right),
\quad 
\hat{G}^{\dag}
\rightarrow 
e^{i \hat{\sigma}_y\theta}
=\left(
\begin{matrix}
\cos{(\theta)} & \sin{(\theta)} \\ 
-\sin{(\theta)} & \cos{(\theta)}
\end{matrix}
\right),
\\ \nonumber
e^{i\alpha/2}\widehat{\text{Ref}}_{\alpha,\ket{0}}&\rightarrow e^{-i \hat{\sigma}_z\alpha/2}
=
\left(
\begin{matrix}
e^{i\alpha/2} & 0 \\ 
0 & e^{-i\alpha/2}
\end{matrix}
\right),
\quad 
e^{i\beta/2}\widehat{\text{Ref}}_{\beta,\ket{t}}\rightarrow e^{-i \hat{\sigma}_z\beta/2}
=
\left(
\begin{matrix}
e^{i\beta/2} & 0 \\ 
0 & e^{-i\beta/2}
\end{matrix}
\right).
\end{align}
Thus $\widehat{\text{Ref}}_{\alpha,\ket{s}}\widehat{\text{Ref}}_{\beta,\ket{t}}=e^{-i(\alpha+\beta)/2}e^{i \hat{\sigma}_y\theta}e^{-i \hat{\sigma}_z\alpha/2}e^{-i \hat{\sigma}_y\theta}e^{-i \hat{\sigma}_z\beta/2}$ in this subspace. Though applying $\hat{G}^{\dag}$ in general takes us out of the subspace, this operator is  always paired with $\hat{G}$ in the Grover iterate and never occurs in isolation -- the representation is faithful. This sequence of alternating $\hat\sigma_{y,z}$ rotations motivate us to define the operator for rotations by angle $\theta$ about an axis in the $\hat\sigma_x$--$\hat\sigma_y$ plane of the Bloch sphere:
\begin{align}
\label{Eq:BlochSphereXYRotation}
e^{-i \hat{\sigma}_\phi\theta}&=e^{-i \hat{\sigma}_z(\pi/2+\phi)/2}e^{-i \hat{\sigma}_y\theta}e^{i \hat{\sigma}_z(\pi/2+\phi)/2}=\left(
\begin{matrix}
\cos{(\theta)} & -i e^{-i\phi}\sin{(\theta)} \\ 
-i e^{i\phi}\sin{(\theta)} & \cos{(\theta)}
\end{matrix}
\right),
\end{align}
where $\hat{\sigma}_\phi=\cos{(\phi)}\hat\sigma_x+\sin{(\phi)}\hat\sigma_y$. We would like to express Eq.~\ref{Eq:Generalized_Sequence} as a product of just these $Q=2N+1$ rotations $e^{-i \hat{\sigma}_{\phi_k}\theta}$. 
Thus we replace the input state $\hat{G}\ket{0}_{ab}=\hat{G}e^{i\alpha_0}\widehat{\text{Ref}}_{\alpha_0,\ket{0}}\ket{0}_{ab}$, and obtain
\begin{align}
\hat{V}_{\vec\alpha,\vec\beta}&=e^{i\alpha_0}\left(\prod^{N}_{k=1}\widehat{\text{Ref}}_{\alpha_k,\ket{s}}\widehat{\text{Ref}}_{\beta_k,\ket{t}} \right)\hat{G}\widehat{\text{Ref}}_{\alpha_0,\ket{0}}.
\end{align}
Promised that $\hat{V}_{\vec\alpha,\vec\beta}$ always acts on input state $\ket{0}_{ab}$, the fact $\hat{G}\ket{0}=e^{-i \hat{\sigma}_y\theta}\ket{t^\perp}$ permits the representation.
\begin{align}
\hat{V}_{\vec\alpha,\vec\beta}
&=
e^{i\alpha_0/2-i\sum^n_{k=1}(\alpha_k+\beta_k)/2}\left(\prod^{N}_{k=1}e^{i \hat{\sigma}_y\theta}e^{-i \hat{\sigma}_z\alpha_k/2}e^{-i \hat{\sigma}_y\theta}e^{-i \hat{\sigma}_z\beta_k/2}\right)e^{-i \hat{\sigma}_y\theta}e^{-i \hat{\sigma}_z\alpha_0/2}.
\end{align}
Since we have the identity $e^{i \hat{\sigma}_y\theta}=e^{-i \hat{\sigma}_z \pi}e^{-i \hat{\sigma}_y\theta}e^{i \hat{\sigma}_z \pi}$, and all $e^{-i \hat{\sigma}_y}$ in Eq.~\ref{Eq:Generalized_Reflections} are sandwiched between $\hat\sigma_z$ rotations, we replace these with the $\hat\sigma_x$--$\hat\sigma_y$ rotations of Eq.~\ref{Eq:BlochSphereXYRotation} and define the composite iterate $\hat{V}_{\vec\phi}$ in Fig.~\ref{Fig:Circuit_AmpAmp_QSP}
\begin{align}
\label{Eq:Composite_Iterate}
\hat{V}_{\vec\phi}=e^{i\Phi}\hat{V}_{\vec\alpha,\vec\beta}=\left(\prod^{2N+1}_{k=1}e^{-i \hat{\sigma}_{\phi_k}\theta}\right)=\mathcal A(\theta)\hat{I}+i\mathcal B(\theta)\hat\sigma_z+i\mathcal C(\theta)\hat\sigma_x+i\mathcal D(\theta)\hat\sigma_y,
\end{align}
where $\Phi$, which depends only on $\vec\alpha,\vec\beta$, is chosen to cancel the global phase of $\hat{V}_{\vec\alpha,\vec\beta}$, $\vec{\phi}$ depends linearly on $\vec{\alpha},\vec{\beta}$, and the decomposition into the Pauli basis is always possible for $\text{SU}(2)$ matrices. 

By replacing the product of two-parameters generalized Grover iterates in Eq.~\ref{Eq:Generalized_Sequence} with a product of more fundamental and simpler one-parameter single-qubit rotations in Eq.~\ref{Eq:Composite_Iterate}, the structure underlying generalized amplitude amplification is made clearer. As these single-qubit rotations isomorphic to those considered in quantum signal processing Eq.~\ref{Eq:QSP_Single_Qubit}, we may apply Lem.~\ref{Thm:AchievableCD} that characterizes any achievable $(\mathcal D)$. Other choices from~\cite{Low2016methodology} such as $(\mathcal A,\mathcal B)$, $(\mathcal A,\mathcal C)$ etc. are also possible.
\end{proof}
\begin{lemma}[Achievable $(\mathcal{C},\mathcal{D})$ -- Thm.~2.4 of~\cite{Low2016methodology}]
\label{Thm:AchievableCD}
For any odd integer $N>0$, a choice of functions $\mathcal C,\mathcal D$ in Eq.~\ref{Eq:QSP_Single_Qubit} is achievable by some $\vec\phi\in\mathbb{R}^{N}$ if and only if all the following are true:\\
(1) $\mathcal{C}(\theta)= C(y),\mathcal{D}(\theta)= D(y)$, where $C,D$ are odd real polynomials in $y=\sin{(\theta)}$ of degree at most $N$;
\\
(2) $\forall y\in[-1,1]$, $\mathcal{C}^2(y)+\mathcal{D}^2(y)\le 1$;
\\
(3) $\forall y\ge 1$, $\mathcal{C}^2(y)+\mathcal{D}^2(y)\ge 1$.
\\
Moreover, $\vec\phi\in\mathbb{R}^{N}$ can be computed in classical $\mathcal{O}(\text{poly}(N))$ time.
\end{lemma}

\begin{figure}[t]
\begin{tabular}{l}
v\includegraphics{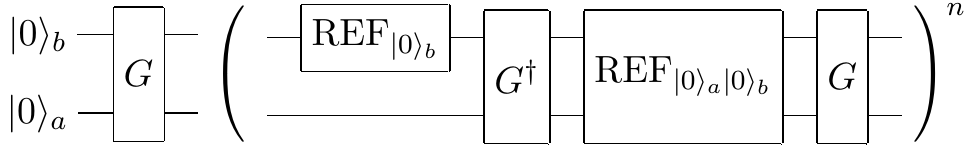}\\
\includegraphics{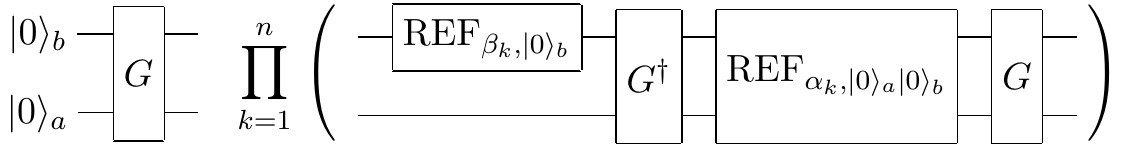}\\
\includegraphics{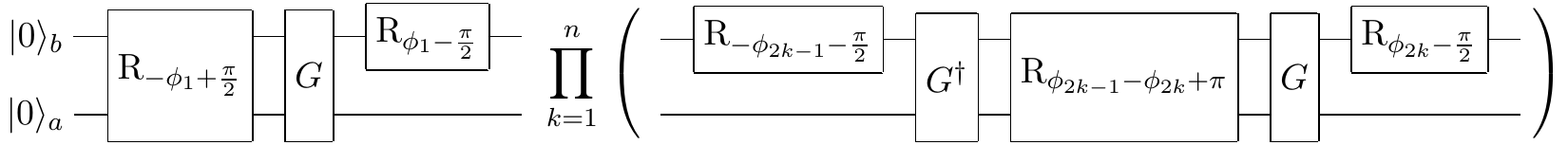}
\end{tabular}
\caption{
\label{Fig:Circuit_AmpAmp_QSP}(top) Circuit diagram for amplitude amplification. (middle) Circuit diagram for ampliutude amplification by phase-matching. (bottom) Circuit diagram for amplitude amplification $\hat{V}_{\vec\phi}$ by quantum signal processing. Note that we abbreviate the reflection operators as $\hat{\text{R}}$ and drop the state subscript here. The query complexity in all cases is $N=2n+1$, and the gate complexity is $\mathcal{O}(N\log{(d)})$.
}
\end{figure}

\subsection{Flexible Amplitude Amplification}
\label{Sec:Flexible_AA}
By taking a superposition of the state prepared by Lem.~\ref{Thm:Generalized_Amplitude_Amplification}, we may stage a cancellation of $\mathcal{C}$ function on the target state in Thm.~\ref{Thm:Controlled_Generalized_Amplitude_Amplification}. This allows us to prepare states with amplitudes dictated only by $\mathcal{D}$.
\begin{proof}[Proof of Thm.~\ref{Thm:Controlled_Generalized_Amplitude_Amplification}.]
Consider the composite iterate in Eq.~\ref{Eq:Composite_Iterate} controlled by a single-qubit ancilla register indexed by subscript $c$.
\begin{align}
\hat{W}_{\vec\phi}=\hat{V}_{\vec\phi}\otimes \ket{+}\bra{+}_c+\hat{V}_{\pi-\vec\phi}\otimes \ket{-}\bra{-}_c,
\end{align}
where $\ket{\pm}=\frac{1}{\sqrt{2}}(\ket{0}\pm\ket{1})$. Note that this can be implemented by controlling $\widehat{\text{Ref}}_{\alpha,\ket{0}},\widehat{\text{Ref}}_{\beta,\ket{t}}$ in  $\hat{V}_{\vec\phi}$. The number of queries to $\hat{G},\hat{G}^\dag$ is unchanged and $\hat{G},\hat{G}^\dag$ need not be controlled unitaries. Thus $\hat{W}_{\vec\phi}$ still has query complexity $N=2n+1$ equal to $\hat{V}_{\vec\phi}$. From the similarity transformation $\hat{\sigma_y}e^{-i\hat{\sigma}_\phi\theta}\hat{\sigma_y}=e^{-i\hat{\sigma}_{\pi-\phi}\theta}$,
\begin{align}
\hat{V}_{\vec\pi-\vec\phi}=\mathcal A(\theta)\hat{I}-i\mathcal B(\theta)\hat\sigma_z-i\mathcal C(\theta)\hat\sigma_x+i\mathcal D(\theta)\hat\sigma_y,
\end{align}
where $\vec\pi$ is the vector where all elements are $\pi$. 
This allows us to stage a cancellation of $\mathcal C$ when $\hat{W}_{\vec\phi}$ is controlled by the ancilla state $\ket{0}_c$:
\begin{align}
\label{Eq:Controlled_Generalized_Reflections}
\hat{W}_{\vec\phi}\ket{0}_a\ket{0}_b\ket{0}_c&=
\mathcal D(\theta)\ket{t}_a\ket{0}_b\ket{0}_c+\mathcal A(\theta)\ket{t^\perp}_{ab}\ket{0}_c+i \mathcal C(\theta)\ket{t}_a\ket{0}_b\ket{1}_c-i \mathcal B(\theta)\ket{t^\perp}_{ab}\ket{1}_c.
\end{align}
where $\ket{t}_a\ket{0}_b\ket{0}_c$ is our new target state that is uniquely marked by $\ket{0}_b\ket{0}_c$. Thus the amplitude of $\mathcal D$ on the target state is completely independent of $\mathcal A,\mathcal B,\mathcal C$ regardless of what they may be. This allows us to directly apply the following result for achievable $\mathcal D$ in Lem.~\ref{Lem:AchievableD}. 
\end{proof}
\begin{lemma}[Achievable $(\mathcal{D})$ -- Thm.~3.4 of~\cite{Low2016methodology}]
\label{Lem:AchievableD}
For any odd integer $N>0$, a choice of function $D$ in Eq.~\ref{Eq:Controlled_Generalized_Reflections} is achievable by some $\vec\phi\in\mathbb{R}^{N}$ if and only if all the following are true:\\
(1) $\mathcal D(\theta)= {D}(y)$, where ${D}$ is an odd real polynomial in $y=\sin{(\theta)}$ of degree at most $N$;
\\
(2) $\forall y\in[-1,1]$, $\mathcal{D}^2(y)\le 1$.
\\
Moreover, $\vec\phi\in\mathbb{R}^{N}$ can be computed in classical $\mathcal{O}(\text{poly}(N))$ time.
\end{lemma}

\subsection{Amplitude Multiplication}
\label{Sec:Amp_Mult}
The proof of amplitude multiplication follows from flexible amplitude amplification by an appropriate choice of polynomials for $D$.
\begin{proof}[Proof of Thm.~\ref{Thm:Linear_Amplitude_Amplification}]
The amplitude multiplication algorithm is a special case of Thm.~\ref{Thm:Controlled_Generalized_Amplitude_Amplification} where $D$ is a polynomial that approximates the truncated linear function
\begin{align}
\label{Eq:Linear_target_functionB}
f_{\text{lin},\Gamma}(x)=
\begin{cases}
 \frac{x}{2\Gamma}, & |x| \in [0, \Gamma], \\
\in [-1,1], & |x| \in (\Gamma,1].
\end{cases}
\end{align}
In Thm.~\ref{Thm.Polynomial_LAA} of Appendix.~\ref{Sec:Polynomials_Amplitude_Multiplication}, we approximate $f_{\text{lin},\Gamma}(x)$ with a polynomial with the following properties:
$\forall\;\Gamma \in [0,1/2]$ and $\epsilon \le\mathcal{O}(\Gamma)$, the odd polynomial $p_{\text{lin},\Gamma,n}$ of degree $n=\mathcal{O}(\Gamma^{-1}\log{(1/\epsilon)})$ satisfies
\begin{align}
\forall\; {x\in[- \Gamma,\Gamma]},\; \left|p_{\text{lin},\Gamma,n}(x)- \frac{x}{2\Gamma}\right|\le \frac{\epsilon|x|}{2\Gamma} \quad\text{and}\quad \max_{x\in [-1,1]} |p_{\text{lin},\Gamma,n}(x)|\le 1.
\end{align}
As this polynomial satisfies the conditions of Thm.~\ref{Thm:Controlled_Generalized_Amplitude_Amplification}, there exists a state preparation unitary 
$
\hat{W}_{\vec\phi}\ket{0}_a\ket{0}_b\ket{0}_c =  p_{\text{lin},\Gamma,n}(y)\ket{t}_a\ket{0}_b\ket{0}_c+\mathcal A(\theta)\ket{t^\perp}_{ab}\ket{0}_c+i \mathcal C(\theta)\ket{t}_a\ket{0}_b\ket{1}_c-i \mathcal B(\theta)\ket{t^\perp}_{ab}\ket{1}_c,
$
where the functions $\mathcal{A},\mathcal{B},\mathcal{C}$ of lesser interest, that consists of $\mathcal{O}(n)$ queries to $\hat{G},\hat{G}^\dag$ and $\mathcal{O}(n\log{(d)})$ primitive gates. Assuming that $\Gamma\in[|\sin{(\theta)}|,1/2]$ is an upper bound on $|\sin{(\theta)}|$,  the amplitude in the target state is $|\bra{t}_a\bra{0}_b\bra{0}_c\hat{W}_{\vec\phi}\ket{0}_a\ket{0}_b\ket{0}_c- \frac{\sin{(\theta)}}{2\Gamma}|\le\frac{\epsilon|\sin{(\theta)}|}{2\Gamma}$. In other words, all initial target state amplitudes $\sin{(\theta)}$ are divided by a constant factor $2\Gamma$ with an multiplicative error $\epsilon$ that can be made exponentially small.
\end{proof}
Note that if one is interested in multiplication by a factor less than one, trivial solutions exist. For any $\Gamma \ge 1/2$, one could prepare an ancilla state $\ket{\Gamma}_c= \frac{1}{2\Gamma}\ket{0}_c+\sqrt{1-\frac{1}{4\Gamma^2}}\ket{1}_c$ and simply define the target state to be $\ket{t}_a\ket{0}_b\ket{0}_c$ in the prepared state $\hat{G}\ket{0}_a\ket{0}_c\ket{\Gamma_c}=\frac{\sin{(\theta)}}{2\Gamma}\ket{t}_a\ket{0}_b\ket{0}_c+\cdots$.

\section{Uniform Spectral Amplification by Amplitude Multiplication}
\label{Sec:Ham_Sim_Overlaps}
We now consider a certain kind of structure within the signal unitary $\hat{U}$ that encodes some Hamiltonian in standard-form. Whereas Sec.~\ref{Sec:Uniform_Hamiltonian_Amplification} treats $\hat{U}$ as a single oracle, we now assume that it factors into other unitaries, say $\hat{U}=\hat{U}^\dag_\text{row}\hat{U}_\text{col}$, or $\hat{U}=\hat{U}^\dag_\text{row}\hat{U}_\text{mix}\hat{U}_\text{col}$, that we assume access to as oracles. This factorization imposes in Sec.~\ref{SubSec:State_Overlaps} the interpretation that encoded Hamiltonians have matrix elements defined by the overlap between some set of quantum states. We investigate in Sec.~\ref{SubSec:Amplified_Overlap} how this structure may be exploited for uniform spectral amplification. By applying amplitude multiplication, this is possible through Lem.~\ref{Thm:Ham_Encoding_Uniform_Amplification_State_Overlaps} in a fairly general setting. In Sec.~\ref{SubSec:Reduction_Sparse_Matrices}, we specialize this to sparse Hamiltonian simulation, which leads to the improved simulation algorithm Thm.~\ref{Cor:Ham_Sim_Sparse_Amplified}. In Sec.~\ref{Sec:Lower_Bound}, this algorithm is proven to be optimal in all parameters, at least up to logarithmic factors, through a matching lower bound Thm.~\ref{Thm:Lower_Bound}.

\subsection{Matrix Elements as State Overlaps}
\label{SubSec:State_Overlaps}
Decomposing the signal unitary into factors motivates a different interpretation of the standard-form
\begin{align}
\label{Eq:Level_Two_Encoding}
\frac{\hat{H}}{\alpha}=(\bra{0}_a\otimes\hat{I}_s)\hat{U}(\ket{0}_a\otimes \hat{I})=(\bra{0}_a\otimes\hat{I}_s)\hat{U}^\dag_\text{row}\hat{U}_\text{col}(\ket{0}_a\otimes \hat{I}_s)
\end{align}
By definition, any unitary operator implements a basis transformation $\hat{U}=\sum_{k}\ket{B_k}\bra{A_k}_{as}$ between complete orthonormal sets of basis states $\{\ket{B_k}_{as}\}$ and $\{\ket{A_k}_{as}\}$, and similarly for $\hat{U}_\text{row},\hat{U}_\text{col}$. Now consider a set of basis states $\{\ket{j}_a\}$ on the ancilla register,
and a set of basis states $\{\ket{u_j}_s\}$ on the system register. Without loss of generality, we may represent $\hat{U}_\text{row}=\sum_{k}\ket{\chi_{0,k}}_{as}\bra{0}_a\bra{u_k}_s+\sum_{j\neq 0}\sum_{k}\ket{\chi_{j,k}}_{as}\bra{j}_a\bra{u_k}_s$ and $\hat{U}_\text{col}=\sum_{k}\ket{\psi_{0,k}}_{as}\bra{0}_a\bra{u_k}_s+\sum_{j\neq 0}\sum_{k}\ket{\psi_{j,k}}_{as}\bra{j}_a\bra{u_k}_s$ for some set of basis states $\{\ket{\chi_{j,k}}_{as}\}$, $\{\ket{\psi_{j,k}}_{as}\}$. Let us substitute this into Eq.~\ref{Eq:Level_Two_Encoding} and drop the $0$ subscript.
\begin{align}
\label{Eq:State_Overlap_Model}
\frac{\hat{H}_{jk}}{\alpha}=\bra{u_j}\frac{\hat{H}}{\alpha}\ket{u_k}&=\left(\bra{0}_a\bra{u_j}_s\hat{U}^\dag_\text{row}\right)\left(\hat{U}_\text{col}\ket{0}_a\ket{u_k}_s\right)=\langle\chi_{0,j}|\psi_{0,k}\rangle_{as}=\langle\chi_{j}|\psi_{k}\rangle_{as}.
\end{align}
In other words, elements of $\hat{H}$ in the $\ket{u_j}_s$ basis may always be interpreted as the overlap of appropriately defined quantum states $\ket{\psi_{k}}_{as},\ket{\chi_{k}}_{as}$, which we call overlap states. Moreover, $\hat{H}$ need not unitary when the dimension of these states is greater than $\hat{H}$. 

More generally, we may factor the signal unitary into three unitaries $\hat{U}=\hat{U}^\dag_{\text{row}}\hat{U}_{\text{mix}}\hat{U}_{\text{col}}$. If we preserve the interpretation of $\hat{U}_{\text{row}}$ and $\hat{U}_{\text{col}}$ as preparing appropriately defined quantum states, the third unitary $\hat{U}_{\text{mix}}$ is a new component that mixes these states to encode the following Hamiltonian in standard-form
\begin{align}
\label{Eq:State_Overlap_Model_imperfect_3}
\frac{\hat{H}}{\alpha}=(\bra{0}_a\otimes\hat{I}_s)\hat{U}^\dag_{\text{row}}\hat{U}_{\text{mix}}\hat{U}_{\text{col}}(\ket{0}_a\otimes \hat{I}_s),
\quad
\frac{\hat{H}_{jk}}{\alpha} &= \bra{\chi_j }_{as}\hat{U}_{\text{mix}} \ket{\psi_k}_{as}.
\end{align}
Note that this reduces to Eq.~\ref{Eq:State_Overlap_Model} by choosing $\hat{U}_{\text{mix}}$ to be identity, or by absorbing it into the definition of either $\hat{U}_{\text{row}}$ or $\hat{U}_{\text{col}}$. Combined with Thm.~\ref{Thm:Ham_Sim_Qubitization}, time evolution by $e^{-i\hat{H}t}$ may be approximated with error $\epsilon$ using $\mathcal{O}\big(t\alpha+\frac{\log{(1/\epsilon)}}{\log\log{(1/\epsilon)}}\big)$ queries to $\hat{U}_\text{row},\hat{U}_\text{mix}$, and $\hat{U}_\text{col}$.

However, the ability to efficiently prepare arbitrary quantum states represents an extremely powerful model of computation. For instance, arbitrary temperature Gibbs state preparation is QMA-complete~\cite{Gharibian2015quantum}. That not all states may be prepared in $\mathcal{O}(1)$ queries to commonly used quantum oracles can be built into the definition of the overlap states by splitting them into `good' components $\ket{\tilde\psi_j}_{a_1s},\ket{\tilde\chi_j}_{a_1s}$ marked by an ancilla state $\ket{0}_{a_2}$, and `bad' components that are discarded. Difficult states then have a small amplitude in the $\ket{0}_{a_2}$ subspace. Thus
\begin{align}
\label{Eq:State_Overlap_Model_imperfect}
\ket{\psi_j}_{as}&=\sqrt{\lambda_{\beta}\beta_j} \ket{\tilde \psi_j}_{a_1 s}\ket{0}_{a_2}+\sqrt{1-\lambda_{\beta}\beta_j}\ket{\psi_{\text{bad},j}}_{a_1s}\ket{1}_{a_2}, 
\\ \nonumber
\ket{\chi_j}_{as}&=\sqrt{\lambda_{\gamma}\gamma_j} \ket{\tilde \chi_j}_{a_1s}\ket{0}_{a_2}+\sqrt{1-\lambda_{\gamma}\gamma_j}\ket{\chi_{\text{bad},j}}_{a_1s}\ket{2}_{a_2}.
\end{align}
Note that the dimension of the ancilla register $a_1a_2$ is equal to $a$. The coefficients $\lambda_{\gamma},\lambda_{\beta} \in (0,1]$ represent a slowdown factor due to the difficulty of state preparation, and the coefficients $\beta_j,\gamma_j\in[0,1]$ normalized to $\max_{j}\beta_j = 1,\max_{j}\gamma_j = 1$ represent how the amplitude in good states can be index-dependent by design. By restricting $\hat{U}_\text{mix}$ to be identity on the register $a_2$, this encodes
the following Hamiltonian in standard-form
\begin{align}
\label{Eq:State_Overlap_Model_imperfect_3}
\frac{\hat{H}}{\alpha}=(\bra{0}_a\otimes\hat{I}_s)\hat{U}^\dag_{\text{row}}\hat{U}_{\text{mix}}\hat{U}_{\text{col}}(\ket{0}_a\otimes \hat{I}_s),
\;
\frac{\hat{H}_{jk}}{\alpha}=
\langle \chi_j |_{as} \hat{U}_{\text{mix}}|\psi_k \rangle_{as}
= \sqrt{\Lambda_{\gamma}\Lambda_{\beta}\gamma_j\beta_k}\langle \tilde\chi_j |_{a_1s}\hat{U}_\text{mix}| \tilde\psi_k \rangle_{a_1s}.
\end{align}
By explicitly including the slowdown factor $\sqrt{\lambda_{\gamma}\lambda_{\beta}}$, the spectral norm $\|\hat{H}\| \le \alpha \sqrt{\lambda_{\gamma}\lambda_{\beta}}$ is also reduced.

\subsection{Amplitude Multiplication of Overlap States}
\label{SubSec:Amplified_Overlap}
This \emph{state overlap} encoding of Hamiltonians motivates the use of amplitude amplification. As the amplitudes of all states $\ket{\tilde \psi_j}$ are attenuated by a constant factor $\sqrt{\lambda_\beta}$, the intuition is that one requires $\mathcal{O}(1/\sqrt{\lambda_\beta})$ queries to the state preparation operator $\hat{U}_{\text{row}}$ to boost the amplitude in the subspace marked by $\ket{0}_b$ by a factor $\mathcal{O}(1/\sqrt{\lambda_\beta})$, and similarly for $\ket{\tilde\chi_j}$. Thus $\mathcal{O}(1/\sqrt{\lambda_\beta}+1/\sqrt{\lambda_\gamma})$ queries appears sufficient to reduce the normalization $\alpha$ by a factor $\sqrt{\lambda_{\gamma}\lambda_{\beta}}$. This suggests that a query complexity of Hamiltonian simulation could be improved to $\mathcal{O}\big(t\alpha(\sqrt{\lambda_{\gamma}}+\sqrt{\lambda_{\beta}})+\frac{\log{(1/\epsilon)}}{\log\log{(1/\epsilon)}}\big)$, which is most advantageous when $\lambda_\beta$ and $ \lambda_\gamma$ are both small. However, realizing this speedup is non-trivial.

In the context of prior art in sparse Hamiltonian simulation, attempts have been made to exploit amplitude amplification~\cite{Berry2012}. There, it was discovered that the sinusoidal non-linearity of amplitude amplification introduces large errors. As these error accumulate over long simulation times $t$, controlling them led to query complexity scaling like $\mathcal{O}(t^{3/2}/\epsilon)$, which is polynomially worse than what intuition suggests. In the following, we avoid these issues by introducing a linearized version of amplitude amplification, which we call the \emph{amplitude multiplication} algorithm. 

Before proceeding, note that amplitude amplification also imposes additional restrictions on the form of the overlap states in Eq.~\ref{Eq:State_Overlap_Model_imperfect}. Amplitude amplification requires the ability to perform reflections $\widehat{\text{Ref}}_{\ket{0}_{a_1}}$ about the subspace marked by $\ket{0}_{a_1}$, as well as reflections $\widehat{\text{Ref}}_{\psi}$ on any arbitrary superposition of initial states $\ket{\psi_j}$, that is $\forall j,\;\widehat{\text{Ref}}_{\psi}\hat{U}_\text{col}\ket{0}_a\ket{u_j}_s= -\hat{U}_\text{col}\ket{0}_a\ket{u_j}_s$, and $\widehat{\text{Ref}}_{\psi}$ performs identity for any other ancilla state. The case for $\hat{U}_\text{row}$ and $\ket{\chi_j}$ is identical. Whereas the first operation
\begin{align}
\widehat{\text{Ref}}_{\ket{0}_{a_2}}=(\hat{I}_{a_2}-2\ket{0}\bra{0}_{a_2})\otimes \hat{I}_{a_1s},
\end{align}
is easy using $\mathcal{O}(1)$ primitive gates, the second operation requires $\hat{U}_\text{col}$ to represent controlled state preparation. In other words, with the input $\ket{s_j}$ on the system register, the overlap state has the decomposition
\begin{align}
\label{Eq:State_Overlap_Model_imperfect_3_decomposed}
\hat{U}_\text{col}\ket{0}_a\ket{u_j}_s=\ket{\psi_j}=\left(\sqrt{\lambda_{\beta}\beta_j} \ket{\bar \psi_j}_{a_1}\ket{0}_{a_2}+\sqrt{1-\lambda_{\beta}\beta_j}\ket{\bar \psi_{\text{bad},j}}_{a_1} \ket{1}_{a_2}\right)\ket{u_j}_s,
\end{align}
thus encoding the following Hamiltonian in standard-form
\begin{align}
\label{Eq:Level_Three_Encoding}
\frac{\hat{H}}{\alpha}=(\bra{0}_a\otimes\hat{I}_s)\hat{U}^\dag_{\text{row}}\hat{U}_{\text{mix}}\hat{U}_{\text{col}}(\ket{0}_a\otimes \hat{I}),
\quad
\frac{\hat{H}_{jk}}{\alpha}= \sqrt{\lambda_\gamma\lambda_\beta\gamma_j\beta_k} \bra{u_j}_s\bra{\bar \chi_j}_{a_1}\hat{U}_\text{mix}\ket{\bar \psi_k}_{a_1}\ket{u_k}_s,
\end{align}
and allowing us to construct the controlled-reflection operator
\begin{align}
\widehat{\text{Ref}}_\psi=\sum_j(\hat{I}_a-2\ket{0}_{a_2}\ket{\bar\psi_j}_{a_1}\bra{\bar\psi_j}_{a_1}\bra{0}_{a_2})\otimes \ket{u_j}\bra{u_j}_s=\hat{U}_{\text{col}}((\hat{I}_{a}-2\ket{0}\bra{0}_{a})\otimes \hat{I}_{s})\hat{U}^\dag_{\text{col}},
\end{align}
using $2$ queries and $\mathcal{O}(\log d)$ primitive gates. 

The error introduced by a naive application of amplitude amplification is illustrated by an explicit calculation. Using a sequence of $m\ge 0$ controlled-Grover iterates $\widehat{\text{Ref}}_\psi\widehat{\text{Ref}}_{\ket{0}_{a_2}}$ making $\mathcal{O}(m)$ queries, one can prepare the state
\begin{align}
\ket{\psi_{\text{amp},j}}
&=
\left(\widehat{\text{Ref}}_\psi\widehat{\text{Ref}}_{\ket{0}_{a_2}}\right)^m\hat{U}_{\text{col}}\ket{0}_a\ket{u_j}_s
\\
&=
\left(\sin{\left((2m+1)\sin^{-1}{\left(\sqrt{\lambda_\beta\beta_j}\right)}\right)} \ket{\bar \psi_j}_{a_1}\ket{0}_{a_2}+\cdots \ket{1}_{a_2}\right)\ket{u_j}_s
=
\left(\sqrt{\beta'_j} \ket{\bar \psi_j}_{a_1}\ket{0}_{a_2}+\cdots \ket{1}_{a_2}\right)\ket{u_j}_s. \nonumber
\end{align}
With the choice $m=\lfloor \frac{\pi}{4\sin^{-1}{\left(\sqrt{\lambda_\beta}\right)}}-\frac{1}{2}\rfloor=\mathcal{O}(\lambda^{-1/2}_\beta)$, we are guaranteed that all $\sqrt{\beta_j'}\ge \sqrt{\lambda_\beta\beta_j}$. Though this improves the normalization, it also specifies an erroneous Hamiltonian as the matrix elements $\langle \chi_{\text{amp},j}|\hat{U}_\text{mix}|\psi_{\text{amp},k}\rangle$ are larger than those of $\hat{H}_{jk}$ by an index-dependent factor.

In contrast, Amplitude multiplication in Thm.~\ref{Thm:Linear_Amplitude_Amplification} avoids this non-linearity and and allows us to boost the normalization of the encoded Hamiltonian with only an exponentially small distortion to its spectrum. This leads to
\begin{lemma}[Uniform spectral amplification by multiplied state overlaps]
\label{Thm:Ham_Encoding_Uniform_Amplification_State_Overlaps}
Let the Hamiltonian $\hat{H}$ be encoded in the standard-form of Eq.~\ref{Eq:Level_Three_Encoding} with normalization $\alpha$. Given upper bounds $\Lambda_\beta \in [\lambda_\beta,1/2],\; \Lambda_\gamma \in [\lambda_\gamma,1/2]$ on the slowdown factors, and a target error $\epsilon\in(0,\text{min}\{\Lambda_\beta,\Lambda_\gamma\})$, the Hamiltonian $\hat{H}_{\text{lin}}$ can be encoded in standard-form with normalization $4\alpha\sqrt{\Lambda_\beta\Lambda_\gamma}$ such that $\|\hat{H}_{\text{lin}}-\hat{H}\|\le \frac{5}{4}\epsilon\|\hat{H}\|< \frac{5}{4}\alpha \epsilon\sqrt{\Lambda_\beta\Lambda_\gamma}$ using $Q=\mathcal{O}((\Lambda_\beta^{-1/2}+\Lambda_\gamma^{-1/2})\log{(1/\epsilon)})$ queries, $\mathcal{O}(Q \log{(d)})$ primitive gates, and $1$ additional ancilla qubit.
\end{lemma}
\begin{proof}
Let us apply Thm.~\ref{Thm:Linear_Amplitude_Amplification}, which requires one additional ancilla qubit, to the state overlap model Eq.~\ref{Eq:State_Overlap_Model_imperfect_3_decomposed}. We identify $\hat{U}_{\text{col}}$ as the state preparation operator that prepares the target state marked by $\ket{0}_{a_2}$ with overlap $\sqrt{\lambda_\beta\beta_j}$. Assume that $\sqrt{\lambda_\beta}\le 1/2$, and let $\Lambda_\beta \in [\lambda_\beta,1/2]$ be an upper bound on the slowdown factor. Then there exists a quantum circuit $\hat{U}'_{\text{col}}$ that makes $Q_\beta=\mathcal{O}(\Lambda_\beta^{-1/2}\log{(1/\epsilon)})$ queries to $\hat{U}_{\text{col}}$ and uses $\mathcal{O}(Q_\beta\log{(d)})$ primitive gates, and similarly for $\hat{U}_{\text{row}}$, to prepare the states
\begin{align}
\label{Eq:State_Overlap_Model_amplified}
\ket{\psi_{\text{lin},j}}
&=
\hat{U}'_{\text{col}}\ket{0}_a\ket{u_j}_s
=
\left(\sqrt{\frac{\lambda_\beta\beta_j}{4\Lambda_\beta}}(1+\epsilon_{\beta,j}) \ket{\bar \psi_j}_{a_1}\ket{0}_{a_2}+\cdots\ket{1}_{a_2}\right)\ket{u_j}_s, 
\\ \nonumber
\ket{\chi_{\text{lin},j}}
&=
\hat{U}'_{\text{row}}\ket{0}_a\ket{u_j}_s
=
\left(\sqrt{\frac{\lambda_\gamma\gamma_j}{4\Lambda_\gamma}}(1+\epsilon_{\gamma,j}) \ket{\bar \chi_j}_{a_1}\ket{0}_{a_2}+\cdots\ket{2}_{a_2}\right)\ket{u_j}_s,
\end{align}
where $|\epsilon_{\beta,j}|,|\epsilon_{\gamma,j}|< \epsilon\in(0,\text{min}\{\Lambda_\beta,\Lambda_\gamma\})\le 1/2$ are state-dependent errors in the amplitude. Let us define the Hamiltonian $\hat{H}_{\text{lin}}$ encoded in in standard-form with normalization $4\alpha\sqrt{\Lambda_\beta\Lambda_\gamma}$ as follows
\begin{align}
\frac{\hat{H}_{\text{lin}}}{4\alpha\sqrt{\Lambda_\beta\Lambda_\gamma}}=(\bra{0}_a\otimes \hat{I}_s)\hat{U}'_{\text{row}}\hat{U}_{\text{mix}}\hat{U}'_{\text{col}}(\ket{0}_a\otimes \hat{I}_s)=\sum_{jk}\frac{\hat{H}_{jk}}{4\alpha}\frac{(1+\epsilon_{\gamma,j})(1+\epsilon_{\beta,k})}{\sqrt{\Lambda_\beta\Lambda_\gamma}}\ket{u_j}\bra{u_k}_s.
\end{align}
We may now evaluate the error of $\hat{H}_{\text{lin}}$ from that of the original Hamiltonian $\hat{H}$, following a similar approach from~\cite{Berry2012}. Let $\hat{\epsilon}_\beta$ be a diagonal matrix with elements $\epsilon_{\beta,j}$ and similarly for $\epsilon_{\gamma,j}$.
Then
\begin{align}
\label{Eq:Overlap_Hamiltonian_Error}
\hat{H}_{\text{lin}}
&=\left(\hat{H}+\hat{\epsilon}_\gamma \hat{H}+\hat{H}\hat{\epsilon}_\beta + \hat{\epsilon}_\gamma\hat{H}\hat{\epsilon}_\beta \right),
\\\nonumber
\|\hat{H}_{\text{lin}}-\hat{H}\|
&\le \|\hat{H}\|\left(\|\hat{\epsilon}_\beta\|+\|\hat{\epsilon}_\gamma\|+\|\hat{\epsilon}_\beta\|\|\hat{\epsilon}_\gamma\|\right)\le \|\hat{H}\|(2\epsilon+\epsilon^2) < \frac{5}{4}\|\hat{H}\|\epsilon < \frac{5}{4}\alpha \sqrt{\Lambda_\beta\Lambda_\gamma}\epsilon.
\end{align}
where the second-last inequality is due to $\epsilon < 1/2$, and the last inequality applies the upper bound $\|\hat{H}\|\le \alpha \sqrt{\Lambda_\beta\Lambda_\gamma}$. Summing up $Q=Q_\beta+Q_\gamma+1$ leads to the claimed query and gate complexities.
\end{proof}

Combining with Thm.~\ref{Thm:Ham_Sim_Qubitization} then furnishes the following result on Hamiltonian simulation.
\begin{lemma}[Hamiltonian simulation by multiplied state overlaps]
\label{Thm:HamSim_Amplified_Overlaps}
Let the Hamiltonian $\hat{H}$ be encoded in the standard-form of Eq.~\ref{Eq:Level_Three_Encoding} with normalization $\alpha$. Given upper bounds $\Lambda_\beta \in [\lambda_\beta,1/2],\; \Lambda_\gamma \in [\lambda_\gamma,1/2]$ on the slowdown factors, $\Lambda\ge \|\hat{H}\|$, and a target error $\epsilon\in(0,\text{min}\{\Lambda_\beta,\Lambda_\gamma\})$, time-evolution $e^{-i\hat{H}t}$ be approximated with error $\epsilon$ using $Q=\mathcal{O}\Big(t\alpha(\sqrt{\Lambda_\beta}+\sqrt{\Lambda_\gamma})\log{(\frac{t\Lambda}{\epsilon})}+(\Lambda_\beta^{-1/2}+\Lambda_\gamma^{-1/2})\frac{\log{(1/\epsilon)}\log{(t\Lambda/\epsilon)}}{\log\log{(1/\epsilon)}}\Big)$ queries, $\mathcal{O}(Q\log{(d)})$ primitive gates, and $\mathcal{O}(1)$ additional ancilla qubits.
\end{lemma}
\begin{proof}
From Lem.~\ref{Thm:Ham_Encoding_Uniform_Amplification_State_Overlaps}, we may encode $\hat{H}_\text{lin}$ in standard-form with normalization $4\alpha\sqrt{\Lambda_\beta\Lambda_\gamma}$ and error $\|\hat{H}_{\text{lin}}-\hat{H}\| =\mathcal{O}(\|\hat{H}\|\epsilon_0) =\mathcal{O}(\Lambda \epsilon_0)$. This requires $Q_0=\mathcal{O}((\Lambda_\beta^{-1/2}+\Lambda_\gamma^{-1/2})\log{(1/\epsilon_0)})$ queries to $\hat{U}_\text{row},\hat{U}_\text{mix},\hat{U}'_\text{col}$ and their inverses, $\mathcal{O}(Q_0 \log{(d)})$ primitive gates, and $1$ additional ancilla qubit. Using the fact $\|e^{i \hat A}-e^{i \hat B}\|\le \|\hat A-\hat B\|$, the error of $e^{-i\hat{H}_\text{lin}t}$ from ideal time-evolution is $\|e^{-i\hat{H}_\text{lin}t}-e^{-i\hat{H}t}\|\le \|\hat{H}_\text{lin}t-\hat{H}t\|=\mathcal{O}(t\Lambda\epsilon_0)$. 
By combining with Thm.~\ref{Thm:Ham_Sim_Qubitization}, time-evolution by $e^{-i\hat{H}_\text{lin}t}$ can be approximated with error $\epsilon_1$ using $Q_1=\mathcal{O}\big(t\alpha\sqrt{\Lambda_\beta\Lambda_\gamma} + \frac{\log{(1/\epsilon_1)}}{\log\log{(1/\epsilon_1)}}\big)$ queries to controlled-$\hat{U}'_\text{row}\hat{U}_\text{mix}\hat{U}'_\text{col}$ and its inverse, $\mathcal{O}(Q_1\log{(d)})$ additional primitive gates, and $\mathcal{O}(1)$ additional ancilla qubits. 
Thus time-evolution by $e^{-i\hat{H}t}$ can be approximated with error $\epsilon=\mathcal{O}(\epsilon_1+t\Lambda\epsilon_0)$ using $Q=Q_0Q_1$ queries to controlled-$\hat{U}_\text{row}, \hat{U}_\text{mix}, \hat{U}_\text{col}$ and their inverses, and $\mathcal{O}(Q_1\log{(d)}+Q_0Q_1\log{(d)})=\mathcal{O}(Q\log{(d)})$ primitive gates. We can control the error by choosing $\epsilon_1=\mathcal{O}(\epsilon)$ and $\epsilon_0 =\mathcal{O}(\epsilon/(t\Lambda))$. Substituting into $Q$ produces the claimed query complexity.
\end{proof}
In the asymptotic limit of large $t\gg \log{(1/\epsilon)}$, the query complexity may be simplified to $\mathcal{O}\Big(t\alpha(\sqrt{\Lambda_\beta}+\sqrt{\Lambda_\gamma})\log{(\frac{t\Lambda}{\epsilon})}\Big)$ queries.

\subsection{Reduction to Sparse Matrices}
\label{SubSec:Reduction_Sparse_Matrices}
The results of Sec.~\ref{Sec:Ham_Sim_Overlaps}, presented in a general setting, apply to the special case of sparse matrices. The reduction follows by making three additional assumptions. First, assume that the dimension of $\ket{0}_a\in \mathbb{C}^{3n}$ is larger than that of $\ket{u_j}_s\in\mathbb{C}^n$. Second, assume that $\forall j\in[n],\;\ket{u_j}_s,$ is the computational basis $\ket{j}_s$. Third, we assume that there exists oracles in Def.~\ref{Def:Sparse_Oracle} that describe $d$-sparse matrices~\cite{Berry2012}:
With these oracles and an upper bound $\Lambda_{\text{max}}\ge \|\hat{H}\|_{\text{max}}$, it is well-known that $\mathcal{O}(1)$ queries suffice to implement the isometry represented by $\hat{U}_{\text{row}}\ket{0}_a$ and $\hat{U}_{\text{col}}\ket{0}_a$ with output states
\begin{align}
\label{Eq:Sparse_Hmax_states}
\hat{U}_{\text{col}}\ket{0}_a\ket{j}_s &= \ket{\psi_j}_{as} = \frac{1}{\sqrt{d}}\sum_{p\in F_{j}}\ket{j}_s\ket{p}_{a_1}\left(\sqrt{\frac{\hat{H}_{jp}}{\Lambda_{\text{max}}}}\ket{0}_{a_2}+\sqrt{1-\frac{|\hat{H}_{jp}|}{\Lambda_{\text{max}}}}\ket{1}_{a_2}\right),
\\\nonumber
\bra{0}_a\bra{k}_s\hat{U}^\dag_{\text{row}} &= \bra{\chi_k}_{as} = \frac{1}{\sqrt{d}}\sum_{q\in F_{k}}\bra{k}_{s}\bra{q}_{a_1}\left(\sqrt{\frac{\delta_{kq}\hat{H}_{kq}+(1-\delta_{kq})\hat{H}^*_{kq}}{\Lambda_{\text{max}}}}\bra{0}_{a_2}+\sqrt{1-\frac{|\hat{H}_{kq}|}{\Lambda_{\text{max}}}}\bra{2}_{a_2}\right),
\\\nonumber
\langle \chi_j |\hat{U}_{\text{mix}} |\psi_k \rangle&=\frac{\hat{H}_{jk}}{\alpha}=\frac{\hat{H}_{jk}}{d\Lambda_{\text{max}}},
\end{align}
where $\delta_{jk}$ is the Kronecker delta function, and $F_j= \{k: k = f(j,l)\; , l\in[d]\}$ is the set of non-zero column indices in row $j$. Note that our definition of the isometry Eq.~\ref{Eq:Sparse_Hmax_states} is an improvement over~\cite{Berry2012} as it avoids ambiguity in both the principal range of the square-roots when $\hat{H}_{jk}<0$ and a sign problem when $\hat{H}_{jj}<0$. We also choose $\hat{U}_{\text{mix}}$ to swap the registers $s$ and $a_1$. From~\cite{Berry2012}, the gate complexity of $\hat{U}_{\text{col}}$, $\hat{U}_{\text{row}}$, and $\hat{U}_{\text{mix}}$ combined is $\mathcal{O}(\log{(n)}+\text{poly}(m))$, where $m=\mathcal{O}(\log{(t\|\hat{H}\|/\epsilon)})$ is the number of bits of precision of $\hat{H}_{jk}$. The contribution from $\text{poly}(m)=\mathcal{O}(m^{5/2})$ is due to integer arithmetic for for computing square-roots and trigonometric functions. This combined with Thm.~\ref{Thm:Ham_Sim_Qubitization} recovers the previous best result on sparse Hamiltonian simulation using $Q=\mathcal{O}\big(td\Lambda_{\text{max}}+\frac{\log{(1/\epsilon)}}{\log\log{(1/\epsilon)}}\big)$ queries~\cite{Low2016HamSim}, and $\mathcal{O}(Q(\log{(n)}+\text{poly}(m)))$ primitive gates.

To see how Thm.~\ref{Thm:HamSim_Amplified_Overlaps} improves on this, we rewrite Eq.~\ref{Eq:Sparse_Hmax_states} in the format of Eq.~\ref{Eq:State_Overlap_Model_imperfect} by collecting coefficients of the subspace marked by $\ket{0}_{a_2}$.
\begin{align}
\label{Eq:State_Overlap_Model_imperfect_collected}
\ket{\psi_j}_{as}&=\sqrt{\frac{\sigma_j}{d\Lambda_{\text{max}}}}\left(\sum_{p\in F_{j}}\sqrt{\frac{\hat{H}_{jp}}{\sigma_j}}\ket{j}_s\ket{p}_{a_1}\right)\ket{0}_{a_2}
+
\cdots\ket{j}_s\ket{1}_{a_2},
\\ \nonumber
\bra{\chi_k}_{as}&=\sqrt{\frac{\sigma_k}{d\Lambda_{\text{max}}}}\left(\sum_{q\in F_{k}}\sqrt{\frac{\delta_{kq}\hat{H}_{kq}+(1-\delta_{kq})\hat{H}^*_{kq}}{\sigma_k}}\bra{k}_{s}\bra{q}_{a_1}\right)\bra{0}_{a_2}
+
\cdots\bra{k}_s\bra{2}_{a_2}, 
\end{align}
where $\sigma_j=\sum_k|\hat{H}_{jk}|$, and the induced one-norm $\|\hat{H}\|_1=\max_{j}\sigma_j$. Note that $\ket{\bar{\psi}_j}=\sum_{p\in F_{j}}\sqrt{\frac{\hat{H}_{jp}}{\sigma_j}}\ket{j}_s\ket{p}_{a_1}$, and similarly for $\ket{\bar{\chi}_j}$. From this, we obtain our main result on sparse Hamiltonian simulation Thm.~\ref{Cor:Ham_Sim_Sparse_Amplified}.
\begin{proof}[Proof of Thm.~\ref{Cor:Ham_Sim_Sparse_Amplified}]
Comparison of Eq.~\ref{Eq:State_Overlap_Model_imperfect_collected} with Eq.~\ref{Eq:State_Overlap_Model_imperfect} yields $\beta_j =\gamma_j = \frac{\sigma_j}{\|\hat{H}\|_1}$, $\lambda_\beta = \lambda_\gamma =  \frac{\|\hat{H}\|_1}{d\Lambda_{\text{max}}}$. Thus we have the upper bound $\Lambda_\beta = \Lambda_\gamma = \frac{\Lambda_1}{d\Lambda_\text{max}}\ge \lambda_\beta=\lambda_\gamma$. Moreover, from Eq.~\ref{Eq:Sparse_Hmax_states}, the normalization constant $\alpha = d \Lambda_{\text{max}}$. The claimed query complexity is obtained by substitution into Cor.~\ref{Thm:HamSim_Amplified_Overlaps}.
\end{proof}

This result is quite remarkable as it strictly improves upon prior art, modulo logarithmic factors, by exploiting additional structural information. In the asymptotic limit of large $\Lambda_1 t \gg \log{(1/\epsilon)}$, the query complexity may be simplified to $\mathcal{O}\Big(t\sqrt{d\Lambda_{\text{max}}\Lambda_1}\log{(\frac{t\Lambda}{\epsilon})}\Big)$.
Using the inequality $\|\hat{H}\|\le \|\hat{H}\|_1\le d\|\hat{H}\|_{\text{max}}$, the worst-case occurs when these norms are all equal thus  $\Lambda=\Lambda_1=d\Lambda_{\text{max}}$. There, the query complexity of Thm.~\ref{Cor:Ham_Sim_Sparse_Amplified} up to logarithmic factors is $\mathcal{O}(td\Lambda_{\text{max}})$, equal to that of prior art~\cite{Low2016HamSim}. However, the best-case $\|\hat{H}\|_1=\mathcal{O}(\|\hat{H}\|_{\text{max}})$ leads to a quadratic improvement in sparsity with query complexity of $\mathcal{O}(t\sqrt{d}\Lambda_{\text{max}})$, also ignoring logarithmic factors.

Another approach implicit in~\cite{Berry2012} assumes that $\sigma_j$ are provided by the quantum oracle $\hat{O}_{C}\ket{j}_s\ket{z}_c=\ket{j}_s\ket{z\oplus \sigma_j}_c$ when queried the $j\in[n]$ row index. This allows us to exactly compensate for the sinusoidal non-linearity of amplitude amplification by modifying initial state amplitudes by some $j$-dependent multiplicative factor.
Thus $\hat{H}$ may be encoded in standard-form with normalization $\mathcal{O}(\sqrt{d\Lambda_{\text{max}}\Lambda_1})$ exactly without any error, leading to a Hamiltonian simulation algorithm with query complexity $Q=\mathcal{O}\big(t(d\|\hat{H}\|_{\text{max}}\|\hat{H}\|_1)^{1/2}+\frac{\log{(1/\epsilon)}}{\log\log{(1/\epsilon)}}\big)$. While improves on Thm.~\ref{Cor:Ham_Sim_Sparse_Amplified} by logarithmic factors, and matches the complexity Claim.~\ref{Claim:Sparse_Ham_Sim}, $\hat{O}_{C}$ is in general difficult to construct.

\subsection{Lower Bound on Sparse Hamiltonian Simulation}
\label{Sec:Lower_Bound}
In this section, we prove the lower bound Thm.~\ref{Thm:Lower_Bound} on sparse Hamiltonian simulation, given information on the sparsity, max-norm, and induced one-norm. The lower bounds in prior art are obtained by constructing Hamiltonians that compute well-known functions. When applied to our situation, one obtains $\Omega(t\|\hat{H}\|_1)$ queries therough the $\text{PARITY}$ problem~\cite{Berry2015Hamiltonian}, and $\Omega(\sqrt{d})$ queries through $\text{OR}$~\cite{Berry2012}. This 
leads to an additive lower bound $\Omega(t\|\hat{H}\|_1+d)$. Using similar techniques, we obtain a stronger lower bound $\Omega(t(d\|\hat{H}\|_1)^{1/2})$ by creating a Hamiltonian that computes the solution to the composed function $\text{PARITY}\circ \text{OR}$. Specifically, we combine a Hamiltonian that solves $\text{PARITY}$ on $n$ bits with constant error using at least $\Omega(s\|\hat{H}\|_{\text{max}}t)$ queries, where $t=\Theta(\frac{n}{s\|\hat{H}\|_{\text{max}}})$, with a Hamiltonian that solves $\text{OR}$ on $m$ bits exactly, with the promise that at most $1$ bit is non-zero, using at least $\Omega(\sqrt{m})$ queries. Note that in all cases, the query complexity with respect to error is at least an additive term $\Omega(\frac{\log{(1/\epsilon)}}{\log\log{(1/\epsilon)}})$~\cite{Berry2015Hamiltonian}.

The Hamiltonian $\hat{H}_{\text{PARITY}}$ that solves $\text{PARITY}$ on $n$ bits is well-known~\cite{Berry2015Hamiltonian}, and is based on the Hamiltonian $\hat{H}_{\text{spin}}$ for perfect state transfer in spin chains. For completeness, we outline the procedure. Consider a Hamiltonian of $\hat{H}_{\text{spin}}$ dimension $n+1$, with matrix elements in the computational basis $\{\ket{j}_s:j\in[n+1]\}$ defined as
\begin{align}
\bra{j-1}_s\hat{H}_{\text{spin}}\ket{j}_s=\sqrt{j(N-j+1)}/N.
\end{align}
Note that this Hamiltonian has sparsity $1$, max-norm $\Theta(1)$, and $1$-norm $\Theta(1)$. Time evolution by this Hamiltonian $e^{-i\hat{H}_{\text{spin}}n\pi/2}\ket{0}_s=\ket{n}_s$ exactly transfers the state $\ket{0}$ to $\ket{n}$ in time $t=\frac{\pi N}{2}$

One way to speed up these dynamics is to uniformly increase the value of all matrix elements. However, any increase in $\|\hat{H}\|_{\text{max}}$ is trivial as it simply decreases $t$ by a proportionate amount.  Another way is to boost the sparsity of $\hat{H}_{\text{spin}}$ by taking a tensor product with a Hamiltonian $\hat{H}_{\text{complete}}$ of dimension $s$ where all matrix elements are $1$ in the computational basis $\{\ket{j}_c:j\in[s]\}$.
\begin{align}
\bra{i}_c\hat{H}_{\text{complete}}\ket{j}_c=1, \quad \forall i\in[s],\; j\in[s].
\end{align}
One the eigenstates of $\hat{H}_{\text{complete}}$ is the uniform superposition $\ket{u}_c=\frac{1}{\sqrt{s}}\sum_{j\in[s]}\ket{j}_c$ with eigenvalue $\hat{H}_{\text{complete}}\ket{u}_c=s\ket{u}_c$. Thus we define the Hamiltonian
\begin{align}
\hat{H}_{sc}=\hat{H}_{\text{spin}}\otimes \hat{H}_{\text{complete}}.
\end{align}
Note that the $\hat{H}_{sc}$ has sparsity $s$, max-norm $\Theta(1)$, and $1$-norm $\Theta(s)$. One can see that $\hat{H}_{sc}$ perform faster state transfer like $e^{-i\hat{H}_{sc}N\pi/(2s)}\ket{0}_s \ket{u}_c=\ket{n}_s \ket{u}_c$ in time $t=\frac{\pi n}{2 s}$. We find it useful to define the state $\ket{j}_{sc}= \ket{j}_s \ket{u}_c$.

Adding another qubit to this composite Hamiltonian together with some slight modification solves ${\text{PARITY}}$. Given an $n$-bit string $x=x_0x_2...x_{n-1}$, let us consider the Hamiltonian of dimension $2$ that computes the $\text{NOT}$ function on the computational basis $\{\ket{j}_{\text{output}}:j\in[2]\}$,
\begin{align}
\hat{H}_{\text{NOT},j}=\left(
\begin{matrix}
x_j \oplus 1 & x_j \\
x_j & x_j \oplus 1
\end{matrix}
\right).
\end{align}
One can see that $\hat{H}_{\text{NOT},j}\ket{0}_\text{output}=\ket{1}_\text{output}$ and $\hat{H}_{\text{NOT},j}\ket{1}_\text{output}=\ket{0}_\text{output}$, as expected of a $\text{NOT}$ function. In the basis $\ket{j}_{sc}$, we define the Hamiltonian
\begin{align}
\label{Eq:Ham_Parity}
\hat{H}_{\text{PARITY}}=\left(\sum_{j\in[n]}\frac{\sqrt{j(N-j+1)}}{N}\ket{j+1}\bra{j}_{sc}\otimes \hat{H}_{\text{NOT},j}\right) + \text{Hermitian conjugate}.
\end{align}
This Hamiltonian also performs perfect state transfer, but since the path of each transition between the states $\ket{0}_\text{output}$ and $\ket{1}_\text{output}$ are gates by a $\text{NOT}$ function on the bit $x_j$, the output state of time-evolution $e^{-i\hat{H}_{\text{PARITY}}N\pi/(2s)}\ket{0}_s\ket{u}_c\ket{0}_\text{output}=\ket{n}_s\ket{u}_c\ket{\bigoplus_j x_j}_\text{output}$. In the computational basis, $\hat{H}_{\text{PARITY}}$ has sparsity $2s$, max-norm $\Theta(1)$, and $1$-norm $\Theta(s)$. Even though $\hat{H}_{\text{NOT},j}$ has only one non-zero element, the sparsity increases by factor $2$ as we cannot compute beforehand the column index the non-zero. Thus measuring the $\text{output}$ register returns the parity of $x$
\begin{align}
\text{PARITY}(x)=\bigoplus_{j=0}^{n-1} x_j,
\end{align}
after evolving for time $t=\frac{\pi n}{2 s}$. It is well-known that the parity on $n$ bits cannot be computed with less than $\Omega(n)$ quantum queries, thus the query complexity of simulating time-evolution by $\hat{H}_{\text{PARITY}}$ for time $t$ is at least $\Omega(ts)$. As sparsity and $1$-norm exhibit the same scaling and in general $\|\hat{H}\|_{1}\le d\|\hat{H}\|_{\text{max}}$, the more accurate statement here if given information on $\|\hat{H}\|_1$ is the lower bound of $\Omega(t\|\hat{H}\|_1)$ queries. In constrast, the lower bound of~\cite{Berry2015Hamiltonian} quotes $\Omega(\text{sparsity} \times t)$ as they consider the case where one is given information only on the sparsity..

We now present the extension to creating a Hamiltonian that solves $\text{PARITY}\circ \text{OR}$. Notably, this Hamiltonian allows one to vary sparsity and $1$-norm independently.
\begin{proof}[Proof of Thm.~\ref{Thm:Lower_Bound}]
The first step is construct a Hamiltonian that solves the $\text{OR}$ function on $m$ bits $x_{0}x_{1}...x_{m-1}$, promised that at most $1$ bit is non-zero. This Hamiltonian of dimension $2m$, in the computational basis $\{\ket{k}_{\text{output}}\ket{j}_{o}:k\in[2],j\in[m]\}$, is 
\begin{align}
\hat{H}_{\text{OR}}=\left(
\begin{array}{c|c}
\hat{C}_{ 1} & \hat{C}_{0} \\
\hline
\hat{C}^\dag_{0} & \hat{C}_{1}
\end{array}
\right).
\end{align}
Note that our construction is based on a modification of~\cite{Berry2014}, where $\hat{C}_{1}$ there is zero matrix. Here, $\hat{C}_{1}$  mimics the top-left component of $\hat{H}_{\text{NOT}}$ in that is performs a bit-flip on the output register if $\text{OR}(x)=0$, and $\hat{C}_{0}$ mimics the top-right component of $\hat{H}_{\text{NOT}}$ in that it performs a bit-flip on the output register if $\text{OR}(x)=1$. These matrices are defined as follows:
\begin{align}
\hat{C}_{0}=\left(
\begin{array}{cccc}
 x_0 & x_1 & \cdots  &x_{m-1} \\
 x_{m-1} & x_{0} & \cdots & x_{m-2}\\
 x_{m-2} & x_{m-1} & \cdots  & x_{m-3}\\
 \vdots & \vdots & \ddots & \vdots \\
 x_{1} & x_{2} & \cdots & x_{0} 
\end{array}
\right), 
\quad
\hat{C}_{1} = \frac{1}{m}\left(
\begin{array}{cccc}
1 & 1 & \cdots  &1 \\
1 & 1 & \cdots  &1 \\
\vdots & \vdots & \ddots  &\vdots \\
1 & 1 & \cdots  &1 
\end{array}
\right)
-\frac{\hat{C}_{0}+\hat{C}_{0}^\dag}{2}.
\end{align}
Note that the non-Hermitian matrix $\hat{C}_{0}$ has rows formed from cyclic shifts of $x$, whereas $\hat{C}_{1}$ is Hermitian. Let us define the uniform superposition $\ket{u}_o=\frac{1}{\sqrt{m}}\sum_{j\in[m]}\ket{j}_o$. It is easy to verify that if at most one bit in $x$ is non-zero, $\hat{C}_{0}\ket{u}_o=\text{OR}(x)\ket{u}_o$. Similarly, $\hat{C}_{1}\ket{u}_o=(\text{OR}(x)\oplus 1)\ket{u}_o$. Thus  $\hat{H}_{\text{OR}}\ket{j}_\text{output}\ket{u}_o=\ket{j\oplus \text{OR}(x)}_\text{output}\ket{u}_o$. Note that $\hat{H}_{\text{OR}}$ has sparsity $2m$, max-norm $\Theta(1)$, and $1$-norm $\Theta(1)$.

Given an $nm$-bit string $x_{0,0}x_{0,1}...x_{0,m-1}x_{1,0}...x_{n-1,m-1}$, the Hamiltonian $\hat{H}_{\text{PARITY}\circ \text{OR}}$ that computes the $n$-bit $\text{PARITY}$ of a number $n$ of $m$-bit $ \text{OR}$ functions is similar to $\hat{H}_{\text{PARITY}}$ in Eq.~\ref{Eq:Ham_Parity}, except that instead of composing with $\text{NOT}$ Hamiltonians defined by the bit $x_j$ for each $j\in[n]$, we compose with $\text{OR}$ Hamiltonians defined by the bits $x_{j,0}x_{j,1}...x_{j,m-1}$  for each $j\in[n]$. By defining $\hat{H}_{\text{OR},j}$ as the Hamiltonian defined by those bits, 
\begin{align}
\label{Eq:Ham_Parity_OR}
\hat{H}_{\text{PARITY}\circ \text{OR}}=\left(\sum_{j\in[n]}\frac{\sqrt{j(N-j+1)}}{N}\ket{j+1}\bra{j}_{sc}\otimes \hat{H}_{\text{OR},j}\right) + \text{Hermitian conjugate}.
\end{align}
On the input state $\ket{0}_s\ket{u}_c\ket{u}_o\ket{0}_\text{output}$, the output of time-evolution $e^{-i\hat{H}_{\text{PARITY}}N\pi/(2d)}\ket{0}_s\ket{u}_c\ket{u}_o\ket{0}_\text{output}=\ket{n}_s\ket{u}_c\ket{u}_o\ket{\bigoplus_j \text{OR}(x_{j,0}x_{j,1}...x_{j,m-1})}_\text{output}$. Thus measuring the $\text{output}$ register returns the parity of $x$
\begin{align}
\text{PARITY}\circ \text{OR}(x)=\bigoplus^{n-1}_{j=0} \text{OR}(x_{j,0}x_{j,1}...x_{j,m-1}),
\end{align}
after time-evolution by $t=n\pi/(2s)$. Note that $\hat{H}_{\text{PARITY}\circ \text{OR}}$ has sparsity $d=2sm$, max-norm $\Theta(1)$, and $1$-norm $\Theta(s)$. It is well-known that the constant-error quantum query complexity of $\text{PARITY}\circ \text{OR}$~\cite{Reichardt2011reflections} is the product of the query complexity of $\text{PARITY}$ with that of $\text{OR}$. As at least $\Omega(\sqrt{m})$ queries are required to compute the $\text{OR}$ of $m$ bits,  $\text{PARITY}\circ \text{OR}(x)$ requires at least $\Omega(n\sqrt{m})$ queries. Thus any algorithm for simulating time-evolution by $\hat{H}_{\text{PARITY}\circ \text{OR}}$ requires at least $\Omega(n\sqrt{m})=\Omega(t\sqrt{ds})$ queries.
\end{proof}
\section{Universality of the Standard-Form}
\label{Sec:Equivalence_Sim_Mea}
We now establish an equivalence between simulation and measurement that justifies our focus on directly manipulating the standard-form encoding of structured Hamiltonians. This equivalence, proven using Thm.~\ref{Thm:Standard_Form_From_Ham_Sim}, allows us to interconvert quantum circuits that implement time-evolution $e^{-i\hat{H}}$ for $\|\hat{H}\|=\mathcal{O}(1)$ and quantum circuits that implement measurement $\|\hat{H}\|$ with only a query overhead logarithmic overhead in error, and a constant overhead in space. 
An application of this result to Hamiltonian simulation is Cor.~\ref{Cor:HamExponentials} for Hamiltonians that is a sum of Hermitian terms, given access only to their exponentials. 

An intuitive picture of when simulation is possible emerges by interpreting the standard-form matrix encoding Def.~\ref{Def:Standard_Form} as a quantum circuit that implements a measurement. To see this explicitly, consider a Hermitian matrix encoded in standard-form-$(\hat{H},\alpha,\hat{U},d)$. Thus for any arbitrary input state $\ket{\psi}_s\in \mathcal{H}_a$, the standard-form applies
\begin{align}
\hat{U}\ket{G}_a\ket{\psi}_s=\frac{1}{\alpha}\ket{G}_a\hat{H}\ket{\psi}_s+\ket{\Phi}_{as}, \quad |\bra{\Phi}_{as}(\ket{G}_a\otimes\hat{I}_s) |= 0,
\end{align}
Note that in this section, we find it helpful to leave $\ket{G}$ explicit, similar to Sec.~\ref{Sec:Standard-form_QSP}. So upon measurement outcome $\ket{G}$ on the ancilla, which occurs with best-case probability $\max_{\ket{\psi}\in\mathcal{H}_s}|\frac{\hat{H}}{\alpha}\ket{\psi}|^2=(\|\hat{H}\|/\alpha)^2$, the measurement operator $\hat{H}/\alpha$ is implemented on the system. As all measurement outcomes orthogonal to $\ket{G}$ do not concern us, we represent their output with some orthogonal unnormalized quantum state $\ket{\Phi}_{as}$. Combined with the Hamiltonian simulation by qubitization results of Thm.~\ref{Thm:Ham_Sim_Qubitization}, one concludes that whenever one has access to a quantum circuit that implements a generalized measurement with measurement operator $\hat{H}/\alpha$ corresponding to one of the measurement outcomes, time-evolution using $\mathcal{O}\left(t\alpha +\frac{\log{(1/\epsilon)}}{\log\log{(1/\epsilon)}}\right)$ queries is possible.

The converse of approximating measurements given $e^{-i\hat{H}t}$ is a standard application of quantum phase estimation. The proof sketch is (1) assume $t$ is chosen such that $\|\hat{H}t\| \le c \le 1$ for some absolute constant $c$ and define $\hat{H}'=\hat{H}t$. (2) Perform quantum phase estimation using $\mathcal{O}(1/\epsilon)$ queries to controlled $e^{-i\hat{H}t}$ to encode the eigenphases $\lambda$ of its eigenstates $\hat{H}'\ket{\lambda}=\lambda\ket{\lambda}$ to precision $\epsilon$ in binary format $\tilde \lambda$ in an $m$-qubit ancilla register $\mathcal{H}_b$, where $m=\mathcal{O}(\log{(1/\epsilon)})$. (3) Perform a controlled rotation on the single-qubit ancilla $\ket{0}_a$ to reduce the amplitude of $\ket{\lambda}$ by factor $\tilde\lambda$. (4) Uncompute the binary register by running quantum phase estimation in reverse. This implements the sequence. 
\begin{align}
\label{Eq:Standard_form_QPE}
\ket{0}_b\ket{0}_a\ket{\lambda}_s
&\rightarrow \ket{\tilde\lambda}_b\ket{0}_a\ket{\lambda}_s
\rightarrow
\ket{\tilde\lambda}_b\left(\tilde\lambda\ket{0}_a+\sqrt{1-|\tilde\lambda |^2}\ket{1}_a)\right)\ket{\lambda}_s 
\\\nonumber
&\rightarrow \ket{0}_b\left(\tilde\lambda\ket{0}_a+\sqrt{1-|\tilde\lambda |^2}\ket{1}_a)\right)\ket{\lambda}_s.
\end{align}
Thus projecting onto the state $\ket{0}_b\ket{0}_a$ implements the measurement operator $\hat{H}'$ with error $\max_\lambda|\lambda - \tilde\lambda| = \mathcal{O}(\epsilon)$, and best-case success probability $\|\hat{H}'\|$. 

As Eq.~\ref{Eq:Standard_form_QPE} is a standard-form encoding of $\hat{H}/\alpha$ with the signal unitary defined by steps (2-4), this establishes one direction in the equivalence between measurement and simulation up to polynomial error and logarithmic space. Ignoring these factors, our study of Hamiltonian simulation reduces to that of generalized measurements except in one edge case: this equivalence does not hold with respect to $t$ when $e^{-i\hat{H}t}$ can be simulated with $o(t)$ queries. However, this case is less interesting as no-fast-forwarding theorems~\cite{Childs2010Limitation} show that $\Omega(t)$ queries are necessary for Hamiltonians that solve generic problems.

We strengthen this equivalence in the opposite direction Thm.~\ref{Thm:Standard_Form_From_Ham_Sim} for approximating measurement operators $\hat{H}'$ using $\log{(1/\epsilon)}$ queries to $e^{-i\hat{H}'}$ and $\mathcal{O}(1)$ ancilla qubits. The idea is to using quantum signal processing techniques to approximate two operator transformations: $\hat{H}_1=\frac{i}{2}(e^{-i\hat{H}'}-e^{i\hat{H}'})$, $\hat{H}_2 = \sin^{-1}{(\hat{H}_1)}$. Thus $\sin^{-1}\left(\frac{i}{2}(e^{-i\hat{H}'}-e^{i\hat{H}'})\right)=\hat{H}'$. All that remains is finding a degree $n$ polynomial approximation to $\sin^{-1}(x)$ with uniform error $n=\mathcal{O}(\log(1/\epsilon))$. However, this seems impossible -- $\sin^{-1}(x)$ is not analytic at $x=\pm 1$, thus its uniform polynomial approximation has degree $n=\mathcal{O}(\text{poly}(1/\epsilon))$. Fortunately, this can be overcome due to the restricted domain $\|\hat{H}t\| \le c$.

\begin{lemma}[Polynomial approximation to $\sin^{-1}(x)$]
\label{Lem.Polynomial_arcsin}
$\forall\;\epsilon \in (0, \mathcal{O}(1)]$, there exists an odd polynomial $p_{\text{arcsin},n}$ of degree $n=\mathcal{O}(\log{(1/\epsilon)})$ such that 
\begin{align}
\max_{ x \in [-1/2,1/2]} \left|p_{\text{arcsin},n}(x)-\sin^{-1}{(x)}\right|\le \epsilon,\quad \text{and}\quad \max_{ x \in [-1,1]} \left|p_{\text{arcsin},n}(x)\right|\le 1.
\end{align}
\end{lemma}
\begin{proof}
We restate Thm.~3 of~\cite{Saff1989polynomial} by Saff and Totik: Let $\beta$ be any number satisfying $\beta > 1$ and let $f\in C^k[-1,1]$ be a piecewise analytic function on $m>0$ closed intervals $[-1,1]=\bigcup^{m}_{j=0}[x_j,x_{j+1}]$, $-1=x_0<x_1<\cdots < x_{m-1}<x_m=1$, where the restriction of $f$ to any of the closed intervals $[x_j,x_{j+1}]$ is analytic, and $f$ is not analytic at each point $x_1,\cdots ,x_{m-1}$. Then there exists constants $g,G>0$ that depend only on $f$, and degree $n>0$ polynomials $p_n$ such that for every $x\in[-1,1]$, $|p_n(x)-f(x)|\le \frac{G}{n^{k+1}}e^{-g n d^{\beta}(x)}$, where $d(x)=\min_{0<j<m}|x-x_j|$. Let us now apply this theorem. Define the function 
\begin{align}
f_{\text{arcsin}}(x)=
\begin{cases}
\sin^{-1}(x), & x \in[ -3/4,3/4], \\
\text{sgn}(x)\sin^{-1}(3/4) & \text{otherwise},
\end{cases}
\end{align}
where $\text{sgn}(x)=\pm x$. $f_{\text{arcsin}}(x)$ is continuous but not differentiable at $x=\pm 3/4$. Thus $f\in C^0[-1,1]$, $\max_{x\in[-1/2,1/2]} d(x)\ge 1/4$, and there exist absolute constants $G',g'>0$ and polynomials $p_n$ such that $\max_{x\in[-1/2,1/2]}|p_n(x)-f_{\text{arcsin}}(x)|\le \frac{G'}{n}e^{-g' n / 4^\beta}=\epsilon$. Hence $n = \mathcal{O}(\log{(1/\epsilon)})$. Since $e^{-g' n d^{\beta}(x)} \le 1$ and $|\sin^{-1}(3/4)|< 0.85$, there exists a constant $n_0> 0$ such that for all $n>n_0$, $\max_{x\in[-1,1]}|f_{\text{arcsin}}(x)-p_n(x)|\le 0.15$ thus $|p_n(x)|\le 1$. If $p_n(x)$ is not odd, replace it with its antisymmetric component $p_n\leftarrow \frac{p_n(x)-p_n(-x)}{2}$ which is odd with at worst the same error. Now let $p_{\text{arcsin},n} = p_n$.
\end{proof}

We now apply this polynomial approximation of $\sin^{-1}(x)$ to the proof of Thm.~\ref{Thm:Standard_Form_From_Ham_Sim}.
\begin{proof}[Proof of Thm.~\ref{Thm:Standard_Form_From_Ham_Sim}]
The transformation from time evolution $e^{-i\hat{H}t}$ to measurement $\hat{H}t$ takes three steps.
First, encode the Hermitian operator $\hat{H}_1=\sin{(\hat{H}t)}$ in standard-form. This can be done with one query to the controlled time-evolution operator $\hat{U}_0=\ket{0}\bra{0}\otimes \hat{I} + \ket{1}\bra{1}\otimes e^{-i\hat{H}t}$ and its inverse $\hat{U}_0^\dag$:
\begin{align}
\hat{U}_1&= \hat{U}^\dag_0(\hat{\sigma}_x\otimes \hat{I})\hat{U}_0 = \ket{1}\bra{0}\otimes e^{i\hat{H}t} + \ket{0}\bra{1}\otimes e^{-i\hat{H}t}, \quad \ket{G}=e^{i\hat{\sigma}_x\pi/4}\ket{0},
\\\nonumber
\hat{H}_1&=(\bra{G}\otimes \hat{I})\hat{U}_1(\ket{G}\otimes \hat{I})=\sin{(\hat{H}t)}.
\end{align}
Second, approximate $\hat{H}_2=\sin^{-1}(\hat{H}_1)$ using quantum signal processing. As the polynomial $p_{\text{arcsin},N}(x)$ of Lem.~\ref{Lem.Polynomial_arcsin} satisfies the conditions of Thm.~\ref{Thm:QSP_B}, the operator transformation $\hat{H}_{\text{lin}}t=p_{\text{arcsin},N}[\hat{H}_1]$ can be implemented exactly with $\mathcal{O}(N)$ queries to $\hat{U}_0$. This encodes $\hat{H}_{\text{lin}}t$ in standard-form with normalization $1$. Now choose $t$ such that $\|\hat{H}t\|\le c = 1/2$. Then $\|\sin{(\hat{H}t)}\|\le \|\hat{H}t\| \le 1/2$ as $\sin(x)\le x$. Third, evaluate the approximation error using Lem.~\ref{Lem.Polynomial_arcsin}. $\|\hat{H}_{\text{lin}}t-\hat{H}t\| \le \max_{x\in [-1/2,1/2]}|p_{\text{arcsin},N}(x)-\sin^{-1}(x)| \le \epsilon$, for $N=\mathcal{O}(\log{(1/\epsilon)})$. 
\end{proof}

Incidentally, the equivalance between simulation and measurement also provides a simulation algorithm for Hamiltonians built from a sum of $d$ Hemritian component $\hat{H}=\sum^d_{j=1}\hat{H}_j$, where one only has access to these components through an oracle for their controlled exponentials $e^{-i\hat{H}_j t_j}$, for any $t_j\in\mathbb{R}$. Though results with similar scaling can be obtained through the techniques of compressed fractional queries~\cite{Berry2014}, this approach has two main advantages. First, the queries $\hat{H}_j$ are not restricted to only have eigenvalues $\pm 1$. Second, it is significantly simpler both in concept and in implementation.
\begin{proof}[Proof of Cor.~\ref{Cor:HamExponentials}]
From Thm.~\ref{Thm:Standard_Form_From_Ham_Sim}, $\mathcal{O}(\log(1/\epsilon_1))$ queries to $\hat{U}$ suffice to encode $\hat{H}_{\text{controlled}}=\sum^d_{j=1}\ket{j}\bra{j}_a\otimes \hat{H}'_j=(\bra{G'}_b\otimes\hat{I}_{as})\hat{U}'(\ket{G'}_b\otimes\hat{I}_{as})$ in standard-form with some state $\ket{G'}_b$ and signal oracle $\hat{U}'$, where $\max_{j}\|\hat{H}'_j-\hat{H}_j\|\le \epsilon_1$ and $\hat{H}_{\text{controlled}}$ acts on the system register $s$. Thus $(\bra{G}_a\bra{G'}_b\otimes\hat{I}_s)\hat{U}'(\ket{G}_a\ket{G'}_b\otimes\hat{I}_s)=\hat{H}_{\text{approx}}/\alpha$ encodes $\hat{H}_{\text{approx}}$ in standard-form where $\|\hat{H}_{\text{approx}}-\hat{H}\|=\|\sum^d_{j=1}\alpha_j(\hat{H}'_j-\hat{H}_j)\|\le \sum^d_{j=1}\alpha_j\|\hat{H}'_j-\hat{H}_j\|\le \alpha \epsilon_1$. Using the fact $\|e^{i \hat A}-e^{i \hat B}\|\le \|\hat A-\hat B\|$~\cite{Berry2014}, we have $\|e^{-i\hat{H}'t}-e^{-i\hat{H}t}\|\le t \alpha \epsilon_1$. By applying Thm.~\ref{Thm:Ham_Sim_Qubitization}, $e^{-i\hat{H}_{\text{approx}}t}$ can be approximated with error $\epsilon_2$ using $\mathcal{O}(t \alpha+\frac{\log{(1/\epsilon_2)}}{\log\log{(1/\epsilon_2)}})\mathcal{O}(\log(1/\epsilon_1))$ queries to $\hat{U}$. By the triangle inequality, this approximates $e^{-i\hat{H}t}$ with error $\le t\alpha\epsilon_1+\epsilon_2$. Thus choose $\epsilon_1 = \frac{\epsilon}{2t\alpha}$ and $\epsilon_2=\epsilon/2$.
\end{proof}

\section{Conclusions}
\label{Sec:Amp_concluson}
We have combined ideas from qubitization and quantum signal processing to solve, in a general setting, the uniform spectral amplification problem of implementing a low-distortion expansion of the spectrum of Hamiltonians. One most surprising application of our results is the simulation of sparse Hamiltonians where we obtain an algorithm with linear complexity in $\mathcal{O}(t(d\Lambda_{\text{max}}\Lambda_1)^{1/2})$, excluding logarithmic factors. This is particularly important as the best-case scaling $\mathcal{O}(\sqrt{d})$ is essential to an optimal realization of the fundamental quantum search algorithm. However, this improvement also appears impossible as prior art claims that $\Theta(td\|\hat{H}\|_{\text{max}})$ queries is optimal. Nevertheless, the two are actually consistent. In the situation where information on $\|\hat{H}\|_1$ is unavailable, previous results are recovered as one may simply choose the worst-case $\Lambda_1=d\Lambda_{\text{max}}=d\|\hat{H}\|_{\text{max}}$. This naturally leads to the question of whether further improvement is possible. For instance, if information on $\|\hat{H}\|$ rather than $\|\hat{H}\|_1$ is made available, our lower bound is consistent with the stronger statement of $\Omega(t(d\|\hat{H}\|_{\text{max}}\|\hat{H}\|)^{1/2})$ queries.

More generally, the universality of our results motivates related future directions. Thus far, a large number of common oracles used to describe Hamiltonians to quantum computers map to the standard-form without much difficultly. Rather than focusing on improving Hamiltonian simulation algorithms, perhaps an emphasis on improving the quality of encoding, through a reduced normalization constant, would be more insightful, easier, and also lead to greater generality. Combined with the extremely low overhead of our techniques, algorithms obtained in this manner could be practical on digital quantum computers sooner rather than later.


\section{Acknowledgments}
G.H. Low is funded by the NSF RQCC Project No.1111337 and ARO quantum algorithms project. We thank Aram Harrow and Robin Kothari for suggesting $\text{PARITY}\circ\text{OR}$ as a possible lower bound.

\appendix

\section{Polynomial Approximations to a Truncated Linear Function}
\label{Sec:Polynomials_Amplitude_Multiplication}
The proof of Thm.~\ref{Cor:Operator_Amplification} and Thm.~\ref{Thm:Linear_Amplitude_Amplification} require a polynomial approximation $p_{\text{lin},\Gamma,n}$ to the truncated linear function
\begin{align}
\label{Eq:Linear_target_function_Appendix}
f_{\text{lin},\Gamma}(x)=
\begin{cases}
 \frac{x}{2\Gamma}, & |x| \in [0, \Gamma], \\
\in [-1,1], & |x| \in (\Gamma,1].
\end{cases}
\end{align}
The remainder of this section is dedicated to constructively proving the existence of $p_{\text{lin},\Gamma,n}$ with the following properties:
\begin{theorem}[Polynomial for linear amplitude amplification]
\label{Thm.Polynomial_LAA}
$\forall\; \Gamma \in [0,1/2]$, $\epsilon \in(0, \mathcal{O}(\Gamma)]$, there exists an odd polynomial $p_{\text{lin},\Gamma,n}$ of degree $n=\mathcal{O}(\Gamma^{-1}\log{(1/\epsilon)})$ such that
\begin{align}
\forall\; {x\in[- \Gamma,\Gamma]},\; \left|p_{\text{lin},\Gamma,n}(x)-\frac{x}{2\Gamma}\right|\le \frac{\epsilon|x|}{2\Gamma} \quad\text{and}\quad \max_{x\in [-1,1]}|p_{\text{lin},\Gamma,n}(x)|\le 1.
\end{align}
\end{theorem}
As close-to-optimal uniform polynomials approximations may be obtained by the Chebyshev truncation of entire functions, our strategy is to find an entire function $f_{\text{lin},\Gamma,\epsilon}$ that approximates $f_{\text{lin},\Gamma}$ over the domain $x\in[-\Gamma,\Gamma]$ with error $\epsilon$. We construct $f_{\text{lin},\Gamma,\epsilon}(x)$ in three steps. First, approximate the sign function $\text{sgn}(x)$ with an error functions, which is entire. Second, approximate the rectangular function $\text{rect}(x)$ with a sum of two error function $\frac{1}{2}\left(\text{erf}(k (x+\delta))+\text{erf}(k (-x+\delta))\right)$. Third, multiply this by $\frac{x}{2\Gamma}$ to approximate $f_{\text{lin},\Gamma,\epsilon}(x)$ with some error $\epsilon$. The approximation error of this sequence is described by Lems.~\ref{Lem:Entire_Sgn}, \ref{Lem:Entire_Rect}, \ref{Lem:Entire_Linear}:
\begin{lemma}[Entire approximation to the sign function $\text{sgn}(x)$]
\label{Lem:Entire_Sgn}
$\forall\;\kappa > 0, x\in\mathbb{R},\epsilon\in(0,\sqrt{2/e\pi}]$, let $k = \frac{\sqrt{2}}{\kappa}\log^{1/2}{(\frac{2}{\pi\epsilon^2})}$. Then the function $f_{\text{sgn},\kappa,\epsilon}(x)=\text{erf}(kx)$ satisfies
\begin{align}
\begin{aligned}
1&\ge|f_{\text{sgn},\kappa,\epsilon}(x)|, 
\\
\epsilon&\ge\max_{|x|\ge \kappa/2}|f_{\text{sgn},\kappa,\epsilon}(x)-\text{sgn}(x)|,
\end{aligned}
\quad
\begin{aligned}
\text{sgn}(x)=
\begin{cases}
1, &  x > 0, \\ 
-1, & x < 0, \\
1/2, & x = 0.
\end{cases}
\end{aligned}
\end{align}
\end{lemma}
\begin{proof}
We apply elementary upper bounds on the complementary error function $\text{erfc}(x)=1-\text{erf}(x)=\frac{2}{\sqrt{\pi}}\int_x^\infty e^{-y^2}dy\le \frac{2}{\sqrt{\pi}}\int_x^\infty\frac{y}{x} e^{-y^2}dy=\frac{1}{x\sqrt{\pi}}e^{-x^2}$ for any $x>0$. Thus $\max_{x\ge \kappa/2}|\text{erf}(kx)-1|\le \frac{2}{k\kappa\sqrt{\pi}}e^{-(k\kappa)^2/4}=\epsilon$ and similarly for $x \le -\kappa/2$. This is solved by $k=\frac{1}{\kappa}\sqrt{2W(\frac{2}{\pi\epsilon^2})}$ where $W(x)$ is the Lambert-$W$ function. From the upper bound $\log{x}-\log{\log{x}}\le W(x)\le \log{x}-\frac{1}{2}\log{\log{x}}$ for $x\ge e$~\cite{Hoorfar2008LambertW}, any choice of $k \ge \frac{\sqrt{2}}{\kappa}\log^{1/2}{(\frac{2}{\pi\epsilon^2})}\ge \frac{\sqrt{2}}{\kappa}$ where $\frac{2}{\pi\epsilon^2}\ge e$ ensures that $\text{erf}(k x)$ is close to $\pm 1$ over $x\ge\kappa/2$.
\end{proof}
\begin{lemma}[Entire approximation to the rect function]
\label{Lem:Entire_Rect}
$\forall\;\kappa > 0,\; w>0,\; x\in\mathbb{R},\;\epsilon\in(0,\sqrt{2/e\pi}]$, let $k = \frac{\sqrt{2}}{\kappa}\log^{1/2}{(\frac{2}{\pi\epsilon^2})}, \delta = (w+\kappa)/2$. Then the function
$f_{\text{rect},w,\kappa,\epsilon}(x)=\frac{1}{2}\left(\text{erf}(k (x+\delta))+\text{erf}(k (-x+\delta))\right)$ satisfies
\begin{align}
\begin{aligned}
1 &\ge |f_{\text{rect},w, \kappa,\epsilon}(x)|,
\\
\epsilon &\ge \max_{|x| \in [0,w/2]\cup[w/2+\kappa,\infty]} |f_{\text{rect},w,\kappa,\epsilon}(x)-\text{rect}(x/w)|, 
\end{aligned}
\quad
\text{rect}(x)=
\begin{cases}
1, &  |x| < 1/2, \\ 
0, & |x| > 1/2, \\
1/2, & |x| = 1/2.
\end{cases}
\end{align}
\end{lemma}
\begin{proof}
This follows from the definition of the rect function $\text{rect}(x/w)=\frac{1}{2}(\text{sgn}(x+w/2)+\text{sgn}(-x+w/2))$. Thus we choose $\delta = (w+\kappa)/2$ and apply the error estimates of Lem.~\ref{Lem:Entire_Sgn}.
\end{proof}
\begin{lemma}[Entire approximation to the truncated linear function]
\label{Lem:Entire_Linear}
$\forall\;\Gamma > 0,\; x\in\mathbb{R},\;\epsilon\in(0,\sqrt{2/e\pi}]$, the function
$f_{\text{lin},\Gamma,\epsilon}(x)=\frac{x}{2\Gamma}f_{\text{rect},2\Gamma, 2\Gamma,\epsilon}(x)$ satisfies
\begin{align}
|f_{\text{lin},\Gamma,\epsilon}(x)|\le 1,
\quad
\max_{|x| \in [0,\Gamma]} \left|f_{\text{lin},\Gamma,\epsilon}(x)-\frac{x}{2\Gamma}\right|\le \frac{|x|\epsilon}{2\Gamma}.
\\\nonumber
\end{align}
\end{lemma}
\begin{proof}
Consider the domain $|x|\in[0,\Gamma]$. There, Lem.~\ref{Lem:Entire_Rect} gives the approximation error $|f_{\text{rect},2\Gamma, 2\Gamma,\epsilon}(x)-1|\le \epsilon$. Multiplying both sides by $\frac{x}{2\Gamma}$ gives the stated result. Now consider the domain $|x|\in[0,2\Gamma]$. There, $|f_{\text{rect},2\Gamma, 2\Gamma,\epsilon}(x)|\le 1$ and $|\frac{x}{2\Gamma}|\le 1$. Thus the product is bounded by $\pm1$. Now consider the domain $x\ge 2\Gamma$. 
Let us maximize 
$f_{\text{lin},\Gamma,\epsilon}(x)$ over $x,\epsilon$. Define $1/\epsilon'=\sqrt{\log{(\frac{2}{\pi\epsilon^2})}}\ge 1$. Thus
$f_{\text{lin},\Gamma,\epsilon}(x)=\frac{x}{4\Gamma}\left(\text{erf}(\frac{x+2\Gamma}{\sqrt{2}\Gamma\epsilon'})+\text{erf}(\frac{2\Gamma-x}{\sqrt{2}\Gamma\epsilon'})\right)$. We make use of the upper bounds
$\text{erfc}(x)=1-\text{erf}(x)\le\frac{1}{x\sqrt{\pi}}e^{-x^2}$ 
and 
$\text{erfc}(x)\le e^{-x^2}$. 
The first term has the bounds 
$1\ge \text{erf}(\frac{x+2\Gamma}{2\Gamma \epsilon'}) \ge 1-\frac{1}{\frac{x+2\Gamma}{\sqrt{2}\Gamma}\sqrt{\pi}\epsilon'}e^{-(\frac{x+2\Gamma}{\sqrt{2}\Gamma\epsilon'})^2}\ge 1-\frac{1}{\sqrt{8\pi}\epsilon'}e^{-(\frac{x+2\Gamma}{\sqrt{2}\Gamma\epsilon'})^2}$. 
The second term has the bounds 
$-1+e^{-(\frac{2\Gamma-x}{\sqrt{2}\Gamma\epsilon'})^2} \ge\text{erf}(\frac{2\Gamma-x}{\sqrt{2}\Gamma\epsilon'})\ge -1$. 
By adding these together and extremizing the upper and lower bounds separately, 
$f_{\text{lin},\Gamma,\epsilon}(x) \in [-0.0011,0.56]$ independent of $\Gamma$ and for all $\epsilon'\in[0,1]$. These bounds apply to $x\le 2\Gamma$ with a minus sign as $f_{\text{lin},\Gamma,\epsilon}(x)$ is an odd function.
\end{proof}

However, the required polynomial must have a non-uniform error $\left|p_{\text{lin},\Gamma,n}(x)- \frac{x}{2\Gamma}\right|\le \frac{|x|}{2\Gamma}\epsilon$, proportional to $|x|$. Though $f_{\text{lin},\Gamma,\epsilon}$ of Lem.~\ref{Lem:Entire_Linear} has that property, its Chebyshev truncation results in a worst-case uniform error $\epsilon$ for all values of $x$. This is overcome by approximating $p_{\text{lin},\Gamma,n}(x)$ as the product of a Chebyshev truncation of the entire approximation to $\text{rect}(x)$ and with $\frac{x}{2\Gamma}$. We now evaluate the scaling of the degree of the Chebyshev truncation of $f_{\text{lin},\Gamma,\epsilon}$ in Lem.~\ref{Lem:Entire_Rect} with respect to their parameters and the desired approximation error. 

Our starting point is the Jacobi-Anger expansion of the exponential decay function:
\begin{align}
\label{Eq:Jacobi-Anger}
f_{\text{exp},\beta}(x)= e^{-\beta(x+1)}=e^{-\beta}\left(I_0(\beta)+2\sum^\infty_{j=1} I_j(\beta) T_j(-x)\right),
\end{align}
where $I_j(\beta)$ are modified Bessel functions of the first kind. The domain of this function and all the following are assumed to be $x\in[-1,1]$. By truncating this expansion above $j>n$, we obtain a degree $n$ polynomial approximation $p_{\text{exp},\beta,n}(x)$ with truncation error $\epsilon_{\text{exp},\beta,n}$:
\begin{align}
p_{\text{exp},\beta,n}(x)&= e^{-\beta}\left(I_0(\beta)+2\sum^n_{j=1} I_j(\beta) T_j(-x)\right),\\
\label{Eq:error_exp}
\epsilon_{\text{exp},\beta,n} &= \max_{x\in[-1,1]}| p_{\text{exp},\beta,n}-f_{\text{exp},\beta}|
=
2e^{-\beta}\sum^\infty_{j=n+1} |I_j(\beta)|. 
\end{align}
Note that the equality in the rightmost term of Eq.~\ref{Eq:error_exp} arises as all the coefficients $I_j(\beta)\ge0$ when $\beta\ge 0$. Thus $\epsilon_{\text{exp},\beta,n}$ is maximized $|T_j(-x)|$ are all simultaneously maximized, which occurs at $x=-1\Rightarrow T_j(-x)=1$.
By solving $\epsilon_{\text{exp},\beta,n}$, one can in principle obtain the required degree $n$ as a function of $\beta,\epsilon$. 

Error estimates for various degree $n$ polynomial approximations to the exponential decay function can be found in the literature. However these approximations are constructed using other methods. For instance, a Taylor expansion leads to scaling linear in $\beta$, and none explicitly bound the sum $\epsilon_{\text{exp},\beta,n}$. 
Fortunately, one particular error estimate in prior art is good enough and can be shown, with a little work, to implicitly bound $\epsilon_{\text{exp},\beta,n}$. We first sketch the proof of this estimate, then later show how it bounds $\epsilon_{\text{exp},\beta,n}$.
\begin{lemma}[Polynomial approximation to exponential decay $e^{-\beta(x+1)}$ adapted from~\cite{Sachdeva2014Exp}]
\label{Lem:Exponential_error_Sachdeva}
$\forall  \beta>0, \epsilon\in(0,1/2],$ there exists a polynomial $p_n$ of degree $n=\lceil\sqrt{2\lceil\max[\beta e^2,\log{(2/\epsilon)}]\rceil\log{(4/\epsilon)}}\rceil$ such that 
\begin{align}
\max_{x \in [-1,1]}| p_n(x)-e^{-\beta(x-1)}| \le \epsilon.
\end{align}
\end{lemma}
\begin{proof}
Consider the Chebyshev expansion of the monomial $x^s=2^{1-s}\sum'^s_{j=0, s-j\;\text{even}}\binom{s}{(s-j)/2}T_j(x)=\mathbb{E}[T_{D_s}(x)]$, where $s \le 0$ is an integer and $\sum'_j$ means the $j=0$ term is halved. The representation an an expectation over the random variable $D_s=\sum^s_{j=1}Y_j$ where $Y_j=\pm 1$ with equal probabilities follows from the identity $xT_j(x)=\frac{1}{2}(T_{j-1}(x)+T_{j+1}(x))$. They show that the Chebyshev truncation of the monomial has error
\begin{align}
\label{Eq:Chebyshev_Monomial} 
p_{\text{mon},s,n}(x)&= 2^{1-s}\sideset{}{'}\sum^{\min(s,n)}_{j=0, n-j\;\text{even}}\binom{s}{(s-j)/2}T_j(x),
\\ \nonumber
\epsilon_{\text{mon},s,n}&=\max_{x\in[-1,1]}|p_{\text{mon},s,n}(x)-x^s|
\le
2^{1-s}\sideset{}{'}\sum^{s}_{j=n+1, n-j\;\text{even}}\binom{s}{(s-j)/2}\le 2e^{-n^2/(2s)},
\end{align} 
which follows from the triangle inequality with $|T_j(x)|\le1$ and the Chernoff bound $P(|D_s|\ge n)\le 2 e^{-n^2/(2s)}$. By replacing each monomial up to degree $t$ in the Taylor expansion of 
$e^{-\beta(x-1)}=e^{-\beta}\sum^\infty_{j=0}\frac{(-\beta)^j}{j!}x^j$ 
with $\tilde p_{\text{mon},s,n}$, they obtain the degree $n$ polynomial  
$\tilde p_n(x)=e^{-\beta}\sum^t_{j=0}\frac{(-\beta)^j}{j!}\tilde p_{\text{mon},j,n}(x)$. 
They show the error of this approximation is split into two terms: 
\begin{align}
\label{Eq:Chebyshev_Exponential_Sachdeva}
\epsilon_{\text{sach},\beta,n}&=\max_{x\in[-1,1]}|\tilde p_n(x)-e^{-\beta(x-1)}|\le \epsilon_1 + \epsilon_2,
\\\nonumber
\epsilon_1 & =2e^{-\beta}\sum^t_{j=n+1}\frac{(\beta/2)^j}{j!}|p_{\text{mon},j,n}-x^j| \le 2e^{-n^2/(2t)},
\quad
\epsilon_2 = 2e^{-\beta}\left|\sum^\infty_{j=t+1}\frac{(\beta/2)^j}{j!}x^j\right|\le 2e^{-\beta - t}.
\end{align}
By choosing $n=\lceil \sqrt{2t \log{(4/\epsilon)}} \rceil$ and $t=\lceil\max\{\beta e^2,\log{(4/\epsilon)}\}\rceil$, $\epsilon_1+\epsilon_2 \le \epsilon$.
\end{proof}

We now demonstrate how this upper bounds $\epsilon_{\text{exp},\beta,n}$.
\begin{lemma}[Chebyshev truncation error of exponential decay $e^{-\beta(x+1)}$]
\label{Lem:Exponential_error}
$\forall\;\beta>0, \epsilon\in(0,1/2]$, the choice
$n=\lceil\sqrt{2\lceil\max[\beta e^2,\log{(2/\epsilon)}]\rceil\log{(4/\epsilon)}}\rceil = \mathcal{O}(\sqrt{(\beta+\log{(1/\epsilon)})\log{(1/\epsilon)}})$, guarantees that 
$\epsilon_{\text{exp},\beta,n} \le \epsilon$.
\end{lemma}
\begin{proof}
This result follows essentially from how the truncating the Jacobi-Anger expansion in Eq.~\ref{Eq:Jacobi-Anger} discards fewer coefficients that are all positive than the procedure of Thm.~\ref{Lem:Exponential_error_Sachdeva}. Hence the maximum truncation error occurs at $x=1$ and is monotonically increasing with the number of coefficients omitted in the truncation. Observe that the first inequality in Eq.~\ref{Eq:Chebyshev_Monomial} is actually an equality $\epsilon_{\text{mon},s,n}=
2^{1-s}\sideset{}{'}\sum^{s}_{j=n+1, n-j\;\text{even}}\binom{s}{(s-j)/2}$. This follows from the same logic as Eq.~\ref{Eq:error_exp} -- all coefficients are positive, thus the maximum error occurs at $x=1$, which simultaneously maximizes all $ T_j(x=1)=1$. Similarly, the first inequality in Eq.~\ref{Eq:Chebyshev_Exponential_Sachdeva} is also actually an equality. Let us express the truncation error of $\epsilon_{\text{sach},\beta,n}$ as a Chebyshev expansion in full
\begin{align}
\epsilon_{\text{sach},\beta,n}=&2e^{-\beta}\max_{x\in[-1,1]}\left|\sum^t_{j=n+1}\frac{(\beta/2)^j}{j!}\sideset{}{'}\sum^j_{k=n+1, j-k\;\text{even}}\binom{j}{(j-k)/2}T_k(x)
\right. \\ \nonumber
&+
\left.
\sum^\infty_{j=t+1}\frac{(\beta/2)^j}{j!}\sideset{}{'}\sum^j_{k=0, j-k\;\text{even}}\binom{j}{(j-k)/2}T_k(x)
\right|.
\end{align}
Note that we have used $(-\beta)^jT_k(-x)=\beta^jT_k(x)$ as all pairs $j-k$ are even. Thus $\epsilon_{\text{sach},\beta,n}$ is maximized at $T_k(x=1)=1$ in the sum above. This can be compared with 
\begin{align}
\epsilon_{\text{exp},\beta,n} &= \max_{x\in[-1,1]}\left|2e^{-\beta}\sum^\infty_{j=n+1} I_j(\beta)T_j(x)\right|
=
\epsilon_{\text{sach},\beta,n} - 2e^{-\beta}\sum^\infty_{j=t+1}\frac{(\beta/2)^j}{j!}\sideset{}{'}\sum^n_{k=0, j-k\;\text{even}}\binom{j}{(j-k)/2}
\\\nonumber
& \le \epsilon_{\text{sach},\beta,n}.
\end{align}
More intuitively, both $\epsilon_{\text{exp},\beta,n}$ and $\epsilon_{\text{sach},\beta,n}$ sum over all coefficients $j > n$ in the Chebyshev expansion, but $\epsilon_{\text{sach},\beta,n}$ in addition sums over some positive coefficients corresponding to $j \le n$.
Thus the upper bound of Lem.~\ref{Lem:Exponential_error_Sachdeva} on $\epsilon_{\text{sach},\beta,n}$ applies to $\epsilon_{\text{exp},\beta,n}$. 
\end{proof}

In the following, we will bound all errors of our polynomial approximations in terms $\epsilon_{\text{exp},\beta,n}$, a partial sum over Bessel functions.

\begin{corollary}[Polynomial approximation to the Gaussian function $e^{-(\gamma x)^2}$]
$\forall\gamma \ge 0, \epsilon\in(0,1/2]$ the even polynomial $ p_{\text{gauss},\gamma,n}$ of even degree $n=\mathcal{O}(\sqrt{(\gamma^2 + \log{(1/\epsilon)})\log{(1/\epsilon)}})$ satisfies
\begin{align}
\label{Eq:p_tilde_gauss}
p_{\text{gauss},\gamma,n}(x)&=  p_{\text{exp},\gamma^2/2,n/2}(2x^2-1) = e^{-\gamma^2/2}\left(I_0(\gamma^2/2)+2\sum^{n/2}_{j=1} I_j(\gamma^2/2) (-1)^{j}T_{2j}(x)\right),
\\ \nonumber
\epsilon_{\text{gauss},\gamma,n}&=\max_{x \in [-1,1]}| p_{\text{gauss},\gamma,n}(x)-e^{-(\gamma x)^2}| = \epsilon_{\text{exp},\gamma^2/2,n/2}\le \epsilon.
\end{align}
\end{corollary}
\begin{proof}
This follows from Eq.~\ref{Eq:Jacobi-Anger} by a simple change of variables.
Let 
$x'=T_2(x)=2x^2-1$, $\gamma^2 = 2\beta$.
Thus 
$e^{-\beta(x'+1)}=e^{-(\gamma x)^2}$. 
As 
$2x^2-1: [-1,1]\Rightarrow [-1,1]$ maps the domain of $e^{-(\gamma x)^2}$ to that of $f_{\text{exp},\beta}(x)$, the definition Eq.~\ref{Eq:p_tilde_gauss} results. Using the Chebyshev semigroup property $T_{j}(\pm T_{2}(x))=(\pm 1)^{j}T_{2j(x)}$, $p_{\text{gauss},k,n}$ is an even polynomial of degree $n$ and its approximation error is obtained by substitution into Eq.~\ref{Eq:error_exp}.
\end{proof}

A polynomial approximation to the error function follows immediately by integrating $p_{\text{gauss},\gamma,n}$.
\begin{corollary}[Polynomial approximation to the error function $\text{erf}(kx)$]
$\forall k > 0, \epsilon\in(0,\mathcal{O}(1)]$ the odd polynomial $p_{\text{erf},k,n}$ of odd degree $n=\mathcal{O}(\sqrt{(k^2 + \log{(1/\epsilon)})\log{(1/\epsilon)}})$ satisfies
\begin{align}
\label{Eq:p_tilde_erf}
p_{\text{erf},k,n}(x)&= \frac{2 k e^{-k^2/2}}{\sqrt{\pi}}\left(I_0(k^2/2)x+\sum^{(n-1)/2}_{j=1} I_j(k^2/2) (-1)^{j}
\left(\frac{T_{2j+1}(x)}{2j+1}-\frac{T_{2j-1}(x)}{2j-1}\right)
\right),
\\ \nonumber
\epsilon_{\text{erf},k,n}&=\max_{x \in [-1,1]}| p_{\text{erf},k,n}(x)-\text{erf}(kx)| \le \frac{4 k}{\sqrt{\pi}n}\epsilon_{\text{gauss},k,n-1}\le\epsilon.
\end{align}
\end{corollary}
\begin{proof}
From the definition of the error function $\text{erf}(kx)=\frac{2}{\pi}\int_0^{kx}e^{-x^2}=\frac{2k}{\sqrt{\pi}}\int_0^x e^{-(kx)^2}dx$,
the polynomial $ p_{\text{erf},k,n+1}(x)=k\int_0^x p_{\text{gauss},k,n}(x)dx$ follows directly from integrating Eq.~\ref{Eq:p_tilde_gauss} term-by-term using the identity $\int_0^x T_{j}(x) dx
=\frac{1}{2}\left(\frac{T_{j+1}(x)}{j+1}-\frac{T_{j-1}(x)}{j-1}\right)$.
The error of the remaining terms is bounded though
\begin{align}
 \epsilon_{\text{erf},k,n}
&\le 
\frac{2 k e^{-k^2/2}}{\sqrt{\pi}}\left|\sum^{\infty}_{j=(n+1)/2} I_j(k^2/2) (-1)^{j}
\left(\frac{T_{2j+1}(x)}{2j+1}-\frac{T_{jn-1}(x)}{2j-1}
\right)\right|
\\\nonumber
&\le
\frac{2 k e^{-k^2/2}}{\sqrt{\pi}}\sum^{\infty}_{j=(n+1)/2} |I_j(k^2/2)|
\left(\frac{1}{2j+1}+\frac{1}{2j-1}\right)
\\\nonumber
&\le
\frac{4 k e^{-k^2/2}}{\sqrt{\pi}n}\sum^{\infty}_{j=(n+1)/2} |I_j(k^2/2)|
=\frac{4 k}{\sqrt{\pi}n}\epsilon_{\text{gauss},k,n-1}.
\end{align}
The error of $\epsilon_{\text{erf},k,n}\le \frac{4 k}{\sqrt{\pi}n}\epsilon_{\text{exp},k^2/2,(n-1)/2}$. However, $n=\Omega(k \log^{1/2}{(1/\epsilon)})$. Thus $\frac{k}{n}=\mathcal{O}(\log^{-1/2}{(1/\epsilon)})=\mathcal{O}(1)$ and does not make the scaling any worse.
\end{proof}

A polynomial approximation to the shifted error function follows by a change of variables.
\begin{corollary}[Polynomial approximation to the shifted error function $\text{erf}(k(x-\delta))$]
\label{Cor:Shifted_Sgn}
$\forall k > 0, \delta \in[-1,1], \epsilon\in(0,\mathcal{O}(1)]$ the polynomial $p_{\text{erf},k,\delta,n}(x)= p_{\text{erf},2k,n}((x-\delta)/2)$ of odd degree $n=\mathcal{O}(\sqrt{(k^2 + \log{(1/\epsilon)})\log{(1/\epsilon)}}$ satisfies
\begin{align}
\label{Eq:p_tilde_erf_shifted}
 \epsilon_{\text{erf},k,\delta,n}&=\max_{x \in [-1,1]}|p_{\text{erf},k,\delta,n}(x)-\text{erf}(k(x-\delta))| \le
\epsilon_{\text{erf},2k,n}\le\epsilon.
\end{align}
\end{corollary}
\begin{proof}
This follows trivially from  $\text{erf}(k(x-\delta))=\text{erf}(2k\frac{x-\delta}{2})$. Note that we have doubled the degree of our polynomials in order to double the width of the domain, which we exploit to allows translations.
\end{proof}
This polynomial approximation of the shifted error function is the basic ingredient we use to construct more complicated functions $\text{sgn}$ and $\text{rect}$ through Lems.\ref{Lem:Entire_Sgn},\ref{Lem:Entire_Rect}.
\begin{corollary}[Polynomial approximation to the sign function $\text{sgn}(x-\delta)$]
$\forall\;\kappa > 0, \delta \in[-1,1], \epsilon\in(0,\mathcal{O}(1)]$ the polynomial $p_{\text{sgn},\kappa,\delta,n}(x)=p_{\text{erf},k,\delta,n}(x)$ of odd degree $n=\mathcal{O}(\frac{1}{\kappa}\log{(1/\epsilon)})$, where $k = \frac{\sqrt{2}}{\kappa}\log^{1/2}{(\frac{2}{\pi\epsilon_1^2})}$, satisfies
\begin{align}
\label{Eq:p_tilde_sgn_shifted}
 \epsilon_{\text{sgn},\kappa,\delta,n}&=\max_{x \in [-1,\delta-\kappa/2]\cup [\delta+\kappa/2,1]}|p_{\text{erf},k,\delta,n}(x)-\text{sgn}(x-\delta)| \le
 \epsilon_{\text{erf},k,\delta,n}+\epsilon_1 
\le 2\epsilon_{\text{erf},k,\delta,n}\le \epsilon.
\end{align}
\end{corollary}
\begin{proof}
The equation for $k$ comes from Lem.~\ref{Lem:Entire_Sgn}. We then choose $\epsilon_1=\epsilon_{\text{erf},k,\delta,n}$ which defines an implicit equation for $\epsilon_1$ and doubles the error.
\end{proof}
\begin{corollary}[Polynomial approximation to the rectangular function $\text{rect}(x/w)$]
$\forall\;\kappa \in (0,2], w \in [0,2-\kappa], \epsilon\in(0,\mathcal{O}(1)]$, the even polynomial $p_{\text{rect},w,\kappa,n}(x)=\frac{1}{2}\left( p_{\text{sgn},\kappa,(w+\kappa)/2,n+1}(x)+ p_{\text{sgn},\kappa,(w+\kappa)/2,n+1}(-x)\right)$ of even degree $n\mathcal{O}(\frac{1}{\kappa}\log{(1/\epsilon)})$ satisfies
\begin{align}
\label{Eq:p_tilde_rect}
\epsilon_{\text{rect},w,\kappa,n}&
=\max_{|x| \in [0,w/2]\cup[w/2+\kappa,1]} |p_{\text{rect},w,\kappa,n}(x)-\text{rect}(x/w)| \le \epsilon_{\text{sgn},\kappa,\delta,n}\le \epsilon.
\end{align}
\end{corollary}
\begin{proof}
This follows from the construction of a rectangular function with two sign functions in Lem.~\ref{Lem:Entire_Rect}.
\end{proof}
\begin{corollary}[Polynomial approximation to the truncated linear function $f_{\text{lin},\Gamma}(x)$]
\label{Lem:Polynomial_Truncated_Linear}
$\forall\;\Gamma \in ( 0,1/2],\epsilon\in(0,\mathcal{O}(\Gamma)]$, the odd polynomial $p_{\text{lin},\Gamma,n}(x)=\frac{x}{2\Gamma}p_{\text{rect},2\Gamma,2\Gamma,n-1}(x)$ of odd $n=\mathcal{O}(\frac{1}{\Gamma}\log{(1/\epsilon)})$ satisfies
\begin{align}
\label{Eq:p_tilde_trunc_linear}
\epsilon_{\text{lin},\Gamma,n}&
=\max_{|x| \in [0,\Gamma]} \frac{2\Gamma}{|x|}\left|p_{\text{lin},\Gamma,n}(x)-\frac{x}{2\Gamma}\right| \le \epsilon_{\text{rect},2\Gamma,2\Gamma,n-1}\le \epsilon.
\end{align}
\end{corollary}
\begin{proof}
This follows from multiplying a rectangular function with a linear function in Lem.~\ref{Lem:Entire_Linear}. One subtlety arises here: The error of $p_{\text{lin},\Gamma,n}$ is bounded by $\epsilon_{\text{rect},2\Gamma,2\Gamma,n-1}$ in the domain $|x|\in[3\Gamma,1]$. Thus  multiplying by $\frac{x}{2\Gamma}$ increases this error to at most $\frac{\epsilon_{\text{rect},2\Gamma,2\Gamma,n-1}}{2\Gamma}$. However, the quantum signal processing conditions in Thm.~\ref{Lem:AchievableD} require all polynomials to be bounded by $1$. This implicitly constrains us to choose $n$ such that $\epsilon_{\text{rect},2\Gamma,2\Gamma,n-1}\le 2 \Gamma$ is also satisfied.
\end{proof}

In all the above cases, the entire functions that are being approximated are bounded by $1$. When the approximation error is $\epsilon$, the resulting polynomial is then bounded by $1+\epsilon$. In such an event, we simply rescale these polynomials by a factor $\frac{1}{1+\epsilon}$. At worst, this only doubles the error of the approximation. We also emphasize that our proposed sequence of polynomial transformations serve primarily to prove their asymptotic scaling. In practice, close-to-optimal constant factors in the degree of these polynomials can be obtained by a direct Chebyshev truncation of the entire functions.

\section{Polynomials for Low-Energy Uniform Spectral Amplification}
\label{Sec:Polynomials_Low_energy}
The proof of Thm.~\ref{Thm:Ham_Encoding_Uniform_Amplification} requires a polynomial approximation $p_{\text{gap},\Delta,n}(x)$ Lem.~\ref{Lem.Polynomial_gapped_linear} to the truncated linear function
\begin{align}
\label{Eq:Linear_target_function}
f_{\text{gap},\Delta}(x)=
\begin{cases}
\frac{x+1-\Delta}{\Delta}, & x \in [-1, -1+\Delta], \\
\in[-1,1], & \text{otherwise}.
\end{cases}
\end{align}
Our strategy is to construct an entire function $f_{\text{gap},\Delta,\epsilon}$ that approximates $f_{\text{gap},\Delta}$ with error $\epsilon$ over the domain of interest. Entire functions are desirable as they are analytic on the entire complex plane. This implies that truncating their expansion $f_{\text{gap},\Delta,\epsilon}(x)=\sum_{j=0}^\infty a_j T_j(x)$ in the Chebyshev basis produces polynomials with a uniform approximation error that scales almost optimally with the degree $n$~\cite{Trefethen2013approximation}. We build $f_{\text{gap},\Delta,\epsilon}$ by using the entire approximation to the sign function $\text{sgn}(x)$ in Lem.~\ref{Lem:Entire_Sgn} of Appendix~\ref{Sec:Polynomials_Amplitude_Multiplication} and some intermediate results on the error function $\text{erf}(x)=\frac{2}{\sqrt{\pi}}\int_0^x e^{-y^2}dy$.
\begin{lemma}[Entire approximation to the gapped linear function $f_{\text{gap},\Delta}(x)$]
\label{Lem.Entire_gapped_lin}
$\forall\;\Delta\in[0,1/2], x\in[-1,\infty], \epsilon\in(0,\sqrt{\frac{1}{2e\pi}}]$. Then the function $f_{\text{gap},\Delta,\epsilon}(x)$ satisfies 
\begin{align}
f_{\text{gap},\Delta,\epsilon}(x)&=\frac{x+1-\Delta}{\Delta}\frac{1-f_{\text{sgn},\Delta,2\epsilon}(x+1-3\Delta/2)}{2},
\\\nonumber
\epsilon &\ge \max_{ x \in [-1,-1+\Delta]} \left|f_{\text{gap},\Delta,\epsilon}(x)-\frac{x+1-\Delta}{\Delta}\right|,
\\\nonumber
0&\le \max_{ x \in [-1+\Delta,\infty]}f_{\text{gap},\Delta,\epsilon}(x)\le 1,
\\\nonumber
\epsilon/10 &\ge \max_{ x \in [1-\Delta,1]} |f_{\text{gap},\Delta,\epsilon}(x)|.
\end{align}
\end{lemma}
\begin{proof}
Let us derive bounds on the following regions:\\
$x\in[-1,-1+\Delta]$: From Lem.~\ref{Lem:Entire_Sgn}, $|\frac{1-f_{\text{sgn},\Delta,2\epsilon}(x+1-3\Delta/2))}{2}-1|\le \epsilon$ approximates the function $1$ with error $\epsilon$. By multiplying both sides with  $\frac{x+1-\Delta}{\Delta}$, $|f_{\text{gap},\Delta,\epsilon}(x)-\frac{x+1-\Delta}{\Delta}|\le \frac{x+1-\Delta}{\Delta}\epsilon \le \epsilon$.
\\
$x\in[-1+\Delta,-1+3\Delta/2]$: From Lem.~\ref{Lem:Entire_Sgn} $|\frac{1-f_{\text{sgn},\Delta,2\epsilon}}{2}|\in [0,1/2]$. In this region, $\frac{x+1-\Delta}{\Delta}\in[0,1/2]$. Thus by multiplying, $f_{\text{gap},\Delta,\epsilon}(x)\in[0,1/2]$.
\\
$x\in [-1+3\Delta/2,1-\Delta]$: From the upper bound $\text{erfc}(x)\le e^{-x^2}$, $f_{\text{gap},\Delta,\epsilon}(x)\le \frac{x+1-\Delta}{2\Delta}e^{-k^2(x+1-3\Delta/2)^2}$, where $k=\frac{\sqrt{2}}{\Delta}\log^{1/2}{(\frac{1}{2\pi\epsilon^2})}$. The worst case occurs when $k$ is smallest hence $\epsilon=\sqrt{\frac{1}{2e\pi}}$ is largest. Thus the upper bound is maximized with value $\frac{1+\sqrt{5}}{4}e^{(\sqrt{5}-3)/4}\le 0.7$ at $x=-1+\frac{1}{4}(5+\sqrt{5})\Delta < -1+2\Delta$.
\\
$x\in [1-\Delta,\infty)$: The upper bound obtained for $x\in[-1+3\Delta/2,1-\Delta]$ still applies here and is monotonically decreasing with $x$. Thus it is maximized when $\Delta=1/2$ is largest and at $x=1-\Delta$. With this upper bound, $f_{\text{gap},1/2,\epsilon}(1/2)\le 2e^{-9k^2/16}< 32\sqrt{\frac{2\pi^9}{\epsilon^9}}< \frac{ \epsilon}{10}$ by substituting $k$ and then using the fact $\epsilon \le \sqrt{\frac{1}{2e\pi}}$.
\\
$x\in [-1+\Delta,\infty]$: $\frac{x+1-\Delta}{\Delta}$ and $\frac{1-f_{\text{sgn},\Delta,2\epsilon}(x+1-3\Delta/2)}{2}$ are both positive, thus $f_{\text{gap},\Delta,\epsilon}(x)$ is positive.
\end{proof}
We now construct a degree $n$ polynomial approximation to $f_{\text{gap},\Delta}(x)$.
\begin{lemma}[Polynomial approximation to the gapped linear function $f_{\text{gap},\Delta}(x)$]
\label{Lem.Polynomial_gapped_linear}
$\forall\;\epsilon \le \mathcal{O}(1)$, there exists an odd polynomial $p_{\text{gap},\Delta,n}$ of degree $n=\mathcal{O}(\Delta^{-1/2}\log^{3/2}{(1/(\Delta\epsilon))})$ such that 
\begin{align}
\max_{ x \in [-1,-1+\Delta]} \left|p_{\text{gap},\Delta,n}(x)-\frac{x+1-\Delta}{\Delta}\right|\le \epsilon
\quad \text{and}\quad 
\max_{ x \in [-1,1]} \left|p_{\text{gap},\Delta,n}(x)\right|\le 1.
\end{align}
\end{lemma}
\begin{proof}
Let us expand $f_{\text{gap},\Delta,\epsilon_1}(x)=\sum_{j=0}a_j T_j(x)$ in the Chebyshev basis. Then the truncation error of $p_{n}(x)=\sum^n_{j=0}a_j T_j(x)$ has a well-known upper bound from Thm.~8.2 of~\cite{Trefethen2013approximation}:
\begin{align}
\label{Eq:Polynomial_Gapped_linear_error}
\max_{x\in[-1,1]}|p_{n}(x)-f_{\text{gap},\Delta,\epsilon_1}(x)|\le \epsilon_2 = \frac{2M \rho^{-n}}{\rho-1}, \quad M = \max_{z\in E_\rho}|f_{\text{gap},\Delta,\epsilon_1}(z)|,
\end{align}
for any elliptical radius $\rho>1$, where $E_\rho=\{z:z=\frac{1}{2}(e^{i\theta}+\rho^{-1} e^{-i\theta}), \theta\in[0,2\pi)\}$ is the Bernstein ellipse. We will need an upper bound on $|\text{erf}(r e^{i\phi})|$ for $r\ge 0, \phi\in[0,2\pi)$:
\begin{align}
\label{Eq:Erfc_UpperBound}
|\text{erf}(r e^{i\phi})| 
&= \frac{2}{\sqrt{\pi}}\left|\int_0^r e^{- r^2 e^{2i\phi}} dr\right| 
 = \frac{2}{\sqrt{\pi}}\left|\int_0^r e^{- r^2 \cos{(2\phi)}}e^{- i r^2 \sin{(2\phi)}} dr\right| \\ \nonumber
& \le \frac{2}{\sqrt{\pi}}\left|\int_0^r e^{-r^2 \cos{(2\phi)}} dr\right| 
= \frac{2r}{\sqrt{\pi}}\max\{1,e^{-r^2 \cos{(2\phi)}}\} = \frac{2r}{\sqrt{\pi}}\max\{1,e^{\text{Re}(-(re^{i\phi})^2)}\}.
\end{align}
We also need the upper bounds $|z|^2=\frac{1}{4}\left(\rho^{2}+\rho^2{2}+2\cos{(2\theta)}\right)\le \rho^2$. 
Let $k=\frac{\sqrt{2}}{\Delta}\log^{1/2}{(\frac{1}{2\pi\epsilon_1^2})}$, $|k(z+1-3\Delta/2)|\le k(|z|+1+3\Delta/2)\le k(\rho+1+3\Delta/2)$. Then
\begin{align}
M &= \max_{z\in E_\rho}\left|\frac{z+1-\Delta}{\Delta}\frac{1-\text{erf}{(k(z+1-3\Delta/2)}}{2}\right| 
\le 
\max_{z\in E_\rho}\frac{|z|+1+\Delta}{2\Delta}\left(1+|\text{erf}{(k(z+1-3\Delta/2)}|\right)
\\\nonumber
&\le \mathcal{O}(\text{poly}(\rho,\Delta^{-1}))\max_{z\in E_\rho}\left(1+\frac{2|k(z+1-3\Delta/2)|}{\sqrt{\pi}}(1+e^{\text{Re}(-(k(z+1-3\Delta/2)^2)})\right)
\\\nonumber
&\le \mathcal{O}(\text{poly}(\rho,\Delta^{-1}))\max_{z\in E_\rho}e^{\text{Re}(-(k(z+1-3\Delta/2)^2)}.
\end{align}
By taking derivatives with respect to $\theta$, the maximum value of the exponent $\alpha=\max_{\theta\in[0,2\pi)}\text{Re}(-(k(z+1-3\Delta/2)^2) = \frac{k^2(\rho^2-1)(2-(2-3\Delta)^2\rho^2+2\rho^4)}{8\rho^2(1+\rho^4)}$. Let us choose $\rho = e^{a}$, where $a=\mathcal{O}(1/\sqrt{k^2\Delta})$. Then $\alpha = \mathcal{O}(1)$. Substituting the value of $k$, we have  $a=\mathcal{O}(\sqrt{\Delta/\log{(1/\epsilon_1)}})$, and $M = \mathcal{O}\left(\text{poly}(\Delta^{-1})\right)$. Thus from Eq.~\ref{Eq:Polynomial_Gapped_linear_error},
\begin{align}
\epsilon_2=\mathcal{O}\left(\text{poly}(\Delta^{-1})e^{-n \sqrt{\Delta/\log{(1/\epsilon_1)}}}\right) 
\Rightarrow 
n = 
\mathcal{O}\left(\Delta^{-1/2}\log^{3/2}{\left(\frac{1}{\max\{\Delta\epsilon_1,\epsilon_2\}}\right)}\right),
\end{align}
where the last equation applies $\log{(\text{poly}(\Delta^{-1})/\epsilon)}=\mathcal{O}(\log(\frac{1}{\Delta\epsilon}))$. Thus the total approximation error is $\max_{x\in[-1-1+\Delta]}|p_{n}(x)-\frac{x+1-\Delta}{\Delta}|\le \epsilon_1+\epsilon_2$. Let $p_{\text{gap,sym},\Delta,n}(x)=\frac{1}{2}(p_{n}(x)-p_{n}(-x))$ be the odd component of $p_{n}(x)$. Using the bounds of Lem.~\ref{Lem.Entire_gapped_lin}, this increases the error in $x\in[-1,-1+\Delta]$ to at most $\frac{11}{10}(\epsilon_1+\epsilon_2)$. By subtracting these bounds, we also have $\max_{x\in[-1,1]}|p_{\text{gap,sym},\Delta,n}(x)|\le 1+\frac{11}{10}(\epsilon_1+\epsilon_2)$. Thus we rescale this to obtain $p_{\text{gap},\Delta,n}(x)=\frac{p_{\text{gap,sym},\Delta,n}(x)}{1+\frac{11}{10}(\epsilon_1+\epsilon_2)}$. Using $\max_{x\in[0,\infty]}|\frac{1}{1+x}-1|\le x$, This increases the error by at most a constant factor $\max_{x\in[-1-1+\Delta]}|p_{\text{gap},\Delta,n}(x)-\frac{x+1-\Delta}{\Delta}|=\mathcal{O}(\epsilon_1+\epsilon_2)$, so choose $\epsilon_1=\epsilon_2=\mathcal{O}(\epsilon)$.
\end{proof}

\bibliography{AmpHamiltonianSimulation}

\end{document}